\documentclass[12pt,english,us]{article}
\UseRawInputEncoding

\usepackage{graphicx}   
\usepackage{amsmath} 
\usepackage{amsfonts} 
\usepackage{amssymb}  
\usepackage{bm}
\usepackage{bbm}
\usepackage{dcolumn}
\usepackage{graphics}
\usepackage{dsfont}
\usepackage{float}
\usepackage{hyperref}
\usepackage{multirow}
\usepackage[bottom]{footmisc}
\usepackage{comment}
\hypersetup{
    colorlinks = true,
    allcolors = [rgb]{0,0.29,0.6}
    }
\usepackage{pdflscape}
\usepackage{enumitem} 
\usepackage{epstopdf}
\usepackage{caption}
 \usepackage{booktabs}
\usepackage{subcaption}
\usepackage{marvosym}
\usepackage{placeins}
\usepackage{array}
\newcolumntype{H}{>{\setbox0=\hbox\bgroup}c<{\egroup}@{}}
\newcolumntype{Z}{>{\setbox0=\hbox\bgroup}c<{\egroup}@{\hspace*{-\tabcolsep}}}

\newtheorem{theorem}{Theorem}

\newtheorem{proposition}{Proposition}
\newtheorem{corollary}{Corollary}

\newtheorem{assumption}{Assumption}
\newtheorem{proof}{Proof}

\usepackage[onehalfspacing]{setspace}  
\usepackage{threeparttable}
\usepackage{tikz}
\usepackage{tikz-qtree}

\usepackage[margin=0.5in]{geometry}

\usepackage{titletoc}
\usepackage{bbold}

\usepackage{natbib}
\usepackage{tabularx}
\usepackage{rotating}
\usepackage{comment}
\newcommand{\sym}[1]{\rlap{$#1$}} 

\usepackage{etoolbox}
\makeatletter
\patchcmd{\NAT@test}{\else \NAT@nm}{\else \NAT@nmfmt{\NAT@nm}}{}{}
\DeclareRobustCommand\citepos
  {\begingroup
   \let\NAT@nmfmt\NAT@posfmt
   \NAT@swafalse\let\NAT@ctype\z@\NAT@partrue
   \@ifstar{\NAT@fulltrue\NAT@citetp}{\NAT@fullfalse\NAT@citetp}}
\let\NAT@orig@nmfmt\NAT@nmfmt
\def\NAT@posfmt#1{\NAT@orig@nmfmt{#1's}}
\makeatother

\newcommand\independent{\protect\mathpalette{\protect\independenT}{\perp}}
\def\independenT#1#2{\mathrel{\rlap{$#1#2$}\mkern2mu{#1#2}}}

\begin{document}

\newgeometry{top=1cm}

\title{The Chained Difference-in-Differences}
\author{Christophe \textsc{Bell\'ego}\thanks{CREST (UMR 9194), Institut Polytechnique de Paris, 5 Avenue Henry Le Chatelier, 91120 Palaiseau,
France (e-mail: christophe.bellego@ensae.fr).}
  \and
  David \textsc{Benatia}\thanks{HEC Montr\'eal, D\'epartement d'\'Economie Appliqu\'ee, 3000 Chemin de la C\^ote-Sainte-Catherine, Montr\'eal, QC H3T 2A7, Canada (Corresponding author, e-mail: david.benatia@hec.ca).}
\and
Vincent \textsc{Dortet-Bernadet}\thanks{Direction G\'{e}n\'{e}rale des Entreprises (DGE),  
Minist\`{e}re Fran\c{c}ais de l'\'{E}conomie et des Finances, 139 rue de Bercy, 75012 Paris, France. (e-mail: vincent.dortet-bernadet@finances.gouv.fr).}}

\date{\today \\ 
}
\maketitle

\vspace{-2em}

\begin{abstract}
This paper studies the identification, estimation, and inference of long-term (binary) treatment effect parameters when balanced panel data is not available, or consists of only a subset of the available data. We develop a new estimator: the chained difference-in-differences, which leverages the overlapping structure of many unbalanced panel data sets. This approach consists in  aggregating a collection of short-term treatment effects estimated on multiple incomplete panels. Our estimator accommodates (1) multiple time periods, (2) variation in treatment timing, (3) treatment effect heterogeneity, (4) general missing data patterns, and (5)  sample selection on observables. We establish the asymptotic properties of the proposed estimator and discuss identification and efficiency gains in comparison to existing methods. Finally, we illustrate its relevance through (i) numerical simulations, and (ii) an application about the effects of an innovation policy in France.
\text{ }  \\


\end{abstract}

\thispagestyle{empty}

\newgeometry{top=4cm}

\normalsize

%


\parskip=5pt
\setlength{\parindent}{0pt}
\maketitle

\section{Introduction}

Many public policies take years before having effects on their targeted outcomes. For instance, most innovation policies have long run objectives such as making new discoveries, enhancing knowledge, and increasing technology development, but induce only little innovation in the short-run \citep{bakker_2013,oconnor_etal_2013,gross_2018}. Measuring these long-term effects is however difficult for two principal reasons. First, identification of treatment effects from observational data raises significant challenges, whereas randomized controlled experiments are often too costly, or raise ethical concerns. 
Second, treatment effects are frequently estimated from panel survey data, where the subjects of interest (e.g. individuals or firms) are not consistently observed over the entire time frame. This problem can be caused by attrition \citep{hausman1979attrition}, if individuals drop out of the sample, or because of the survey design itself, if individuals are frequently replaced to prevent attrition or over-soliciting respondents  \citep{nijman1991efficiency}. 


In this paper, we study the identification, estimation, and inference of long-term treatment effects in settings where balanced panel data is not available, or consists of only a subset of the available data. We develop a new estimator, the chained difference-in-differences (DiD), which leverages the overlapping structure of many unbalanced panel data sets \citep{baltagi_song_2006}. The intuition behind the chained DiD estimator is simple. Letting $Y_t$ be an outcome of interest observed in three distinct time periods $t_0<t_1<t_2$, the long difference  $Y_{t_2}-Y_{t_0}$ can always be decomposed into a sum of short differences $(Y_{t_2}-Y_{t_1})+(Y_{t_1}-Y_{t_0})$. Our estimator generalizes this simple idea by optimally aggregating  short-term treatment effect parameters, or ``chain links'',  obtained from (possibly many) overlapping incomplete panels. For instance, one subsample of individuals might be used to estimate the DiD from $t_0$ to $t_1$ while another would be used to estimate the DiD from $t_1$ to $t_2$. The sum of both effects identifies the long-term DiD from $t_0$ to $t_2$, which may not be feasible or may suffer from efficiency losses if a large part of the sample is discarded by using only a balanced subsample. Our estimator is designed for multiple time periods, it accommodates variations in treatment timing, treatment effect heterogeneity, general missing data patterns, sample selection on observables, and it may deliver substantial efficiency gains compared to estimators using only subsets of the data or treating the unbalanced panel as repeated cross-sectional data.

Building upon  the influential work of \cite{callaway2021difference}, we show that our estimator is consistent, asymptotically normal, and computationally simple. Their  multiplier bootstrap and pre-trend tests remain asymptotically valid in our setting. We identify three main advantages of our approach: (1) it does not require having a balanced panel subsample as it is the case with a standard DiD; (2) identifying assumption allows for sample selection on time-persistent unobservable (or latent) factors, unlike the cross-section DiD which treats the sample as repeated cross-sectional data; and (3) it may also deliver efficiency gains compared to existing methods, notably when the outcome variable is highly time-persistent.

The proposed approach is especially relevant with regard to the significant interest for long-term evaluations of interventions, with many applications related to education and labor economics \citep{angrist2006long,kahn_2010,oreopoulos_etal_2012,autor_houseman_2010,garcia_Perez_etal_2020,lechner2011long,mroz_savage_2006,stevens_1997}. 
The estimation of treatment effects generally consists in performing a DiD focusing only on units that are observed over the entire time frame: the balanced subsample. This approach, hereafter referred to as the long DiD,  allows getting rid of both time- and individual-specific unobservable heterogeneity, but may also involve discarding many individuals with missing observations when estimating long-term effects. The evaluation of long-term effects is sometimes difficult, if not infeasible, because of such missing data problems. 
Missing observations in panel data sets may exist for two principal reasons: (1) by design or (2) because of attrition.

First, the design of rotating panel surveys attempts to alleviate the burden of administering and responding to statistical surveys, and to prevent attrition, by replacing subjects regularly \citep{heshmati1998efficiency}. The survey is administered to a cohort of subjects in a limited number of periods. The cohort is then replaced by another cohort randomly drawn from the population of interest. The resulting data set is hence composed of a collection of incomplete panels. The data is said to have an overlapping structure if there is at least one period in which two separate cohorts are administered the survey. 

The most famous example of rotating panel survey is the Current Population Survey (CPS), where each cohort is interviewed for a total of 8 months (over a period of 16 months), and part of the sample is replaced each month by a new subsample.\footnote{The design is detailed at \url{https://www.bls.gov/opub/hom/cps/design.htm}.} This survey is one of the most widely used data sources in economic and social research. It has been used or cited in 1000+ articles between 2000 and 2022, including publications in top journals such as the American Economic Review, Journal of Political Economy, Quarterly Journal of Economics, American Sociological Review, and Demography.\footnote{This figure is based on a Google Scholar search conducted on December 5, 2022.} Other rotating panel surveys include, but are not limited to, the Medical Expenditure Panel Survey, the General Social Survey since 2006, the Consumer Expenditure Survey \citep{blundell_etal_2008} in the U.S., the Labour Force Survey in the U.K., and the \emph{Enqu\^{e}te Emploi} (Labour Force Survey) and \emph{Enqu\^{e}te sur les Moyens Consacr\'{e}s \`{a} la R\&D} (R\&D Survey) in France.  

Despite their prevalence, the econometrics literature on rotating panels is almost non-existent \citep{baltagi_song_2006}. \cite{heshmati1998efficiency} uses rotating panels for production function estimation. \cite{nijman1991efficiency} study the optimal choice of the rotation period for estimating a linear combination of period means. Likewise, our estimator consists in a linear combination of parameters corresponding to the long-term average treatment effect parameter. To the best of our knowledge, our paper is the first to study identification, estimation, and inference of treatment effects in this context. 
 
Second, individuals may also drop out of the surveys and cause attrition, which can be particularly severe for long panels. This attrition raises some concerns because it is associated with selection. Attrition can be due to either ``ignorable'' or ``non-ignorable'' selection rules \citep{verbeek1996incomplete}. Ignorable attrition implies that missing data occurs completely at random, and as such focusing on a balanced panel subsample does not threaten identification. \cite{verbeek1992testing}, among others, propose a test of the ignorability assumption. Non-ignorable attrition means that missing data is related to either observable or unobservable factors.  \citet{hirano2001combining} provide important identification results for  the additive non-ignorable class of attrition models, which nest several well-known methods \citep{hausman1979attrition,little1989analysis}. These results have been extended to multi-periods panels by \cite{hoonhout2019nonignorable} and apply to our setting, in particular we consider their Sequential Missing At Random assumption.  \cite{bhattacharya2008inference} studies the properties of a sieves-based semi-parametric estimator of the attrition function initially proposed by \cite{hirano2001combining}. These methods are often referred to as models of selection on unobservables because the probability of attrition is allowed to depend on variables that are not observed when an individual drops out, but not on unobservable error terms \citep{moffit1999sample}. Inverse propensity weighting is the most popular method for addressing non-ignorable attrition caused by observable factors, even beyond the estimation of average treatment effects \citep{chaudhuri2018indirect}. There also exist methods using instrumental variables to address attrition due to latent factors  \citep{frolich2014treatment}. Other approaches focus on particular structures of attrition, like monotonically missing data \citep{chaudhuri2020efficiency,barnwell2021note}.
 
Attrition is generally addressed before estimating treatment effects with the long DiD, either by reweighting the observations in the balanced subsample with inverse propensity scores \citep{hirano2003efficient}, or by imputing the missing observations \citep{hirano1998combining}. 
However, discarding a possibly large proportion of the data by focusing on the balanced subsample may lead to significant efficiency losses \citep{baltagi_song_2006,chaudhuri2020efficiency,barnwell2021note}, or may lead to discarding the complete dataset, as illustrated in our application using the French R\&D survey.  If there are too few individuals observed over the entire time horizon, the data is typically treated as repeated cross-sections and treatment effects are estimated with the cross-section DiD \citep{abadie_2005,callaway2021difference}.\footnote{This approach consists in taking the difference of differences of averages, where different sets of units are used to compute each of the four averages.} 

The identification of the average treatment effect using the cross-section DiD requires not only a parallel trends assumption on the population, but also that the sampling process be (conditionally) independent of the \emph{levels} of the outcome variable. This assumption is violated if, for instance, treated units with larger unobserved individual shocks are relatively more likely to be sampled in later periods. In such a case, the treatment groups observed early on will differ in unobservable ways from those observed later on. An identification problem arises as soon as the sampling process does not affect the observed control groups in the exact same way. 
Instead, our approach requires that the sampling process be (conditionally) independent of the \emph{trend} of the outcomes, because the chained DiD allows eliminating individual heterogeneity like the long DiD before taking expectations. Therefore, the chained DiD is robust to some forms of attrition caused by unobservable heterogeneity, unlike the cross-section DiD. Remark that these two identifying assumptions are non-nested in general. Note further that the attrition models discussed earlier can also be used as a first-step in our framework. Although we do not correct for the efficiency losses of using such first-step plug-in estimates in our paper, we identify some suitable options to do so \citep{frazier2017efficient}.

The chained DiD rests on a parallel trends assumption, conditioned on being sampled.  The parallel trends assumption posits that in the absence of treatment, the trends in potential outcomes for the treated and non-treated groups \emph{in the population} would have been similar. It offers a weaker alternative to the assumption that the potential outcomes for both treated and untreated groups \emph{in the population} would have been identical. Extending this logic to unbalanced panel data, our additional assumption pertains to the composition of the sample. If the population is composed of different unobservable types that correlate with treatment assignment and the likelihood of being sampled, then a selection bias may exist. To address this problem, we assume that all unobservable types exhibit similar \emph{trends} in outcomes conditional on their treatment status (and possibly other covariates) so that the parallel trend assumption also holds in the sample. 

Recent developments in the literature about treatments effects with multiple periods and treatment heterogeneity have revealed that two-way fixed-effects estimators fail to identify the treatment effect parameters of interest in many contexts \citep{chaisemartin_haultfoeuille_2020,goodman_bacon_2018,borusiak_jaravel_2017,sun_abraham_2020}.\footnote{\citet{de2022survey} provide an excellent survey of the literature, which we briefly summarize in Section \ref{sec:framework}.} Our paper builds upon \cite{callaway2021difference} which address this issue by generalizing the approach developed in \cite{abadie_2005} to multiple periods and varying treatment timing.  We contribute further to this literature by extending this framework to settings with incomplete panel data. An alternative approach would have been to adapt the general framework developed by \cite{chaisemartin_haultfoeuille_2020} to our setting, or the other related papers focused on staggered adoption designs and event studies with multiple periods. However, the chained DiD fits very well within the framework developed by \citet{callaway2021difference} which focuses on estimating all group-time average treatment effects before aggregating all those parameters into summary parameters of interest. Likewise, our approach focus on (even smaller) building blocks: one-period-difference group-time average treatment effects, which measure the increase of average treatment effect of group $g$ from period $t-1$ to period $t$. In the general case, these blocks correspond to k-period-difference group-time average treatment effects. We also discuss regression alternatives inspired by \citet{borusiak_jaravel_2017} and \citet{wooldridge2021two} in the paper.

We illustrate the performance of the chained DiD in two ways. First, we use simulations to compare the long DiD, chained DiD and cross-section DiD in terms of bias and variance under several data generating processes (DGP). Our simulations include a stratified panel data set, composed of a balanced panel and a rotating panel. Second, we study the long-term employment effects of a large-scale innovation policy in France giving grants to collaborative R\&D projects. Technical progress and innovation, stimulated by R\&D activities, are key to economic growth \citep{scherer_1982,aghion_howitt_1998,griffith_etal_2004} but firms tend to invest too little because of the public good nature of innovations \citep{arrow_1962,nelson_1959}. Therefore, measuring the long term effects of subsidized R\&D is important to inform policymakers and improve future policies. 


This application is especially relevant because the French R\&D survey, which contains firm-level data on R\&D activities, consists of rotating panels but does not include a balanced subsample except for the largest firms. In addition, it is possible to fully observe a limited number of variables for all firms close to those provided by the R\&D survey by using administrative data. We are thus able to compare the results of each of the three estimators by focusing on two of those variables: total employment and highly qualified workforce. The long DiD is applied to the complete data and serves as the benchmark estimates. The chained DiD and cross-section DiD estimators are applied to both data sets, and to an ``artificial'' unbalanced panel which is generated by discarding all observations from the complete administrative data set which are missing in the R\&D survey. This application is somehow comparable to a simulation exercise but uses real data. 

Our results show that the policy had a positive effect on employment for firms that received a grant to participate to a collaborative R\&D project. We find that the estimates as well as their standard errors obtained with the chained DiD estimator are close to those obtained with the long DiD. In contrast, the cross-section estimator delivers biased estimates which also lack sufficient precision to detect any statistically significant effect associated with the policy. 

The remainder of the paper is organized as follows. Section \ref{sec:framework} presents our methodology and asymptotic results. Numerical simulations are in Section \ref{sec:simu}. Section \ref{sec:application} contains the application to R\&D policy. Section \ref{sec:conclusion} concludes the paper.

\section{Identification, Estimation and Inference}\label{sec:framework}

\subsection{Basic Framework}
We first present the main insights in a simple framework. The main notation is as follows. There are $\mathcal{T}$ periods and each particular time period is denoted by $t = 1,...,\mathcal{T}$. In a standard DiD setup, $\mathcal{T}=2$, no one is treated in $t=1$, and all treatments take place in $t=2=\mathcal{T}$. To gain intuition, we first focus on the case where all treatments take place in $t=2$ but assume $\mathcal{T} > 2$. 

Define $G$ to be a binary variable equal to one if an individual is in the treatment group, and $C=1-G$ as a binary variable equal to one for individuals in the control group. Also, define $D_t$ to be a binary variable equal to one if an individual is treated in $t$, and equal to zero otherwise. Let $Y_t(0)$ denote the potential outcome at time $t$ without the treatment and $Y_t(1)$ denote its counterpart with the treatment. The realized outcome in each period can be written as $Y_t = D_t Y_t(1) + (1-D_t)Y_t(0)$. 

We  focus on the \emph{long-term average treatment effect on the treated} which corresponds to the average treatment effect in period $t>2$ on individuals in the treatment group, hence first treated in period 2. It is formally defined by
\begin{equation}
ATT(t) = E\left[Y_t(1) - Y_t(0) | G=1 \right].
\end{equation}
The identification of $ATT(t)$ with panel data has attracted much research attention. In this paper, we are mainly interested in a case where balanced panel data is not available.

\paragraph{The missing data pattern.}

We assume that a new random sample of $n_t$ individuals is drawn at each period $t$ to replace a subsample of previously observed individuals. This replacement can occur due to attrition or by design. For the moment, we assume that subsamples may differ in size $n_t$ but individuals are only observed for two consecutive periods as in many rotating panel design.\footnote{In practice, rotating panels may involve subgroups sampled over a different number of consecutive periods. The result of this paper continues to apply, so the same estimator and inference procedure can be used even with a more sophisticated sampling process as presented in Section \ref{sec:moregeneralsampling}.} Therefore, one cannot observe the entire path $\{Y_1,Y_2,...,Y_{\mathcal{T}}\}$ for any individual. The key feature of our approach is that there is some overlap across subsamples, that is there are at least two different subsamples observed in each period $1<t< \mathcal{T}$.

This structure is in stark contrast with the literature about attrition in panel data, which typically assumes that there always exists a balanced subsample where individuals are observed throughout the entire time frame \citep{hirano2001combining,hoonhout2019nonignorable}. Although there is no need for such a balanced subsample here, we extend our method to more general settings in Section \ref{sec:moregeneralsampling}. This extension is illustrated in Sections \ref{sec:simu} and \ref{sec:application}.

To characterize the sampling process, we define $S_t$ to be a binary variable equal to one for individuals observed at $t$ and zero otherwise, and use $S_{t,t+1}$ to denote $S_tS_{t+1}$ which indicates if an individual is observed at both $t$ and $t+1$. The missing data pattern is summarized in Table \ref{table:table1} for $\mathcal{T}=3$. In the general framework, we will assume that treatments can also vary with some observable covariates $X$, that individuals can be observed in non-consecutive periods, and that sampling probabilities may also be functions of (possibly different) observable covariates or past outcomes.\footnote{In this paper, we do not consider cases where treatment status (or covariates) may also be missing, but relevant solutions already exist for difference-in-differences methods \citep{botosaru2018difference}.}

\begin{table}[H]
\centering
\caption{\label{table:table1}
Missing data pattern in a three-period panel data set}
{
\begin{tabular*}{\columnwidth}{@{\hspace{\tabcolsep}\extracolsep{\fill}}l*{8}{D{.}{.}{1}}} \toprule          & \multicolumn{3}{c}{Obs. Indicators}  &   \multicolumn{5}{c}{Variables}  \\ Sub-population  & S_1 & S_2 & S_3 & Y_1 & Y_2 & Y_3 & X & G \\ \hline  Incomplete Panel 1 & 1 & 1 & 0 & \times & \times & \cdot & \times & \times \\
Incomplete Panel 2 & 0 & 1 & 1 & \cdot & \times & \times & \times & \times \\\bottomrule \end{tabular*}
}
\end{table}

\paragraph{The long DiD.} The common approach to identify treatment effects is to consider the long difference-in-differences defined by
\begin{equation}\label{eq:ATT 2 basic}
\begin{aligned}
ATT(t)  = & E[Y_t(1)-Y_1(0)|G=1] - E[Y_t(0)-Y_1(0)|C=1],
\end{aligned}
\end{equation}
under the standard parallel trend assumption $E[Y_t(0)-Y_1(0)|G=1]=E[Y_t(0)-Y_1(0)|C=1]$. $ATT(t)$ corresponds to the long-term effect for $t>2$. Unfortunately, individuals are never observed more than two consecutive periods in this framework. Calculating the averages of $Y_{it}-Y_{i1}$ for $t>2$ for the treatment and control groups is hence infeasible.

\paragraph{The cross-section DiD.} If the panel consists of incomplete panel data, the identification of the parameter of interest can be achieved by assuming that the sampling process $S_{t}$ is independent of $(Y_{t},D_{1},D_{2},...,D_{\mathcal{T}})$ in addition to the parallel trend assumption \citep{abadie_2005}. In this case, $ATT(t)$ is identified by the ``cross-section DiD'' given by
\begin{equation*}
\begin{aligned}
ATT_{CS} & = \left(E[Y_t(1)|S_t G=1] - E[Y_1(0)|S_1 G=1]\right) - \left(E[Y_t(0)|S_t C=1] - E[Y_1(0)|S_1 C=1]\right).
\end{aligned}
\end{equation*}
The sampling assumption allows replacing the averages of differences,  $E[Y_t(0)-Y_1(0)|a=1]$ for $a \in\{G,C \}$, by the difference of averages, e.g. $E[Y_t(0)|S_ta=1]-E[Y_1(0)|S_1a=1]$  for $a \in\{G,C \}$. This approach does not eliminate the individual-specific unobservable heterogeneity.\footnote{This approach is readily implemented in the R/Stata package \textbf{did/csdid} and corresponds to the repeated cross sections version of \citet{callaway2021difference}.}

\paragraph{The chained DiD.} Our approach takes advantage of the overlapping panel structure. Remark that each term in \eqref{eq:ATT 2 basic} can be decomposed into 
\begin{equation}
\begin{aligned}
E\left[ Y_t(1) - Y_1(0)|a=1 \right]
& =  \sum_{\tau=1}^{t-1}  E\left[ Y_{\tau+1}(D_{\tau+1}) - Y_{\tau}(D_{\tau}) |a=1 \right] \\
\end{aligned}
\end{equation}
for $a \in\{G,C \}$. Thus, identification of $ATT(t)$ is obtained by summing the short-term DiD as in
\begin{equation*}
\begin{aligned}
& ATT_{CD}(t)  
= \sum_{\tau=1}^{t-1}   \left( E\left[ Y_{\tau+1}(D_{\tau+1}) - Y_{\tau}(D_{\tau}) |G=1 \right] - E\left[ Y_{\tau+1}(D_{\tau+1}) - Y_{\tau}(D_{\tau}) |C=1 \right] \right) \\
 & = \sum_{\tau=1}^{t-1}   \left( E\left[ Y_{\tau+1}(D_{\tau+1}) - Y_{\tau}(D_{\tau}) |S_{\tau,\tau+1}G =1  \right]   - E\left[ Y_{\tau+1}(D_{\tau+1}) - Y_{\tau}(D_{\tau}) |S_{\tau,\tau+1}C =1 \right] \right),
\end{aligned}
\end{equation*}
where the second equality holds under the assumption that the sampling process $S_{t,t+1}$ is  independent of $(Y_{t+1}-Y_{t},D_{1},D_{2},...,D_{\mathcal{T}})$. This approach not only allows eliminating the individual-specific heterogeneity, it also makes use of a different identifying assumption than the one for the cross-section DiD in this setting. Remark that these assumptions are generally non-nested.  However, replacement of individuals in panel survey data is often based on some variables observed before the replacement period but not on their evolution per se, as illustrated in our application.\footnote{Surveys are usually not specifically designed to use panel econometric methods.} Our application provides an illustration of that argument.

\paragraph{Simple model.} In order to gain intuition about the identification, estimation and inference of $ATT(t)$ in this context, we suppose that the potential outcome is generated by a components of variance process
\begin{equation}\label{eq: Yt components of variance}
Y_{it}(D_i) = \alpha_i +  \delta_t + \sum_{\tau=2}^t \beta_{\tau} D_{i\tau} + \varepsilon_{it},
\end{equation}
where $D_i = (D_{i1},D_{i2},...,D_{i\mathcal{T}})$ denotes the vector of treatment status, $ATT(t) = \sum_{\tau=2}^t \beta_{\tau}$ is the impact of the treatment evaluated at $t$. $\alpha_i$ is an individual-specific time-persistent unobservable component of variance $\sigma_{\alpha}^2$, $\delta_t$ is a time-specific unobservable component, and $\varepsilon_{it}$ is a mean-zero individual-transitory shock that is an auto-regressive  error process of order 1 represented by $\varepsilon_{it+1} = \rho\varepsilon_{it} + \eta_{it+1}$ with $\rho \in [0,1]$ and $\eta_{it+1}$ being a white noise of variance $\sigma_{\eta}^2$. Assume further that $D_{it} \independent (Y_{it}(0),Y_{it}(1))$, for all $t>2$. These assumptions imply $V(\varepsilon_{it}) = \sigma_{\varepsilon}^2 = \sigma_{\eta}^2/(1-\rho^2)$ if $0\leq\rho <1$, and $V(\varepsilon_{it}) = \sigma^2_{\varepsilon_1} + (t-1) \sigma^2_{\eta}$ if $\rho =1$.

\subsubsection{Identification}
If the sampling process is independent of all the components of $Y_{it}(D_i)$, then both $ATT_{CD}$ and $ATT_{CS}$ identify the parameter of interest. However, if the unobservable individual-specific component $\alpha_i$ is correlated with the sampling process, $ATT_{CS}$ admits the bias term
\begin{equation}\label{eq:bias crosssection 1}
\begin{aligned}
 \left(E[\alpha_i|S_t G=1] - E[\alpha_i|S_1 G=1]\right) - \left(E[\alpha_i|S_t C=1] - E[\alpha_i|S_1 C=1]\right).
\end{aligned}
\end{equation}
Furthermore, if the individual transitory shock  $\varepsilon_{it}$  is correlated with the sampling process, $ATT_{CS}$ may also admit the bias term 
\begin{equation}\label{eq:bias crosssection 2}
\begin{aligned}
 \left(E[\varepsilon_{it}|S_t G=1] - E[\varepsilon_{i1}|S_1 G=1]\right) - \left(E[\varepsilon_{it}|S_t C=1] - E[\varepsilon_{i1}|S_1 C=1]\right).
\end{aligned}
\end{equation}

These bias terms are non-zero if the compositions of the \emph{sampled} treatment and control groups evolve differently through time, due to either the time-persistent heterogeneity $\alpha_i$ or the time-varying error $\varepsilon_{it}$. 
The former situation happens, for instance, if individuals with larger $\alpha_i$ are more likely to be sampled in the control group in period 1 but become relatively more likely to be sampled in the treatment group in period $t$. 
The latter case occurs, for instance, if $\rho>0$ and untreated individuals  with larger $\varepsilon_{i\tau-1}$ are more likely to be sampled in period $\tau$. If  $\varepsilon_{i\tau-1}$ is highly persistent, the control group observed in later periods will be increasingly composed of individuals with larger past idiosyncratic shocks.  Addressing these biases can be difficult, if not impossible, due to the unobservability of these errors.

In contrast, $ATT_{CD}$ cancels out all $\alpha_i$'s, hence it is immune to biases caused by the time-persistent individual heterogeneity like \eqref{eq:bias crosssection 1}. However, the time-varying errors may not fully cancel out and a  bias term related to \eqref{eq:bias crosssection 2} could exist, as shown by 
\begin{equation}\label{eq:bias chainediff 1}
\begin{aligned}
 \sum_{\tau=2}^t\left( E[(\rho-1)\varepsilon_{i\tau-1} + \eta_{i\tau}|S_{\tau\tau-1} G=1] \right) - \left(E[(\rho-1)\varepsilon_{it-1} + \eta_{it}|S_{\tau\tau-1} C=1] \right),
\end{aligned}
\end{equation}
if $\rho<1$, and for instance, individuals with larger $\varepsilon_{i\tau-1}$ are more likely to be sampled in period $\tau$. This bias disappears if $\varepsilon_{it}$ follows a random walk ($\rho=1$).

Therefore, if the sampling process depends on unobservable components of the outcomes $Y_{it}$ but is not correlated with the first-differences in outcomes $Y_{it+1}-Y_{it}$ conditional on treatment status, then $ATT(t)$ is still identified by $ATT_{CD}(t)$  but not necessarily by $ATT_{CS}(t)$. This result implies that identification with $ATT_{CD}(t)$ is more robust to some potential dependence between the sampling process and unobservable time-persistent and time-varying factors. 
 \citet{ghanem2023selection} show that relying on parallel trends assumptions in difference-in-differences studies, with balanced panel data, imply restrictions on selection into treatment and the time-series properties of the outcomes similar to ours.\footnote{Specifically, they identify two key scenarios: (1) selection into treatment based on fixed effects, leading to a stationarity restriction; and (2) selection on time-varying pre-treatment factors, implying a martingale property.}  

\paragraph{The R\&D survey example.} To illustrate the practical significance of this discussion about identification, let us consider the example of R\&D subsidy programs and their impact on firms' labor forces, as explored in our application section. Consider that the firm population is divided into two distinct, unobservable (time-persistent) types: high-quality and low-quality research teams. The parallel trends assumption would require that, in the absence of subsidies, firms with high-quality and low-quality research teams would have followed similar trends in outcomes. 

The risk of selection bias emerges  if the allocation of subsidies is not random and correlates with  unobservable team quality. The long DiD inherently neutralizes this bias if the selection into treatment is solely linked to individual fixed-effects \citep{ghanem2023selection}. Nonetheless, this selection bias warrants further scrutiny if the propensity of high and low-quality teams to consistently participate in R\&D surveys varies, possibly due to 
 differences in resource availability, past experience with surveys, corporate culture, or differing levels of motivation across types. In this context, the cross-section DiD method becomes susceptible to selection bias. However, the chained DiD approach, which compares the changes in outcomes over multiple time periods, remains unaffected by this bias, as long as the variation in survey response rates does not correlate with how outcomes evolve conditional on treatment assignment.

\subsubsection{Estimation}

A direct estimation method for both $ATT_{CD}(t)$  and $ATT_{CS}(t)$ is to substitute expectations by their sample counterparts in each expression. These Horvitz-Thompson estimators can be written as the weighted averages 
\begin{equation}\label{eq:ATTCD_simple}
\begin{aligned}
\widehat{ATT_{CD}}(t) &  = 
 \frac{1}{n}\sum_{i=1}^{n}\sum_{\tau=1}^{t-1}  \left\{ \widehat{w^G}_{i\tau\tau+1}\left(y_{i\tau+1} - y_{i\tau}\right) -   \widehat{w^C}_{i\tau\tau+1}  \left(y_{i\tau+1} - y_{i\tau}\right) \right\},
\end{aligned}
\end{equation}

\begin{equation}\label{eq:ATTCS_simple}
\begin{aligned}
\widehat{ATT_{CS}}(t)    = &  \frac{1}{n}\sum_{i=1}^{n}\left\{ \left(\widehat{w^G}_{i t}y_{it} - \widehat{w^G}_{i 1}y_{i1}\right) -   \left(\widehat{w^C}_{i t}y_{it} - \widehat{w^C}_{i 1}y_{i1}\right) \right\},
\end{aligned}
\end{equation}
where $y_{it}$ denote the outcome variable for individual $i$ in period $t$, and the weights are defined as
\begin{equation*}
\begin{aligned}
\widehat{w^a}_{i \tau}  =  \frac{ S_{i\tau} a_i }{\frac{1}{n}\sum_{i=1}^{n}S_{i\tau} a_i  }, \text{ and } 
\quad \widehat{w^a}_{i \tau \tau+1}  =  \frac{ S_{i\tau,\tau+1} a_i }{\frac{1}{n}\sum_{i=1}^{n} S_{i\tau,\tau+1} a_i  },
\end{aligned}
\end{equation*}
with $a_i \in\{ G,C \}$, and for all $\tau = 1,...\mathcal{T}-1$. A formal treatment is presented in the general framework. Remark that these two estimators can be written as elementwise products and divisions of matrices.\footnote{For example, the chained DiD can be written in matrix form using the Hadamard product $\odot$ and division $\oslash$ as
$\widehat{ATT_{CD}}(t)   = 
 \mathbb{1}_n' \left[ \Delta Y \odot \left( W_1\oslash \mathbb{1}_n\mathbb{1}_n'W_1 - W_0\oslash \mathbb{1}_n\mathbb{1}_n'W_0 \right) \right] \mathbb{1}_{t-1},$
where $\mathbb{1}_l$ denotes a $l$-dimensional vector of $1$, $\Delta Y$ is a $n\times (t-1)$ matrix where each row is the difference $(y_{i2},...,y_{it})-(y_{i1},...,y_{it-1})$, $W_1$ a $n\times (t-1)$ matrix with $i,\tau$ element $S_{i\tau\tau+1}G_i$, and similarly for $W_0$.}

In this simple context, estimators of $ATT_{CD}(t)$  and $ATT_{CS}(t)$ can also be obtained by linear regression using a least-squares dummy variables (LSDV) approach as follows.\footnote{We thank a referee for this  valuable insight.} The $ATT_{CD}(t)$'s are obtained by  estimating the two-way fixed-effect (TWFE) event-study regression model
\begin{equation}\label{eq:TWFE event-study}
    Y_{it} = \sum_{i'=1}^n\alpha_{i'}\mathbf{1}_{\{i=i' \}} + \sum_{t'=1}^{\mathcal{T}}\delta_{t'}\mathbf{1}_{\{t=t' \}} + \sum_{l=1}^{\mathcal{T}-1} \gamma_{l+1} \mathbf{1}_{\{t=g_i - 1 + l \}} + \varepsilon_{it},
\end{equation}
by defining $\mathbf{1}_{\{.\}}$ as a dummy variable for event $\{.\}$, and $g_i$ as the first period where individual $i$ received the treatment. In this setting, $g_i=2$ for all $i$ in the treatment group and $g_i>\mathcal{T}$ for all $i$ in the control group. This regression makes it explicit that each $\gamma_{l+1}$ corresponds to the effect $l$-periods after the treatment. 

If the treatment is randomly assigned, and in the absence of variations in treatment timing, Theorem 3 of \citet{de2023two} confirms that the parallel trend assumption yields 
   $ E\left[ \hat{\gamma}_{t} \right] = ATT(t),$
with
\begin{equation}\label{eq:borusyak simple}
    \hat{\gamma}_{t} = \frac{1}{\sum_i G_i}\sum_{i:G_i=1}\left(Y_{i,t}-Y_{i,1}\right) - \frac{1}{\sum_i C_i} \sum_{i:C_i=1} \left(Y_{i,t}-Y_{i1}\right),
\end{equation}
corresponding to the long DiD estimator, if feasible. In contrast, if the sample is limited to two consecutive observations for each individual  then one can show that $\hat{\gamma}_{t}$ corresponds exactly to the chained DiD estimator in \eqref{eq:ATTCD_simple}.\footnote{Remark, nevertheless, that there is no need to sum the $\hat{\gamma}_{t}$'s  because they already correspond to the $ATT(t)$.} Similarly, the cross-section DiD estimator in \eqref{eq:ATTCS_simple} can be obtained by substituting the fixed-effects $\sum_{i'=1}^n\alpha_{i'}\mathbf{1}_{\{i=i' \}}$ by a constant term $\alpha$ in  \eqref{eq:TWFE event-study}. 

However, these TWFE event-study regressions are not heterogeneity-robust \citep{de2023two}. We identify two key limitations: (1) controlling for covariates,\footnote{Theorem S4
of \citet{chaisemartin_haultfoeuille_2020} shows that adding a term $X'\theta$ may not deliver unbiased ATT estimates under a conditional parallel trend assumption with a linear functional form for the covariates $X$.} and (2) settings with staggered adoption designs and heterogenous treatment effects.\footnote{Theorem 4
of \citet{de2023two} shows that the TWFE event-study may be biased if the treatment effect is heterogenous across
groups or over time, essentially because it includes comparisons of groups going from untreated to treated to
groups treated at both periods--the so-called ``forbidden comparison'' \citep{borusiak_jaravel_2017}--,and can be contaminated by the treatment effects at other periods \citep{de2023severaltreat}.} We will address these issues in the general case for the Horvitz-Thompson estimator \eqref{eq:ATTCD_simple}, and discuss hetereogeneity-robust regression alternatives \citep{borusiak_jaravel_2017,wooldridge2021two}.

\subsubsection{Inference}

In this simple design, the long DiD estimator $\widehat{ATT_{LD}}(t)$ is an efficient estimator of $ATT(t)$ when feasible. We compare the efficiency of the two estimators $\widehat{ATT_{CD}}(t)$ and $\widehat{ATT_{CS}}(t)$ under standard assumptions in Proposition \ref{prop: variance RSP RCS}.\footnote{In this paper, 'efficient estimator' refers to one with the lowest variance among all unbiased estimators.}





\begin{proposition}\label{prop: variance RSP RCS}
Assuming potential outcomes to be specified as in \eqref{eq: Yt components of variance}, in addition to the standard assumptions of DiD settings as presented in the proof, and a Missing Completely At Random  assumption, then $\widehat{ATT_{CD}}(t)$ is efficient for $\rho=1$. In addition, it is more efficient than $\widehat{ATT_{CS}}(t)$ if
\begin{equation*}
\begin{aligned}
\frac{2(t-1)}{1+\rho}\sigma_\eta^2 & \leq    \frac{  1.5\sigma_\eta^2 }{1-\rho^2} +  \sigma_{\alpha}^2, \quad \text{for $\rho \in [0,1)$.} 
\end{aligned}
\end{equation*}
\end{proposition}

On the one hand, the chained DiD introduces additional noise by adding up individual-transitory shocks $\varepsilon_{it}$. On the other hand, the cross-section DiD's precision depends on the variance of unobserved individual heterogeneity $\alpha_i$ and serial correlation of individual shocks $\varepsilon_{it}$. It also makes use of twice as many observations to calculate each empirical expectation since two cohorts  are observed at each $1<t<\mathcal{T}$, i.e. $t-1,t$ and $t,t+1$. In comparison, the long DiD would not suffer from any of those two efficiency losses.  

According to Proposition \ref{prop: variance RSP RCS}, $\widehat{ATT_{CD}}(t)$ is efficient under random-walk idiosyncratic errors. \citet{harmon2022difference} shows this efficiency result of the chained DiD estimator in the case with variations in treatment timing and heterogeneity, under the random walk assumption. In addition, he shows that the methods proposed by \citet{borusiak_jaravel_2017} and \citet{wooldridge2021two}, although efficient under spherical errors, are not efficient under the random walk assumption if there are multiple pre-treatment periods.

Intuitively, $\widehat{ATT_{CD}}(t)$ delivers a more precise estimate only if the sum of the extra individual-transitory shocks has a smaller variance than that of the individual-specific heterogeneity and individual transitory shock. This condition is plausible as long as $t$ is not too large, that is not too many incremental effects are aggregated together, or if $\sigma_\alpha^2$ is relatively large. As $t$ increases, a new term is added to the sum and results in a marginal increase of the variance proportional to  $ \frac{2}{1+\rho}\sigma_\eta^2 $.  For example, for $\rho=0$, the condition in Proposition \ref{prop: variance RSP RCS} becomes $\sigma_\eta^2/2 \leq \sigma_\alpha^2$ for $t=2$, and $5\sigma_\eta^2/2 \leq \sigma_\alpha^2$ for $t=6$.

Therefore, $\widehat{ATT_{CD}}(t)$ should largely dominate in terms of precision in settings with autocorrelated idiosyncratic shocks, and where the variance of unobserved individual heterogeneity is large. 

 
 \subsection{General framework}
 In the general framework, we allow for: 1) heterogeneity in treatment effects; 2) variation in treatment timing; 3)  sample selection on exogenous observables and past outcomes; and 4) general missing data patterns. 
 
 We first introduce additional notations. Let us substitute the treatment group dummy $G$ by $G_g$, a binary variable that is equal to one if an individual is first treated in period $g$. There are hence several cohorts of treatment groups. The control group binary variable $C$ denotes individuals who are never treated. Notice that for each individual $\sum_{g=2}^T G_g +C=1$. Also, define $D_t$ to be a binary variable equal to one if an individual is treated in $t$, and equal to zero otherwise. This variable will be useful to denote if the individual was first treated for some $g \leq t$. Also, define the generalized propensity score as $p_g(X) = P(G_g =1|X,G_g+C=1)$. This score measures the probability of an individual with covariates $X$ to be treated conditional on being in the treated cohort $g$ or the control group.
 
There are many parameters of interest in this setting, such as, for example, the average treatment effect $k$ periods after the treatment date. All possible parameters consist of aggregates of the most basic parameter: the average treatment effect in period $t$ for a cohort treated in date $g$, denoted by
\begin{equation}
    ATT(g,t) = E[Y_t(1)-Y_t(0)|G_g = 1].
\end{equation}
This parameter is referred to as the \emph{group-time average treatment effect} in \cite{callaway2021difference}, and we will build upon their results to study the collection of such parameters. When only unbalanced panel data is available and $t>g+1$, two methods are possible: (1) the cross-section DiD; and (2) the chained DiD. 

In what follows, we assume that individuals are sampled only for two consecutive periods and then drop out forever, so the long DiD is never feasible. We will introduce more general missing data patterns in the next subsection.

\subsubsection{Rotating panel data}


In order to identify the $ATT(g,t)$ and accommodate varying treatment timing and treatment effect heterogeneity on observable covariates $X$, we impose assumptions as follows.

\begin{assumption}[Sampling]\label{ass: sampling 2}
For all $t=1,...,\mathcal{T}-1$, $\left\{ Y_{it},Y_{it+1},X_i,D_{i1},D_{i2},...,D_{i\mathcal{T}} \right\}_{i=1}^{n_t}$ is independent and identically distributed (iid) conditional on $S_{it,t+1} =1 $.
\end{assumption}


\begin{assumption}[Missing Trends At Random]\label{ass: sampling on trends 2}
For all $t=1,...,\mathcal{T}-1$,
\begin{equation}
    S_{it,t+1} \perp Y_{it+1}-Y_{it},X_i |   D_i.
\end{equation} 
\end{assumption}

\begin{assumption}[Conditional Parallel Trends]\label{ass: unconditional parallel trend 2}
For all $t=2,...,\mathcal{T}$, $g=2,...,\mathcal{T}$, such that $g \leq t$,
\begin{equation}
E \left[ Y_{t+1}(0) - Y_{t}(0) | X,G_g=1\right] = E \left[ Y_{t+1}(0) - Y_{t}(0) | X,C=1\right] \text{ } a.s..
\end{equation}
\end{assumption}

\begin{assumption}[Irreversibility of Treatment]\label{ass: irreversibility}
For all $t=2,...,\mathcal{T}$,
\begin{equation}
D_{t-1} =1 \text{ implies that } D_t = 1.
\end{equation}
\end{assumption}

\begin{assumption}[Overlap]\label{ass: existence of treatment groups 2}
For all $g=2,...,\mathcal{T}$, $P(G_g=1) > 0$ and for some $\varepsilon > 0$, $p_g(X) < 1 -\varepsilon$ a.s.  
\end{assumption}

Assumption \ref{ass: sampling 2}  means that we are considering a rotating panel structure. Conditional on being sampled in two consecutive periods, individuals are assumed to be iid. However, unlike Assumption 3.3 in \cite{abadie_2005}, it does not imply that observations are representative of the population of interest because we do not assume that the iid draws are taken from the population distribution. Instead, we focus on the case where the identification with a cross-section DiD can fail by considering Assumption \ref{ass: sampling on trends 2}, under which $E\left[Y_{it}|S_{it},X_i,D_i\right] = E\left[Y_{it}|X_i,D_i\right]$ may not hold. 

Assumption \ref{ass: sampling on trends 2} constitutes the principal departure from \cite{callaway2021difference}. It states that the sampling process is independent of  the first-differences of individual outcomes, and the covariates $X$, conditionally on treatment assignment, i.e. \emph{trends} are missing at random. In particular, it implies that  $E[Y_{it+1}-Y_{it}|X,a=1,S_{t,t+1}=1]$ for $a \in \lbrace G_2,...,C \rbrace$ and $E[Y_{it+1}-Y_{it}|a=1,S_{t,t+1}=1]$ for $a \in \lbrace G_2,...,C \rbrace$ correspond to their population counterpart. 
Assumption \ref{ass: sampling on trends 2} is used to simplify exposition by ruling out sample selection on observables, apart from treatment status.  In practice, we only require conditional mean-independence between $S_{t,t+1}$ and $Y_{t+1}-Y_{t}$. We will present extensions to sample selection on observables shortly. 



Assumption \ref{ass: unconditional parallel trend 2} is a key identifying assumption in DiD settings with treatment heterogeneity. It means that the average outcomes for the treatment and control groups, conditional of observables, would have followed parallel paths in absence of the treatment. It is extensively discussed in \cite{abadie_2005}, \cite{callaway2021difference}. \citet{ghanem2023selection} and \citet{de2022not} provide insightful results about the underlying treatment selection mechanisms. 

Assumption \ref{ass: irreversibility} implies that once an individual is first treated, that individual will continue to be treated in the following periods. In other words, there is no exit from the treatment.\footnote{Departures from this assumption are considered in \citet{de2022difference}.} Finally, Assumption \ref{ass: existence of treatment groups 2} ensures that there are positive probabilities to belong to the control and treatment groups for any possible value of $X$. Remark that $X$ and $D_i$ are assumed to be observed for all individuals.


\paragraph{Data generating process.} Under Assumption \ref{ass: sampling 2}, the data generating process consists of random draws from the  mixture distribution $F_M(\cdot)$ defined as\footnote{In the application, we will discuss more complicated situations in which the data is generated by stratified
 sampling. The same results apply using a suitably reweighted sample 
\citep{wooldridge2010econometric,davezies2009faut}.}
\begin{equation*}
\begin{aligned}
& \sum_{t=1}^{\mathcal{T}} \lambda_{t,t+1} F_{Y_t,Y_{t+1},G_1,...,G_{\mathcal{T}},C,X|S_{t,t+1}}(y_t,y_{t+1},g_1,...,g_{\mathcal{T}},C,X|S_{t,t+1}=1),
\end{aligned}
\end{equation*}
where $\lambda_{t,t+1}=P(S_{t,t+1} = 1)$ is the sampling probability,   $y_t$ and $y_{t+1}$ denote the outcome, respectively, for an individual sampled at $t$ and $t+1$. Expectations under the mixture distribution does not correspond to population expectations.  This difference arises because of different sampling probabilities $\lambda_{t,t+1}=P(S_{t,t+1} = 1)$ across time periods and because Assumption \ref{ass: sampling on trends 2} does not preclude from some forms of dependence between the sampling process and the unobservable heterogeneity in $Y_{it}$. We introduce conditioning of $P(S_{t,t+1})$ on observables in Corollary \ref{corrolary:attritionmodel1} and \ref{corrolary:attritionmodel2}. 
However, this assumption ensures that expectations of first-differences under the mixture correspond to population expectations once conditioned on the time periods. Thereafter, $E_M[\cdot]$ denotes expectations with respect to the mixture distribution $F_M(\cdot)$, its empirical counterpart being the sample mean.

An important result of the paper is given in Theorem \ref{theorem: general identification}. We define the weights 
\begin{equation*}
w^G_{\tau\tau-1}(g) = \frac{G_g S_{\tau,\tau-1}}{E_M[G_g S_{\tau,\tau-1}]}
\end{equation*}
and 
\begin{equation*}
w^C_{\tau\tau-1}(g,X) = \frac{p_g(X)C S_{\tau,\tau-1}}{1-p_g(X)}/E_M[\frac{p_g(X)C S_{\tau,\tau-1}}{1-p_g(X)}].
\end{equation*}

\begin{theorem}\label{theorem: general identification}
Under Assumptions \ref{ass: sampling 2} - \ref{ass: existence of treatment groups 2}, and for $2 \leq g \leq t \leq \mathcal{T}$, the long-term average treatment effect in period $t$ is nonparametrically identified, and given by
\begin{equation*}
\begin{aligned}
ATT_{CD}(g,t)  = \sum_{\tau=g}^{t}\Delta ATT(g,\tau) ,
\end{aligned}
\end{equation*}
where $\Delta ATT(g,\tau)  =   E_M\left[ w^G_{\tau\tau-1}(g)\left(Y_{\tau} - Y_{\tau-1}\right) \right] -  E_M\left[ w^C_{\tau\tau-1}(g,X)  \left(Y_{\tau} - Y_{\tau-1}\right) \right]$ is the 1-period-difference group-time average treatment effects, which measure the increase of average treatment effect of group $g$ from period $t-1$ to period $t$. In addition, the cross-section DiD does not identify ATT(g,t) under Assumption \ref{ass: sampling on trends 2}. 
\end{theorem}

Those identification results suggest the two-step estimator

\begin{equation*}
\begin{aligned}
 \widehat{ATT}_{CD}(g,t)  =  \frac{1}{n}\sum_{i=1}^n \sum_{\tau=g}^t \left\{ \widehat{w}^G_{i\tau\tau-1}(g)\left(Y_{i\tau} - Y_{i\tau-1}\right) -  \widehat{w}^C_{i\tau\tau-1}(g,X)  \left(Y_{i\tau} - Y_{i\tau-1}\right) \right\}. 
\end{aligned}
\end{equation*}
where 
\begin{equation*}
\widehat{w}^G_{i\tau\tau-1}(g) = \frac{G_{ig} S_{i\tau-1}S_{i\tau}}{\frac{1}{n} \sum_{i=1}^n G_{ig} S_{i\tau-1}S_{i\tau} }
\end{equation*}
and 
\begin{equation*}
\widehat{w}^C_{i\tau\tau-1}(g) = \frac{\hat{p}_g(X_i)C_i S_{i\tau-1}S_{i\tau}}{1-\hat{p}_g(X_i)}/\frac{1}{n} \sum_{i=1}^n\frac{\hat{p}_g(X_i)C_i S_{i\tau-1}S_{i\tau}}{1-\hat{p}_g(X_i)},
\end{equation*}
with $\hat{p}_g(\cdot)$ being an estimated parametric propensity score function, such as logit or probit, obtained in a first step.

Let us denote $\widehat{ATT}_{g \leq t}$ the vector of $ \widehat{ATT}_{CD}(g,t)$'s for $g\leq t$. The next theorem establishes its joint limiting distribution.

\begin{theorem}\label{theorem: asy;ptotic results} Under Assumptions \ref{ass: sampling 2} - \ref{ass: existence of treatment groups 2} and a standard assumption on the parametric estimates of the propensity scores (Assumption 5 in \cite{callaway2021difference}  or 4.4 in \cite{abadie_2005}), for all $2 \leq g \leq t \leq \mathcal{T}$, 
\begin{equation*}
\sqrt{n}\left( \widehat{ATT}_{g \leq t}-ATT_{g \leq t} \right) \overset{d}{\to} N(0,\Sigma),
\end{equation*}
as $n \to \infty$ and where 
the covariance $\Sigma$ is detailed in the proof.
\end{theorem}

In the proof of the above theorem, we also show how the multiplier bootstrap procedure proposed by \cite{callaway2021difference}  adapts to this asymptotic result. The main difference comes from redefining the influence function, but the result about its asymptotic validity applies without other modification so it is not repeated here. We also refer the reader to their paper for a complete discussion of summary parameters which use the $ATT(g,t)$ as building-blocks. We provide some results for these summary parameters and an extension to using the not yet treated as the control group in Online Appendices \ref{sec:summaryparam} and \ref{sec:notyettreated}.

\paragraph{Sample selection on observables.} We extend these results to address sample selection on observables by considering two alternative MAR assumptions. Let us first relax Assumption \ref{ass: sampling on trends 2} so we have sample selection on $X$ and $D$.

\begin{assumption}[Missing Trends At Random 2]\label{ass: MAR Exog}
For all $t=2,...,\mathcal{T}$, 
\begin{equation}
    S_{itt+1} \perp Y_{it+1}-Y_{it} |  X_i,  D_i.
\end{equation}
\end{assumption}

Assumption \ref{ass: MAR Exog} states that \emph{trends} are now missing at random conditional on $X$ and $D$. In particular, it implies that  $E[Y_{it+1}-Y_{it}|X,a=1,S_{t,t+1}=1]$ for $a \in \lbrace G_2,...,C \rbrace$ corresponds to its population counterpart but requires outcome trends to be reweighted using $X$. Notice that it is possible to have different subsets of covariates $X_1,X_2\in X$ for selection into treatment and sample selection with minor modifications.  

\begin{corollary}\label{corrolary:attritionmodel1} Under the conditions stipulated in Theorem \ref{theorem: general identification}, and by replacing Assumption \ref{ass: sampling on trends 2} with Assumption \ref{ass: MAR Exog}, the theorem continues to hold. This is achieved through the application of  modified weights 
\begin{equation}    
\tilde w^G_{\tau\tau-1}(g,X) = \frac{E[S_{\tau}  S_{\tau-1}|G_g]}{E[S_{\tau}  S_{\tau-1}|X,G_g]} w^G_{\tau\tau-1}(g,X),
\end{equation}
\begin{equation}    
\tilde w^C_{\tau\tau-1}(g,X) = \frac{E[S_{\tau}  S_{\tau-1}|C]}{E[S_{\tau}  S_{\tau-1}|X,C]} w^C_{\tau\tau-1}(g,X).
\end{equation}
\end{corollary}

Corollary \ref{corrolary:attritionmodel1} extends the two-step estimator derived from Theorem \ref{theorem: general identification} to account for sample selection based on observables. This can be achieved by reweighting each difference $Y_t-Y_{t-1}$
 with the corresponding factor  $\frac{E[S_{\tau}  S_{\tau-1}|a]}{E[S_{\tau}  S_{\tau-1}|X,a]}$ for $a \in \{G_2,...,C \}$ before applying the same chained DiD estimator. By doing so, we control for potential biases arising from observable characteristics influencing sample selection in addition to unobservable factors, such as unobservable individual heterogeneity.\footnote{Remark that one can also opt for applying the stabilized weights $P(S_{\tau}  S_{\tau-1}|X_i,a_i)^{-1}/n^{-1}\sum_{i=1}^nP(S_{\tau}  S_{\tau-1}|X_i,a_i)^{-1}$ ex-ante.}

Finally, we extend our framework by relaxing Assumption \ref{ass: MAR Exog} to accommodate sample selection on variables $X$, $D$, and past outcomes, aligning with the Sequential Missing at Random (SMAR) assumption described in \citet{hoonhout2019nonignorable}.\footnote{We thank a referee for suggesting this valuable extension.} For the sake of clarity, we maintain the assumption that an individual is sampled for at most two consecutive periods.

\begin{assumption}[Sequential Missing At Random]\label{ass: SMAR}
For all $t=1,...,\mathcal{T-1}$, 
\begin{equation}
    Y_t\perp S_t|X,D,S_{t-1}=0,
\end{equation}
\begin{equation}
    Y_{t+1}\perp S_{t+1}|Y_{t},X,D,S_{t}=1.
\end{equation}
\end{assumption}

 Assumption \ref{ass: SMAR} states that the first time a unit is sampled is a function of $X$ and $D$.  However, for subsequent periods, the likelihood of being sampled again also depends on the past observed outcome.

\begin{corollary}\label{corrolary:attritionmodel2} Under the conditions stipulated in Theorem \ref{theorem: general identification}, and by replacing Assumption \ref{ass: sampling on trends 2} with Assumption \ref{ass: SMAR}, the theorem continues to hold. This is achieved through the application of  modified weights \begin{equation}
    \tilde{\tilde w}^G_{\tau\tau-1} = \frac{E[S_{\tau}  S_{\tau-1}|G_g]}{E[S_{\tau}  |X,Y_{\tau-1},S_{\tau-1}=1,G]E[S_{\tau-1}  |X,G]} w^G_{\tau\tau-1}(g,X),
    \end{equation} 
    \begin{equation}
    \tilde{\tilde w}^C_{\tau\tau-1} = \frac{E[S_{\tau}  S_{\tau-1}|C]}{E[S_{\tau}  |X,Y_{\tau-1},S_{\tau-1}=1,C]E[S_{\tau-1}  |X,C]} w^C_{\tau\tau-1}(g,X).
    \end{equation}
\end{corollary}

Corollary \ref{corrolary:attritionmodel2} extends the two-step estimator derived from Theorem \ref{theorem: general identification} to account for sample selection based on observables and past outcomes by using modified weights. The chained DiD estimator is otherwise unchanged.

\subsubsection{General missing data patterns}\label{sec:moregeneralsampling}

This framework naturally extends to general missing data patterns beyond rotating panel structures. For simplicity of exposition, we remove the dependence of $g$ in our notations. We now consider not only first-differences $\Delta_1 ATT(t)$ but also $k$-period differences $\Delta_k ATT(t)$ from $t-k$ to $t$ for $k=1,...,t-1$. The key requirement is that individuals must be observed at least twice to be used. 

Table \ref{table:table2} provides an example of different missing data patterns with 4 subsamples and 3 periods, and all treatments take place at $t=1$. The first subsample is a balanced panel (BP). It identifies $\Delta_1 ATT(2)$, $\Delta_1 ATT(3)$, and $\Delta_2 ATT(3)$. The second, third and fourth subsamples are incomplete panels (IP1, IP2, IP3) which only identifies a single parameter.\footnote{Note that we do not consider refreshment samples because they do not identify any parameter on their own. They could still be used to address attrition ex-ante \citep{hoonhout2019nonignorable}.}

\begin{table}[H]
\centering
\caption{\label{table:table2}
Example of a more general missing data pattern}
{
\begin{tabular*}{\columnwidth}{@{\hspace{\tabcolsep}\extracolsep{\fill}}l*{7}{D{.}{.}{1}}} \toprule          & \multicolumn{3}{c}{Obs. Indicators}  &   \multicolumn{4}{c}{Identified Parameters}  \\ Sub-population  & S_1 & S_2 & S_3 &   &  &   \\ \hline  
Balanced Panel & 1 & 1 & 1 &  \multicolumn{4}{c}{$\Delta_1 ATT(2)$, $\Delta_1 ATT(3)$, $\Delta_2 ATT(3)$ } \\
Incomplete Panel 1 & 1 & 1 & 0 &  \multicolumn{4}{c}{$\Delta_1 ATT(2)$ } \\
Incomplete Panel 2 & 0 & 1 & 1 &  \multicolumn{4}{c}{ $\Delta_1 ATT(3)$ } \\
Incomplete Panel 3 & 1 & 0 & 1 &  \multicolumn{4}{c}{$\Delta_2 ATT(3)$} \\
\bottomrule \end{tabular*}
}
\end{table}



There are multiple ways to identify the ATT from period $1$ to period $3$ in this example. We can identify $ATT(3)$ with (1)  BP alone since $ATT(3)=\Delta_2 ATT(3)$ or (2) $ATT(3)=\Delta_1 ATT(2) + \Delta_1 ATT(3)$; or by combining  (3) BP  and  IP3; or (4) BP and IP1; or (5) IP1 and IP2; and with (6) IP3 alone. The optimal combination of the $\Delta_k ATT(t)$ parameters into $ATT(t)$ parameters for general missing data patterns hence involves solving an (overidentified) linear inverse problem. Doing so will make use of all subsamples and deliver efficiency gains compared to focusing on one possible solution, e.g. using only the balanced panel. 

The inverse problem arises as follows. Consider the estimates of all possible $\Delta_k ATT(t)$, for all $t\geq 2$ and $t-1\geq k\geq 1$, 
stacked altogether into a vector $\mathbf{\Delta ATT}$ of length $L_{\Delta}$. $\Delta_k ATT(t)$ is called the k-period-difference group-time average treatment effects, which measure the increase of average treatment effect of group $g$ from period $t-k$ to period $t$. By definition $\Delta_k ATT(t) = ATT(t) - ATT(t-k)$, hence we can write
\begin{equation}\label{eq:inverse problem1}
    \mathbf{\Delta ATT} = W \mathbf{ATT},
\end{equation}
where $\mathbf{ATT}$ is the vector of $ATT(t)$ for $t\geq2$ of length $L\leq L_{\Delta}$ and $W$ is a matrix where each element takes value in $\lbrace -1,0,1 \rbrace$. Following this reasoning, the example in Table \ref{table:table2} can be written as


\begin{equation}
    \left[\begin{array}{c}
         \Delta_1 ATT(2)^{BP,IP1}  \\
         \Delta_1 ATT(3)^{BP,IP2}  \\
         \Delta_2 ATT(3)^{BP,IP3}  \\
    \end{array}\right] =  \left[\begin{array}{c c}
         1 & 0  \\
         -1 & 1  \\
         0 & 1  \\
    \end{array}\right] \left[\begin{array}{c}
         ATT(2)  \\
         ATT(3)  \\
    \end{array}\right],
\end{equation}

where we pool all subsamples that identify each $\Delta_k ATT(t)$.\footnote{Remark that we could also estimate these parameters separately for each subsample before stacking them into $\mathbf{\Delta ATT}$. That would require to estimate propensity scores conditional on subsample membership.} We propose to solve this problem using a GMM approach. Denoting $\Omega$ the covariance matrix of $\mathbf{\Delta ATT}$, the optimal GMM estimator of $\mathbf{ATT}$ corresponds to $\left(W'\Omega^{-1}W\right)^{-1}W'\Omega^{-1}\mathbf{\Delta ATT}$ since $W\bm{ATT}$ is non-random.\footnote{Remark that if an element of $\mathbf{ATT}$ is not identified, the matrix $W'\Omega^{-1}W$ (or $W'W$) will not be invertible but a (Moore-Penrose) pseudo-inverse can still be used to identify the other elements. In addition, pseudo-inverses, e.g. Tikhonov's, will deliver a stable inverse of $\Omega$ if $\mathbf{\Delta ATT}$ is high-dimensional \citep{carrasco2007linear}.} This method allows delivering efficiency gains by using all individual time-series with at least two observations, without much additional computational complexity.

\paragraph{Identification, estimation, and inference.}

In order to identify the $ATT(g,t)$ in this general framework, we must modify Assumptions \ref{ass: sampling 2} to \ref{ass: sampling on trends 2} as follows.

\begin{assumption}[Sampling]\label{ass:generalmissing sampling 2}
For all $t=1,...,\mathcal{T}$, $\left\{ Y_{it-k},Y_{it},X_i,D_{i1},D_{i2},...,D_{i\mathcal{T}} \right\}_{i=1}^{n_{t,k}}$ is independent and identically distributed (iid) conditional on $S_{it-k,t} =1 $, for $k=1,...,t-1$.
\end{assumption}

\begin{assumption}[Missing Trends At Random 3]\label{ass:generalmissing sampling on trends 2}
For all $t=1,...,\mathcal{T}$ and $k=1,...,t-1$,
\begin{equation}
S_{it-k,t} \perp Y_{it}-Y_{it-k}, X_i  | D_i.
\end{equation}
\end{assumption}




Assumption \ref{ass:generalmissing sampling 2}  means that we are considering an incomplete panel structure. Conditional on being sampled in the same two periods, individuals are assumed to be iid. However, it does not imply that individuals are sampled in two periods only, nor that observations are representative of the population of interest because we do not assume that the iid draws are taken from the population distribution. Instead, Assumption \ref{ass:generalmissing sampling on trends 2} states that the sampling process is statistically independent of  the joint distribution of k-period-differences of individual outcomes, and observables, conditionally on treatment status in any period. In particular, it implies that  $E[Y_{it}-Y_{it-k}|X,a=1,S_{t-k,t}=1]$ for $a \in \lbrace G_1,G_2,...,C \rbrace$ corresponds to its population counterpart. However, this needs not be true for $E[Y_{it}|X,a=1,S_{t}=1]$ for $a \in \lbrace G_1,G_2,...,C \rbrace$. 

A critical insight here is that under general missing data patterns, the sampling assumption can lead to testable hypotheses about trends in various data subsets. Indeed, Assumption \ref{ass:generalmissing sampling on trends 2} implies a certain level of homogeneity across different subpanels. This assumption implies, for the above example, that the trends in outcomes observed in the balanced panel aligns with those in incomplete panel 1, between periods 1 and 2, for both untreated and treated groups. However, this assumption can be relaxed to allow for sample selection based on observables, including subpanel membership, as described in Corollary \ref{corrolary:attritionmodel1}. 

Our estimation procedure is as follows:
\begin{enumerate}
    \item  Compute $\Delta_k ATT(g,t) = 1/n\sum_{i=1}^n[(\hat{w}(g)^{G}_{it,t-k}-\hat{w}(g)^{C}_{it,t-k})(Y_{it} - Y_{it-k})]$, for all $k,g,t$, and stack them into a $L_{\Delta}$-dimensional vector $\bm{\Delta ATT}$;
    \item Define the matrix $W$ appropriately;
    \item Estimate the asymptotic covariance matrix $\hat{\Omega} = n^{-1} \Psi\Psi'$, where $\Psi$ is a $L_{\Delta} \times n$ matrix with elements defined in \eqref{eq:psi's 2}.
    \item Estimate the optimal GMM estimator $\widehat{\bm{ATT}} = (W'\hat{\Omega}^{-1}W)^{-1}W'\hat{\Omega}^{-1}\bm{\Delta ATT}$;\footnote{Note further that this approach embeds the estimator proposed above in the rotating panel data setting when using only the $\Delta ATT_1(t)$ and replacing $\Omega$ by the identity matrix.}
\end{enumerate}

The next theorem establishes the joint limiting distribution of this estimator.

\begin{theorem}\label{theorem: general asymptotic results} Under Assumptions \ref{ass: unconditional parallel trend 2} - \ref{ass: existence of treatment groups 2}, a standard assumption on the parametric estimates of the propensity scores (Assumption 5 in \cite{callaway2021difference}  or 4.4 in \cite{abadie_2005}), and Assumptions \ref{ass:generalmissing sampling 2} and \ref{ass:generalmissing sampling on trends 2},   for all $2 \leq g \leq t \leq \mathcal{T}$, 
\begin{equation*}
\sqrt{n}\left( \widehat{\bm{ATT
}}-\bm{ATT} \right) \overset{d}{\to} N(0,\bm{\Sigma}),
\end{equation*}
as $n \to \infty$ and where 
the covariance $\bm{\Sigma}$ is detailed in the proof.
\end{theorem}

The theorem shows that this estimator is consistent and asymptotically normal. The bootstrap procedures for $ATT(g,t)$'s (Online Appendix \ref{sec:bootstrap}) and for summary parameters (Online Appendix \ref{bootstrapsummary}) apply with minor modifications, as discussed in the proof. Although our estimator brings efficiency gains by using an optimal weighting matrix, it could be further improved by addressing the efficiency losses from first-step plug-in estimates of propensity scores to adjust for sample selection on observables. This is left for future research.\footnote{We identified two means to address this caveat. \citet{frazier2017efficient} propose a computationally simple yet general approach that involves targeting and penalization to enforce the asymptotic efficiency for two-step extremum estimators such as ours, whereas \citet{chaudhuri2019review} and \citet{sant2020doubly} propose doubly-robust estimators which allow preserving efficiency.}

\paragraph{Regression alternatives.}
\citet{wooldridge2021two} provides a TWFE regression alternative to \citet{callaway2021difference} which results in estimators of the $ATT$'s that are numerically identical to the imputation approach of \citet{borusiak_jaravel_2017}. In the absence of covariates $X$, this TWFE regression involves individual, period, and the interaction of cohort with time-since-adoption fixed effects, as shown by
\begin{equation}\label{eq:regressionalternative}
    Y_{it} = \sum_{i'=1}^n\alpha_{i'}\mathbf{1}_{\{i=i' \}} + \sum_{t'=1}^{\mathcal{T}}\delta_{t'}\mathbf{1}_{\{t=t' \}} + \sum_{g=2}^{\mathcal{T}}\sum_{\tau=2}^{\mathcal{T}} \beta_{g\tau} G_{g}\mathbf{1}_{\{t=\tau \}} + \varepsilon_{it}.
\end{equation}
This TWFE event-study regression is equivalent to the chained DiD estimator, provided we use individual fixed-effects (and not cohort fixed-effects), as well as exclude covariates and focus on cases where only two consecutive periods are observed per individual.

A notable advantage of the regression approaches \citep{borusiak_jaravel_2017,wooldridge2021two} is their flexibility in accommodating unbalanced panel data, unlike the method of \citet{callaway2021difference}. In these frameworks, the choice between chained DiD and cross-section DiD estimators in these methods hinges on selecting between individual and cohort fixed-effects, each with its own set of identifying assumptions about the sampling process, as discussed earlier.

Despite their adaptability, these alternatives are not without limitations. When incorporating covariates, the model complexity increases significantly due to the need for many interaction terms. They also rely on potentially restrictive linearity assumptions in potential outcomes and treatment effects instead of using propensity scores. Furthermore, aggregating the estimated parameters to derive the target parameter of interest may not be trivial. It is therefore unclear that these methods are more computationally efficient than our estimators.

\section{Numerical Simulations}\label{sec:simu}

We propose a simulation design adapted from the first section. Let us specify the potential outcome as a components of variance:
\begin{equation}
    Y_{it}(D_i) = \alpha_i + \delta_t + \sum_{\tau=2}^t \beta_{\tau} D_{i\tau} + \varepsilon_{it},
\end{equation}
where $D_{i\tau} \in \lbrace 0,1 \rbrace$ denotes whether individual $i$ has been treated in $\tau$ or earlier. Let us assume that $t \in \lbrace 0,...,T+1 \rbrace$ and treatments can only occur in $t\geq 2$ so that $G \in \lbrace 2,...,T+1 \rbrace$. The data generating process is characterized by the following assumptions:
\begin{itemize}
    \item The individual-specific unobservable heterogeneity is iid gaussian: $\alpha_i \sim N(1,\sigma_{\alpha}^2)$, where $\sigma_{\alpha}^2 = 2$ ; 
    \item The time-specific unobservable heterogeneity is iid gaussian: $\delta_t \sim N(1,1)$;
    \item The error term is iid gaussian: $\varepsilon_{it} \sim N(0,\sigma_{\varepsilon}^2)$, where $\sigma_{\varepsilon}^2 = 0.5$;
    \item The probability to receive the treatment at time $g$, conditional on being treated at $g$ or in the control group, is defined as
    \begin{equation*}
    Pr(G_{ig}=1|X_i,\alpha_i,G_{ig}+C_i=1) = \frac{1}{1+\exp\left(\theta_0 + \theta_1 X_i + \theta_2 \alpha_i \times g \right)},
\end{equation*}
where $X_{i} \sim N(1,1)$ is observable for every $i$, unlike $\alpha_i$, and $\theta_0 = -1$, $\theta_1 = 0.4$ and $\theta_2 = 0$ or $\theta_2 = 0.2$. In the latter case, the treatment probability varies with treatment timing and the unobserved individual heterogeneity;
\item The sampling probability in the consecutive periods $t,t+1$ conditional on $\alpha_i$ is given by
    \begin{equation*}
    Pr(S_{itt+1}=1|\alpha_i) = \frac{1}{1+\exp\left(\lambda_0 + \lambda_1 \alpha_i\times t \right)},
\end{equation*}
with $\lambda_0 = -1$, and $\lambda_1 = 0$ or $\lambda_1 = 0.2$, so that the sampling process can  also vary with time and the unobserved individual heterogeneity.
\end{itemize}

We simulate the sampled data in two steps. First, we generate a population sample for each period $t$ to represent individuals that are either treated at $t$ or in the control group. Second, we sample from this population using the specified process. We formalize this procedure as follows:

\begin{enumerate}
\item Generate a population of individuals
\begin{enumerate}
    \item Draw $N = 2 \times \max\limits_t\{ \frac{n}{E_{\alpha}\left[Pr(S_{itt+1})\right] } \}$ individuals per period in order to have $T+2$ population samples of $N$ individuals, where each individual is characterized by a vector $(\alpha_i,\delta_t,X_i,\varepsilon_{it})$;
    \item Separately for each population sample $g$, draw a uniform random number $\xi_{i}\in[0,1]$ per individual. If $\xi_{i}\leq Pr(G_{it}=1|X_i,\alpha_i,G_{it}+C_i=1)$, then set $(G_{ig} = 1,C_i=0)$, otherwise set $(G_{ig} = 0,C_i=1)$.
    \item Compute $Y_{it}$ from $(\alpha_i,\delta_t,X_i,\varepsilon_{it},G_{i0},...,G_{iT+1},C_i)$;
\end{enumerate}    
\item Sample from this population
\begin{enumerate}
    \item Draw a uniform random number $\eta_{it}\in[0,1]$ per individual $i$ and period $t$. If $\eta_{it}\leq Pr(S_{itt+1}=1|\alpha_i)$, then set $S_{itt+1} = 1$ and $S_{i\tau \tau+1}=0$ for $\tau \neq t$;
    \item Draw (without replacement) $n$ individuals per period $t$ from the population for which $S_{it t+1}=1$;
    \item Compute the different estimators.
    \item Repeat steps 1(b)-2(c) 1,000 times and report the mean and standard deviation of the estimators.
\end{enumerate}
\end{enumerate}

We consider several simulation designs: 
\begin{itemize}
    \item DGP 1: $\theta_2=0$ and $\lambda_1=0$. This is the baseline case where the probability of treatment and the sampling process do not depend on individual heterogeneity so all estimators are unbiased.
    \item DGP 2: $\theta_2=0.2$ and $\lambda_1=0.2$. In this case, both the probability of treatment and the sampling process depend on individual heterogeneity leading to biased estimates for the cross-section DiD.
    \item DGP 3: $\theta_2=0.2$ and $\lambda_1=0.2$. In this case, we simulate a stratified sample where 90\% of the individuals are sampled on a rotating basis as before and those with $\alpha_i$ superior to the 90th percentile are always observed (10\%).
  \item DGP 4: $\theta_2=0.2$ and $\lambda_1=0.2$. We also simulate a stratified sample where individuals with $\alpha_i$ superior to the 60th percentile are always observed (40\%) and the rest is drawn as a rotating subpanel (60\%).
\end{itemize}

For all simulations, we set $T=6$, so there is a total of 8 periods. In each period, the population size is 4800, and we draw 150 individuals such that $S_{itt+1}=1$. Finally, we set $\beta_{\tau}$ to take the values $\{1.75, 1.50, 1.25, 1.00, 0.75, 0.50\}$ for $\tau=\{1,2,3,4,5,6\}$, that is, the treatment effect is positive and decreasing over time, relative to the treatment starting date. For each sample, we estimate the chained DiD (using simulated $S_{itt+1}=1$) and the cross-section DiD (using $S_{it}=1$ that are obtained from $S_{itt+1}=1$). 

For DGPs 1 and 2, we also estimate the long DiD assuming that all sampled individuals are observed for the entire time frame. The simulation results are given in Table \ref{tab:table3}. This long DiD here is infeasible but serves as a benchmark to illustrate the significant loss of information resulting from having an unbalanced panel. However, the chained DiD delivers unbiased estimates in both cases unlike the cross-section DiD. Remark that we assume iid errors although the chained DiD perform best for random-walk errors.

\begin{table}[H]
\centering
\caption{\label{tab:table3}
Simulation results for a rotating panel}
{\footnotesize
\begin{tabular*}{1\hsize}{@{\hskip\tabcolsep\extracolsep\fill}l*{6}{c}}
\hline
           &           \multicolumn{ 3}{c}{DGP 1} &           \multicolumn{ 3}{c}{DGP 2} \\
           \cmidrule(lr){2-4} \cmidrule(lr){5-7}
           & Chained DiD &     CS DiD &   Long DiD & Chained DiD &     CS DiD &   Long DiD \\
\hline
$\beta_1$ &      1.748 &      1.745 &       1.75 &      1.752 &      1.894 &       1.75 \\
           &    (0.099) &    (0.199) &    (0.017) &     (0.097) &    (0.148) &    (0.016) \\
$\beta_2$ &        1.5 &      1.498 &        1.5 &       1.499 &      1.774 &      1.501 \\
           &    (0.164) &    (0.305) &     (0.02) &     (0.157) &    (0.218) &    (0.018)  \\
$\beta_3$ &      1.248 &       1.25 &      1.251 &      1.254 &      1.651 &       1.25 \\
           &    (0.231) &    (0.355) &    (0.023) &    (0.224) &    (0.257) &     (0.02)  \\
$\beta_4$ &          1 &      1.006 &          1 &       1.002 &      1.522 &          1 \\
           &      (0.3) &    (0.412) &    (0.027) &     (0.293) &    (0.291) &    (0.023)  \\
$\beta_5$ &      0.741 &      0.765 &       0.75 &     0.739 &      1.369 &       0.75 \\
           &    (0.406) &    (0.521) &    (0.033) &   (0.395) &    (0.365) &    (0.028) \\
$\beta_6$ &      0.499 &      0.522 &      0.499 &      0.5 &      1.209 &      0.503 \\
           &    (0.586) &    (0.711) &    (0.046) &    (0.603) &    (0.515) &     (0.04) \\
\hline
\end{tabular*}  

}
\begin{tablenotes} \item \scriptsize  Notes: 
This table shows results obtained from the simulations described above. Simulated $\beta_{\tau}$ take the values $\{1.75, 1.50, 1.25, 1.00, 0.75, 0.50\}$ for $\tau=\{1,2,3,4,5,6\}$.
\end{tablenotes}
\end{table}

For DGPs 3 and 4, the long DiD is estimated only for individuals that belong to the balanced subpanel. The simulation results are given in Table \ref{tab:table4}. We show estimates from the chained DiD GMM estimator using two weighting matrices:  (1) Ch DiD uses the identity matrix, and (2) CD-GMM uses the optimal weighting matrix presented earlier. It appears that when the balanced subpanel consists of only 10\% of the data, the chained DiD estimators outperform the long DiD. However, the (asymptotically) optimal weighting matrix does not always deliver more precise estimates than the identity matrix in small samples, at least for this simulation design. In comparison to the identity matrix, It seems that the optimal weights allow mitigating the precision loss for longer-term effects (e.g. $\beta_6$) at the cost of losing some precision for smaller-term effects (e.g. $\beta_1$).

\begin{table}[H]
\centering
\caption{\label{tab:table4}
Simulation results for a stratified sample}
{\footnotesize
\begin{tabular*}{1\hsize}{@{\hskip\tabcolsep\extracolsep\fill}l*{8}{c}}
\hline
           &           \multicolumn{ 4}{c}{DGP 3} &           \multicolumn{ 4}{c}{DGP 4} \\
           \cmidrule(lr){2-5} \cmidrule(lr){6-9}
           & Ch DiD & CD-GMM &     CS DiD &   Long DiD & Ch DiD & CD-GMM &     CS DiD &   Long DiD \\
\hline
$\beta_1$&1.753&1.748&1.714&1.753&1.754&1.755&1.715&1.754 \\
&(0.085)&(0.151)&(0.14)&(0.085)&(0.052)&(0.057)&(0.096)&(0.052)\\
$\beta_2$&1.501&1.488&1.413&0.897&1.502&1.503&1.422&1.504\\
&(0.127)&(0.195)&(0.208)&(0.315)&(0.061)&(0.06)&(0.137)&(0.069)\\
$\beta_3$&1.256&1.237&1.112&0.868&1.252&1.25&1.135&1.25\\
&(0.177)&(0.234)&(0.239)&(0.311)&(0.072)&(0.062)&(0.148)&(0.068)\\
$\beta_4$&1.005&0.988&0.8&0.812&1.003&1.002&0.855&1.002\\
&(0.215)&(0.285)&(0.268)&(0.336)&(0.077)&(0.065)&(0.158)&(0.071)\\
$\beta_5$&0.749&0.749&0.459&0.706&0.75&0.751&0.567&0.751\\
&(0.285)&(0.331)&(0.339)&(0.352)&(0.09)&(0.074)&(0.183)&(0.078)\\
$\beta_6$&0.508&0.52&0.123&0.511&0.509&0.511&0.316&0.511\\
&(0.412)&(0.398)&(0.479)&(0.393)&(0.121)&(0.094)&(0.237)&(0.097)\\
\hline
\end{tabular*}  

}
\begin{tablenotes} \item \scriptsize  Notes: 
This table shows results obtained from the simulations described above. Simulated $\beta_{\tau}$ take the values $\{1.75, 1.50, 1.25, 1.00, 0.75, 0.50\}$ for $\tau=\{1,2,3,4,5,6\}$.
\end{tablenotes}
\end{table}


\section{Application: The Employment Effects of an Innovation Policy in France}\label{sec:application}

We now turn to an application of these methods for estimating the causal impact of a French innovation policy supporting collaborative R\&D  projects over the period 2010-2016. 

\subsection{Background}

This innovation policy is made up of different subsidy schemes aimed at developing R\&D collaborations between firms and, often, public organizations. These schemes aim at subsidizing collaborative projects oriented towards applied research and experimental development.\footnote{The schemes are FUI, ISI, PSPC, PIAVE, RAPID and ADEME. Although they share the same general objective, they support different forms of R\&D projects. For example, PSPC projects are much larger in size than others, FUI projects systematically involve companies and public research organizations, ADEME projects have environmental objectives, and so on. A detailed description of the schemes is available in \cite{rapport_dge_2020}.} 
To obtain funding, a firm must set up a research project in partnership with at least one other institution. The project is then submitted to one specific subsidy scheme, often following calls for proposals, much like research grant in academic research. The selection of projects, and the associated funding, is based on a list of criteria including, but not limited to, the innovative nature of the project, its credibility, maturity, or commercial character. 

This innovation policy  provides an ideal setting to apply our method because R\&D projects take several years to complete. Evaluating the effectiveness of the policy hence requires estimating its long-term effects. Unfortunately, one of the main data sources about firm-level R\&D activities comes from a survey with a  rotating panel design with multiple strata. Firms investing a large amount in R\&D expenditure, i.e. large companies, are systematically surveyed, while those spending less, i.e. small and medium-sized enterprises (SMEs) and intermediate-sized enterprises (ISEs), are surveyed only two consecutive years and then dropped out of the sample. Furthermore, large firms are always involved in at least one R\&D project, so it is not possible to find a plausible counterfactual for them. Consequently, our policy evaluation focus on SMEs and ISEs, which form an unbalanced panel. 

The mechanism for selecting projects submitted by companies should ideally lead to the acceptance of good projects and the rejection of weaker ones. It is conceivable that the accepted projects would have been implemented regardless of subsidies, indicating that what we observe is somehow also reflective of the role played by project quality rather than simply the impact of subsidies. However, the associated risk with these projects is high, and it is unclear that they would have been carried out without subsidies, which  account for more than 30\% of the total project financing on average. More importantly, the primary goal of our empirical analysis is to assess whether subsidies for collaborative projects are useful in boosting R\&D activities within SMEs and ISEs. In this context, the outcome variable is not directly measuring the project's success. Instead, we use the number of researchers and highly qualified workforce, which serves as an indicator of the additional investment in R\&D activities. Empirical evidence also indicates significant variability within companies, where one project may be accepted while another is rejected. This variability suggests that it would be too simplistic to label a company as uniformly successful or unsuccessful at attracting subsidies for its projects.\footnote{Note that the treatment is properly defined as receiving a first subsidy for a collaborative project in our application. } 

Furthermore, the gradual implementation of the studied mechanisms means that a treated company at a given period often serves as a control for other companies in the period preceding its treatment (the ``not yet treated''). This is particularly true given the prolonged maturation of projects, and it is common for the same project to be submitted multiple times in various calls for projects before a definitive acceptance or rejection. Rejected projects may also find better alignment with other funding opportunities and schemes that better suit their objectives. Consequently, many projects are declined not necessarily due to their intrinsic quality but rather because of insufficient maturity or misalignment with the goals of the various subsidy schemes.

The application presented in this paper focuses on the average treatment effect of participating for the first time in any of these schemes without distinguishing their individual effects. It is relevant to analyze this average effect to the extent that all the  schemes contribute to the common objective of subsidizing collaborative R\&D projects. This aggregation leads to more precise estimates because it maximizes the sample size, but at the same time requires to account for treatment effect heterogeneity and variations in treatment timing.\footnote{It is almost impossible to precisely estimate the individual effect of the smallest schemes as they have subsidized a very small number of projects.} 

We estimate the treatment effect of this policy on employment. We focus on employment because it is the main economic variable for which it is also possible to consistently observe an almost identical measure for all firms using administrative data. Therefore, by focusing on such outcome variable, we can compare the results obtained with the chained DiD estimator to the unfeasible estimates obtained using the long DiD. The complete results of this policy evaluation, including the effects on a larger number of economic variables, are available in \cite{rapport_dge_2020}.

\subsection{Data}

Our data contains information about all R\&D projects financed under these schemes over the period 2010-2016. The data includes an unique identifier for each partner participating in a project.\footnote{This identifier corresponds to the SIREN number, a unique identification number for French businesses  supervised by the French national institute of statistics.} Using this identifier, we have collected exhaustive firm-level data from administrative sources that provide the main annual indicators on the economic activity of companies over the period 2007-2017. In particular, we collected the following variables: total workforce and the number of engineers in the workforce from administrative records on firms' employment as outcomes of interest, and other variables used in the propensity scores.\footnote{Firms' revenues come from annual tax data restated by INSEE (FICUS/FARE datasets). Employment information comes from administrative records on firms' employment (DADS datasets). Data on cluster policy comes from the French competitiveness cluster (``P\^{o}le de Comp\'etitivit\'e") management database. Finally, data on support for innovation comes from the ``Cr\'edit Imp\^{o}t Recherche'' (CIR), a research tax credit, and the ``Jeunes Entreprises Innovantes'' (JEI) scheme, a tax and social exemption aimed at young innovative firms. The CIR is the main tool for supporting innovation in France. Contrary to the devices evaluated in this article, the CIR is an indirect tax aid, in the sense that it is automatically distributed to companies making eligible R\&D expenditures and that apply for it.} 

Data on the number of researchers is obtained from the R\&D survey.\footnote{This survey also provides detailed information on R\&D expenditures, the financing of these expenditures, and some outputs.} This annual survey collects information from about 9000 firms companies each year. The survey has a stratified sampling design: firms with intramural R\&D expenditures above 750,000 euros are systematically surveyed the following year, while others are surveyed only two years in a row. The vast majority of the SMEs and ISEs in the scope of the study are part of the second stratum. In this context, the application of the standard long DiD method is not possible, which justifies the use of the method developed in this article. 

Having knowledge of the amount of CIR (research tax credit) paid to companies and their participation in the French cluster policy is essential. Indeed, the amount of tax credit granted is a good proxy for a company's propensity to be active in R\&D. In addition, competitiveness clusters aim to create a network of firms and research organizations to facilitate the formation of collaborative R\&D projects. Participation in one of these clusters also reveals a firm's tendency for this form of R\&D. These different variables are therefore suitable candidates to explain the probability of receiving the treatment in the propensity score.\footnote{More specifically, the propensity score includes the log of R\&D grants, the log of the number of engineers, the log of investment, the log of the variation in R\&D grants, the log of the variation in turnover, an indicator for being in a competitiveness cluster, and an indicator for being in the IT sector.} It is important to observe some variables comprehensively across firms so they can be included in the propensity score, as they must be observed in $g-1$, $t$, and $t+1$ to compute the elementary building block constituting each chain link of the chained DiD.

The number of firms that can potentially carry out an innovative R\&D activity is very small compared to the total number of firms in France (about 3,000,000). Therefore, in order to avoid comparing firms participating in a collaborative R\&D project to other firms which are unlikely to pursue such activity, we restrict the scope of the study to firms active in R\&D at least one year over the whole period considered. This activity is measured by merging together all the sources of information available to us for this purpose: the databases of the research tax credit, the JEI scheme, and the R\&D survey. Doing so leaves us with about 30,000 firms in the sample. 

 The schemes covered by the study are the main support mechanisms for collaborative R\&D in France and involve the highest amounts of public support. However, there are other alternatives not discussed here. 
Altogether, the various schemes considered have provided funding for 1697 projects over the period, and have involved 8724 partners.\footnote{A firm can participate in several projects.} These projects received a total aid of 3.6 billion euros and involved total expenditures of 10.4 billion euros from their partners. 

\paragraph{Sample selection on observables.} The chained difference-in-differences method relies on the identifying assumption that changes in the potential outcomes are (conditionally) independent of being sampled twice in a row. A key advantage of the method is that it is  robust to sample selection on unobservable time-persistent factors. However, the design of the R\&D survey implies that the probability of observing the same company twice is positively correlated with observables, in particular past R\&D expenditures. Surveyed companies with larger (or increasing) R\&D in $t$ are more likely to be surveyed again in $t+1$. Furthermore, companies likely to engage in R\&D, identified through auxiliary information, are all surveyed and integrated into the survey data if they conduct R\&D. The compilation of various annual surveys therefore results in some sample selection on past R\&D expenditures at the company level.

We address this selection issue using inverse propensity weighting of the first-differences $Y_t - Y_{t-1}$ following Corollary \ref{corrolary:attritionmodel1} under a variation of  Assumption \ref{ass: MAR Exog}. For tractability, we specify a logit model for the conditional probability $P(S_{tt+1} = 1| Z)$, where $Z$ represents the set of variables determining the survey sampling design, including the level and changes in R\&D expenditures and participation in known R\&D support mechanisms up to date $t-1$.\footnote{Remark that this simplified approach assumes sampling $S_{t-1}S_t$ to be mean-independent  of treatment status $D$ and treatment selection variables $X$, conditional on $Z$.} This approach controls for the effects of selection in the R\&D survey sample and stabilizes the scope of the evaluation.


\subsection{Results}

In this application, we focus on the effects on (1) total workforce and (2) the employment of managers and highly qualified workers. Those variables are  observed exhaustively from administrative data (DADS). We estimate the dynamic treatment effect on these two observed outcomes by using three estimators: the long DiD, the chained DiD, and the cross-section DiD. 

Although these outcomes are consistently observed through time, attrition can still occur. For example, firms may disappear over time because of economic difficulties or because they are acquired by another firm. These companies are not taken into account by the long DiD estimator, unless with an attrition model, whereas they are accounted for by the chained DiD and cross-section DiD estimators. 

We balance the data to both facilitate the comparison across estimators and get closer to our theoretical framework. That is, we keep the firms that are consistently observed  from 2007 to 2017 in the administrative data. This balanced panel is referred to as the exhaustive panel. We use the long DiD, the chained DiD and the cross-section DiD on this exhaustive panel.

Then, we construct an unbalanced version of this panel by discarding all observations whenever a firm was not sampled in the R\&D survey in a given year. Both the chained DiD and cross section DiD estimators are used on this (artificial) unbalanced panel. The objective is to study how each estimator is affected by discarding observations, and compare its performance to the long DiD.

Results are summarized in Figures \ref{fig:figure1} and \ref{fig:figure2} for the effects on total workforce and Figures \ref{fig:figure3} and \ref{fig:figure4} for highly qualified workers, along with 95\% confidence intervals obtained from the multiplier bootstrap.\footnote{To maximize readability, Figures 2 and 4 are the same as Figures 1 and 3, respectively, to which we have added the cross-section DiD on the unbalanced panel, whose variance is very high. } They show the dynamic treatment effects relative to the beginning of the treatment obtained on a panel balanced on exhaustive variables.\footnote{``Exhaustive" refers to the exhaustively observed outcome. ``Unbalanced" refers to the use of an exhaustively observed outcome from which we artificially discard all observations identified as missing in the R\&D survey. }  Table \ref{tab:table5}  and Table \ref{tab:table6} in the Online Appendix provide additional details, including pre-trend tests rejecting the null hypothesis of pre-trends.

\begin{figure}[H]
    \centering
    \includegraphics[scale=0.75]{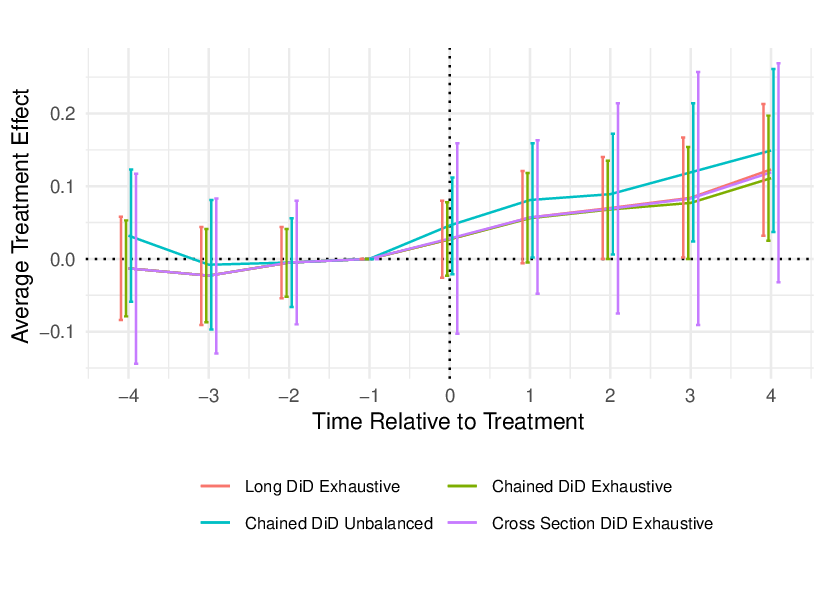}
    \vspace{-2em}\caption{Effects on total workforce for selected estimators}
    \label{fig:figure1}
\end{figure}

\begin{figure}[H]
    \centering
    \includegraphics[scale=0.75]{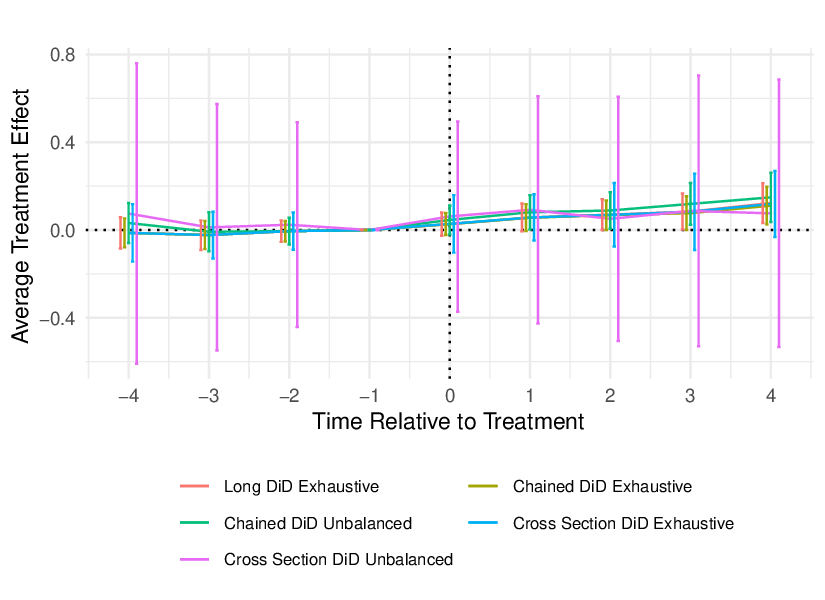}
    \vspace{-1.5em}\caption{Effects on total workforce for all estimators}
    \label{fig:figure2}
\end{figure}

\begin{figure}[H]
    \centering
    \includegraphics[scale=0.75]{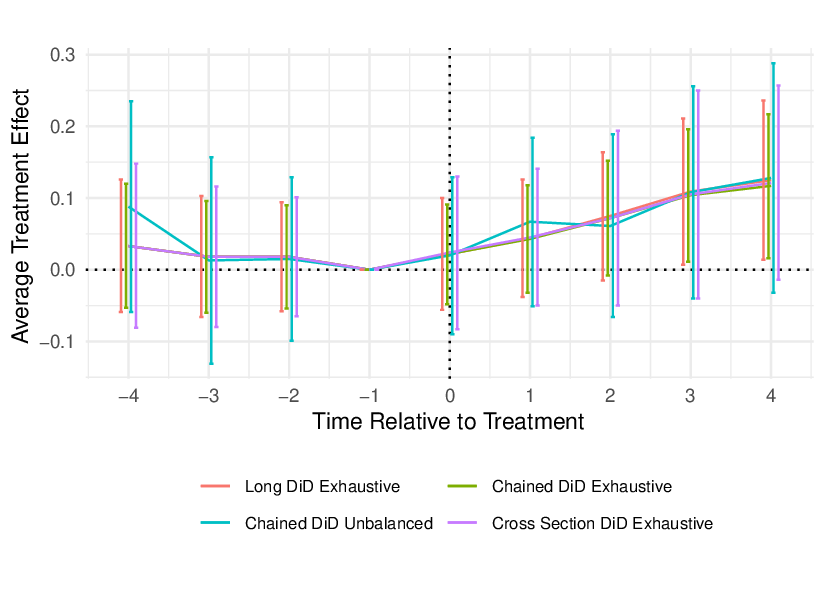}
    \vspace{-2em}\caption{Effects on highly qualified workforce for selected estimators}
    \label{fig:figure3}
\end{figure}

\begin{figure}[H]
    \centering
    \includegraphics[scale=0.75]{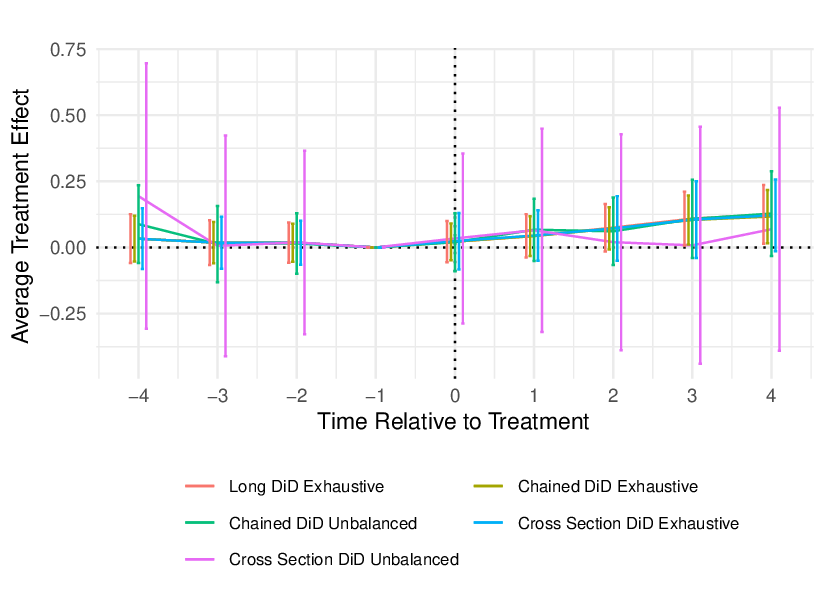}
    \vspace{-1.5em}\caption{Effects on highly qualified workforce for all estimators}
    \label{fig:figure4}
\end{figure}

The long DiD reveals that total employment has increased by 5.7\% the year after the project started ($\beta_2$ in Table \ref{tab:table5}). The long-term increase amounts to 12\% five years after the start. For highly qualified workers, the effect is also positive but not statistically significant during the first two years. It becomes significant from the third year onwards. 

As expected, the estimates obtained using the complete (exhaustive) panel are very similar between the long DiD, the chained DiD and the cross-section DiD estimators. The standard errors are also similar between the long DiD and the chained DiD estimators but they are much higher for the cross section DiD estimator, which suggests that the variance of the unobserved heterogeneity is large in this data. 
On the one hand, the estimated coefficients remain quite similar and the standard errors are only slightly higher with the chained DiD estimator on the unbalanced panel. 
On the other hand, the estimates are considerably worse when obtained with the cross-section DiD estimator on the unbalanced panel (Figures \ref{fig:figure2} and \ref{fig:figure4}). The point estimates are different and the standard errors become too large to appreciate the effects of the policy.


Finally, we apply the chained DiD and cross-section DiD estimators on outcomes similar to those studied just above but coming from the R\&D survey. In this context, it is not possible to apply the standard long DiD estimator because there are too few observations to calculate long differences. For consistency reasons, we present estimates obtained on the same set of firms as those used for the tables \ref{tab:table5} and \ref{tab:table6}.\footnote{That is, we use the set of data that is balanced on the exhaustive variables, which is then merged with the R\&D survey, and estimate the effects on the variables reported in the R\&D survey.} 

The outcome variables do not correspond to the exact same definition of employment depending on whether they come from the exhaustive administrative source or from the R\&D survey. Total employees headcount from DADS administrative data corresponds to observations at the legal unit level, whereas the R\&D survey sometimes provide information on employment at the group level.\footnote{A legal unit is a legal entity of public or private law. A firm, in the sense of a group, is an economic entity that may comprise several legal units thanks to financial links.} The employment of highly qualified workers from the DASD is close to the number of R\&D researchers and engineers filled in the R\&D survey. 
Highly qualified workers include engineers but also qualified workers dedicated to other tasks than R\&D. Conversely, R\&D researchers and engineers include researchers who are specifically assigned to research tasks. 
Despite their difference, these variables measure similar outcomes and are highly correlated, which justifies their comparison.

Results using outcome variables observed from R\&D survey are presented in Figure \ref{fig:figure5}, and details are provided in Table \ref{tab:table7} in the Online Appendix. The effects obtained with the chained DiD on total employment are somewhat less significant than those presented in Figure \ref{fig:figure1}, but the coefficients have the same order of magnitude. As might be expected since the policy directly aims at fostering R\&D activities, the effects on the number of researchers are stronger and more significant. On the other hand, the effects obtained with the  cross section DiD are, once again, much less precise. 

\begin{figure}[H]
    \centering
    \includegraphics[scale=0.75]{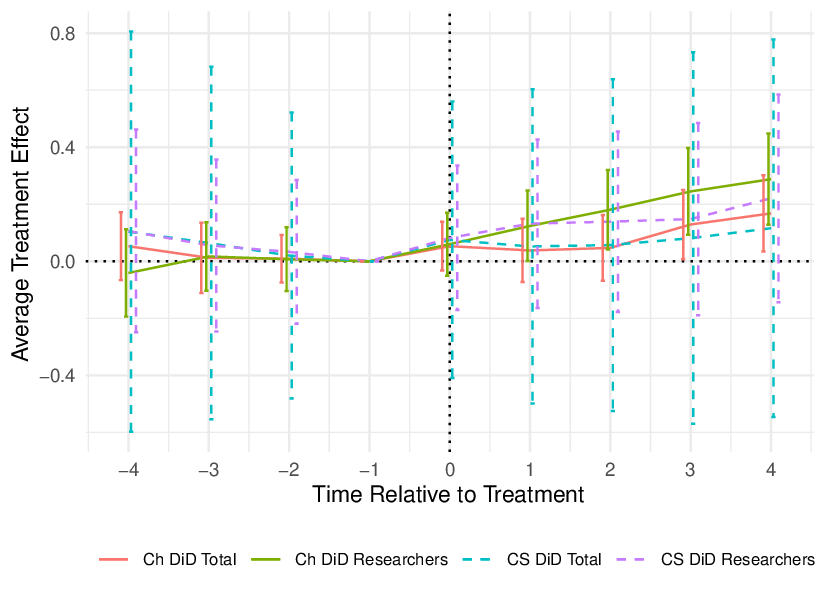}
    \vspace{-1.5em}\caption{Effects using only R\&D survey data}
    \label{fig:figure5}
\end{figure}

We identify two primary reasons for the superior performance of the chained DiD estimator over the cross section DiD estimator in this application.  First, employment typically exhibits high temporal persistence, coupled with considerable unobserved heterogeneity. Proposition \ref{prop: variance RSP RCS} demonstrates that the time-dependence of the outcome variable is a crucial determinant in the relative precision of these estimators. Second, there may be sample selection due to unobservable time-persistent factors in the R\&D survey. Such selection arises if, for example, the propensity to consistently participate in R\&D surveys varies, possibly due to differences in resource availability, past experience with surveys, corporate culture, or differing levels of motivation across research teams in firms.


The time-persistent differences between high-performing and non-performing companies are eliminated by first-differencing. If there are distinct dynamics between treated and untreated firms, then differences could be observed even before treatment initiation. However, placebo tests for pre-trends show no significant employment differences upstream of initial treatment period between treatment and control groups (effects $\beta_{-3}$, $\beta_{-2}$ and $\beta_{-1}$ in Tables \ref{tab:table5} to \ref{tab:table7}). This empirical test supports the absence of a violation of the parallel trend assumption.

Finally, we reproduce the results of Tables \ref{tab:table5}, \ref{tab:table6}, and \ref{tab:table7} by estimating the treatment effects with the complete original sample from the administrative data, without discarding the individual firms that are not consistently observed throughout the period to create a balanced panel. The results are presented in Online Appendix \ref{appe:application} and confirm the better performance of the chained DiD estimator.

\section{Conclusion}\label{sec:conclusion}

In this paper, we have developed a new estimator to identify long-term treatment effects in unbalanced panel data sets. This is an important issue, not only because of attrition but also due to how surveys are designed. Common practices are either to use a long DiD estimator by balancing the data, at the cost of losing precision and possibly biasing the results, or to use a cross-section DiD estimator at the cost of not accounting for unobserved heterogeneity. We introduce a new method that simply consists of aggregating short-term DiD estimators obtained from two periods. Our theoretical results show that this estimator identifies the average treatment effects of interest, is consistent and asymptotically normal, accounts for treatment heterogeneity and varying treatment timing, as well as general missing data patterns, and may deliver efficiency gains. 
An application to an innovation policy implemented in France reveals that, indeed, this estimator allows identifying statistically significant long-term treatment effects where previous methods fail to do so.

\section*{Acknowledgments}
We are indebted to Serena Ng for her guidance, an associate editor, and two anonymous referees for their  valuable comments and suggestions that greatly improved the quality of this work. We would also like to thank Brantly Callaway, Joel Cuerrier, Laurent Davezies, Cl\'{e}ment de Chaisemartin, Xavier D'Haultef{\oe}uille, Yannick Guyonvarch, Xavier Jaravel, and all participants to seminars and conferences for insightful discussions and comments. 

The R package \textbf{cdid} is available on CRAN.

\section*{Funding}
This work was supported by the Social Sciences and Humanities Research Council of Canada (430-2022-00544).

\bibliographystyle{aea2}
\bibliography{bib_chained_did.bib}

\clearpage

\clearpage
\appendix
\setcounter{page}{1}
\setcounter{section}{0}
\numberwithin{table}{section}
\numberwithin{figure}{section}
\Large 
\begin{center}
\textbf{Online Appendices}
\end{center}
\normalsize
\emph{Bell\'{e}go, C., Benatia, D. and V. Dortet-Bernadet (2024) The Chained Difference-in-Differences. To appear in the Journal of Econometrics.}

\small
Appendix \ref{app:mathappendix} collects all the proofs and Appendix \ref{appe:application} provides additional results for the empirical application.

\section{Mathematical Appendix}\label{app:mathappendix}

\subsection{Simple setting}

 \begin{proof}[Proof of Proposition \ref{prop: variance RSP RCS}]
 
  In order to identify and estimate the $ATT(t)$ in this simple setting, we impose the following standard assumptions.

\begin{assumption}[Sampling]\label{ass: sampling}
For all $t=1,...,\mathcal{T}$, $\left\{ Y_{it},Y_{it+1},D_{i1},D_{i2},...,D_{i\mathcal{T}} \right\}_{i=1}^{n_t}$ is independent and identically distributed (iid) conditional on $S_{t,t+1} =1 $.
\end{assumption}

\begin{assumption}[Missing Completely At Random]\label{ass: sampling on levels}
For all $t=1,...,\mathcal{T}$,
\begin{equation}
Pr(S_t=1|Y_1,...,Y_{\mathcal{T}},D_{1},...,D_{\mathcal{T}}) = Pr(S_t=1).
\end{equation}
\end{assumption}

\begin{assumption}[Unconditional Parallel Trends]\label{ass: unconditional parallel trend}
For all $t=2,...,\mathcal{T}$,
\begin{equation}
E \left[ Y_t(0) - Y_{t-1}(0) | G=1\right] = E \left[ Y_t(0) - Y_{t-1}(0) | C=1\right] \text{ } a.s..
\end{equation}
\end{assumption}

\begin{assumption}[Irreversibility of Treatment]
For all $t=2,...,\mathcal{T}$,
\begin{equation}
D_{t-1} =1 \text{ implies that } D_t = 1.
\end{equation}
\end{assumption}

\begin{assumption}[Existence of Treatment and Control Groups]\label{ass: existence of treatment groups}
\begin{equation}
P(G=1)= 1-P(C=1) \in (0,1).
\end{equation}
\end{assumption}

Substituting \eqref{eq: Yt components of variance} and rearranging yields

\begin{equation*}
\begin{aligned}
\widehat{ATT_{CD}(t)} 
& = 
\sum_{\tau=1}^{t-1}  \frac{1}{n}\sum_{i=1}^{n}\left\{ \widehat{w^G_{i\tau\tau+1}}\left(y_{i\tau+1} - y_{i\tau}\right) -   \widehat{w^C_{i\tau\tau+1}} \left(y_{i\tau+1} - y_{i\tau}\right) \right\} \\
& = 
\sum_{\tau=1}^{t-1}  \frac{1}{n}\sum_{i=1}^{n}\left\{ \widehat{w^G_{i\tau\tau+1}}\left(\delta_{\tau+1} - \delta_{\tau} + \beta_{t} + \varepsilon_{i\tau+1} - \varepsilon_{i\tau}  \right) -   \widehat{w^C_{i\tau\tau+1}} \left(\delta_{\tau+1} - \delta_{\tau} + \varepsilon_{i\tau+1} - \varepsilon_{i\tau}\right) \right\} \\
& = 
\sum_{\tau=1}^{t-1} \left[ \frac{1}{n}\sum_{i=1}^{n}\widehat{w^G_{i\tau\tau+1}} \beta_{\tau+1}    
+ \frac{1}{n}\sum_{i=1}^{n}\left(\widehat{w^G_{i\tau\tau+1}} - \widehat{w^C_{i\tau\tau+1}} \right)\left(\varepsilon_{i\tau+1} - \varepsilon_{i\tau}\right)  \right] \\
& = 
\sum_{\tau=2}^t \beta_{\tau} 
+ \sum_{\tau=1}^{t-1} \left[\frac{1}{n}\sum_{i=1}^{n}\widehat{w^G_{i\tau\tau+1}} \left(\varepsilon_{i\tau+1} - \varepsilon_{i\tau}\right) - \frac{1}{n}\sum_{i=1}^{n}\widehat{w^C_{i\tau\tau+1}} \left(\varepsilon_{i\tau+1} - \varepsilon_{i\tau}\right)  \right], \\
\end{aligned}
\end{equation*}
where the second equality follows from the fact that $w^G_{i\tau\tau+1}\neq 0$ and $w^C_{i\tau\tau+1}\neq 0$ only if $y_{i\tau+1} - y_{i\tau}$ is observed, and the fourth and fifth equalities follow from $ \sum_{i=1}^{n}\widehat{w^G_{i\tau\tau+1}} =\sum_{i=1}^{n} \widehat{w^C_{i\tau\tau+1}} = 1$. The second term in the final expression vanishes in expectations from Assumptions \ref{ass: sampling} and \ref{ass: unconditional parallel trend}. 

The second estimator is given by

\begin{equation*}
\begin{aligned}
\widehat{ATT_{CS}(t)}      
     = &  \frac{1}{n}\sum_{i=1}^{n}\left\{ \left(\widehat{w^G_{i t}}y_{it} - \widehat{w^G_{i 1}} y_{i1} \right)  -   \left(\widehat{w^C_{i t}}y_{it} - \widehat{w^C_{i 1}}y_{i1}\right)  \right\} \\
     = & \frac{1}{n}\sum_{i=1}^{n} \left( \widehat{w^G_{i t}} - \widehat{w^C_{i t}}  \right) y_{it}  -    \frac{1}{n}\sum_{i=1}^{n} \left(\widehat{w^G_{i 1}} - \widehat{w^C_{i 1}}  \right)y_{i1}   \\
     = &  \frac{1}{n}\sum_{i=1}^{n} \left( \widehat{w^G_{i t}} - \widehat{w^C_{i t}}  \right) \left(\alpha_i +\delta_t + \sum_{\tau=2}^t \beta_{\tau}D_{i\tau} +\varepsilon_{it}\right)  -    \frac{1}{n}\sum_{i=1}^{n} \left(\widehat{w^G_{i 1}} - \widehat{w^C_{i 1}}  \right)  \left(\alpha_i +\delta_1  +\varepsilon_{i1}\right)  \\
 = &   \frac{1}{n}\sum_{i=1}^{n}  \widehat{w^G_{i t}} \sum_{\tau=2}^t \beta_{\tau} 
    +    \frac{1}{n}\sum_{i=1}^{n} \left( \widehat{w^G_{i t}} - \widehat{w^C_{i t}} - \widehat{w^G_{i 1}} + \widehat{w^C_{i 1}}  \right) \alpha_i +    \frac{1}{n}\sum_{i=1}^{n} \left( \widehat{w^G_{i t}} - \widehat{w^C_{i t}}   \right) \delta_t  \\ 
    & -    \frac{1}{n}\sum_{i=1}^{n} \left( \widehat{w^G_{i 1}} - \widehat{w^C_{i 1}}   \right) \delta_1   +    \frac{1}{n}\sum_{i=1}^{n} \left( \widehat{w^G_{i t}} - \widehat{w^C_{i t}}   \right)\varepsilon_{i1}  
     -    \frac{1}{n}\sum_{i=1}^{n} \left( \widehat{w^G_{i 1}} - \widehat{w^C_{i 1}}   \right) \varepsilon_{i1}   \\
     = &   \sum_{\tau=2}^t \beta_{\tau} 
    +    \frac{1}{n}\sum_{i=1}^{n} \left( \widehat{w^G_{i t}} - \widehat{w^C_{i t}}   \right)\varepsilon_{it}  
     -    \frac{1}{n}\sum_{i=1}^{n} \left( \widehat{w^G_{i 1}} - \widehat{w^C_{i 1}}   \right) \varepsilon_{i1}  \\
      & +    \frac{1}{n}\sum_{i=1}^{n} \left( \widehat{w^G_{i t}} - \widehat{w^C_{i t}}   \right)\alpha_i  
     -    \frac{1}{n}\sum_{i=1}^{n} \left( \widehat{w^G_{i 1}} - \widehat{w^C_{i 1}}   \right) \alpha_i,  \\
\end{aligned}
\end{equation*}

where the second and third terms vanish in expectations under Assumption  \ref{ass: unconditional parallel trend}. Finally, the expectation of the last term equates zero under Assumption \ref{ass: sampling on levels}.

Following similar steps, the corresponding long DiD estimator would be given by

\begin{equation*}
\begin{aligned}
\widehat{ATT_{LD}(t)}      = 
     & = 
\sum_{\tau=2}^t \beta_{\tau} 
+ \left[\frac{1}{n}\sum_{i=1}^{n}\widehat{w^G_{i1t}} \left(\varepsilon_{it} - \varepsilon_{i1}\right) - \frac{1}{n}\sum_{i=1}^{n}\widehat{w^C_{i1t}} \left(\varepsilon_{it} - \varepsilon_{i1}\right)   \right], \\
\end{aligned}
\end{equation*}
if balanced panel data was available.

We will prove the asymptotic normality of these estimators in the general framework in Theorem \ref{theorem: asy;ptotic results}. In this proof, we only derive their asymptotic variance for the mentioned example.

\begin{equation*}
\begin{aligned}
E\left[ n\left(\widehat{ATT_{CD}(t)} - ATT(t)\right)^2 \right] & = n E\left[ \left( \sum_{\tau=1}^{t-1} \left[\frac{1}{n}\sum_{i=1}^{n}\left(\widehat{w^G_{i\tau\tau+1}} - \widehat{w^C_{i\tau\tau+1}} \right)\left(\varepsilon_{i\tau+1} - \varepsilon_{i\tau}\right)  \right]\right)^2 \right] \\
& =   \sum_{\tau=1}^{t-1} \frac{1}{n}\sum_{i=1}^{n} E\left[\left(\widehat{w^G_{i\tau\tau+1}} - \widehat{w^C_{i\tau\tau+1}} \right)^2\left(\varepsilon_{i\tau+1} - \varepsilon_{i\tau}\right)^2  \right]  \\
& =   \sum_{\tau=1}^{t-1} \frac{1}{n}\sum_{i=1}^{n} E\left[\left(\widehat{w^G_{i\tau\tau+1}} - \widehat{w^C_{i\tau\tau+1}} \right)^2 \right]E\left[\left(\varepsilon_{i\tau+1} - \varepsilon_{i\tau}\right)^2  \right],  \\
& =   \sum_{\tau=1}^{t-1} E\left[\left(\widehat{w^G_{i\tau\tau+1}} - \widehat{w^C_{i\tau\tau+1}} \right)^2 \right]E\left[\left(\varepsilon_{i\tau+1} - \varepsilon_{i\tau}\right)^2  \right],  \\
\end{aligned}
\end{equation*}
where the third equality follows from independence of $G$ and $\varepsilon_{i\tau+1} - \varepsilon_{i\tau}$. As $n\to \infty$, the weak law of large numbers implies $E\left[\widehat{w^G_{i\tau\tau+1}}^2 \right] \overset{p}{\to} \frac{1}{P(S_{ \tau}S_{ \tau+1}G=1)}$ because $S_{itt+1}a_i \in{0,1}$ and $E[S_{itt+1}a_iS_{jtt+1}a_j]=0$ for $i\neq j$. Therefore, the asymptotic variance for $n\to\infty$ is
\begin{equation*}
\begin{aligned}
E\left[ n\left(\widehat{ATT_{CD}(t)} - ATT(t)\right)^2 \right] &  =    \sum_{\tau=1}^{t-1}\left[\frac{1}{qp} + \frac{1}{q(1-p)}  \right]\left[(\rho-1)^2\sigma_\varepsilon^2 + \sigma_\eta^2\right]   \\
& =    \frac{(t-1)}{qp(1-p)}    \frac{2(1-\rho)}{1-\rho^2}\sigma_\eta^2 \\
& =    2\frac{(t-1)}{qp(1-p)}    \frac{1}{1+\rho}\sigma_\eta^2,
\end{aligned}
\end{equation*}
when assuming $P(S_{i \tau}S_{i \tau+1}) = q \in(0,1)$ for all $\tau,i$ and $P(G_i=1) = p \in(0,1)$. Note that this variance is bounded below by $\frac{1}{qp(1-p)}  (t-1) \sigma_\eta^2$
when $\rho=1$, and bounded above by
 $ \frac{1}{qp(1-p)}  2(t-1) \sigma_\eta^2$,
when $\rho=0$. Also, having a complete panel implies $q=1$ in this setting.

Similarly, the variance of the second estimator can be developed into

\begin{equation*}
\begin{aligned}
E\left[ n \left(\widehat{ATT_{CS}(t)} - ATT(t)\right)^2 \right]  =&   E\left[ \left( \widehat{w^G_{it}} - \widehat{w_{it}^C} \right)^2 \right] \sigma_{\varepsilon_t}^2 
+ E\left[ \left( \widehat{w_{i1}^G} - \widehat{w_{i1}^C} \right)^2 \right] \sigma_{\varepsilon_1}^2 \\
& +  E\left[  \left( \widehat{w^G_{i t}} - \widehat{w^C_{i t}}  \right)^2 \alpha_i^2  \right]  +  E\left[  \left(  \widehat{w^G_{i 1}} - \widehat{w^C_{i 1}}  \right)^2 \alpha_i^2  \right] \\
& -  E\left[  \left(  \widehat{w^G_{i 1}} - \widehat{w^C_{i 1}}  \right)\alpha_i  \right] E\left[\left(  \widehat{w^G_{i t}} - \widehat{w^C_{i t}}  \right) \alpha_i  \right] \\
 =&   E\left[ \left( \widehat{w^G_{it}} - \widehat{w_{it}^C} \right)^2 \right] \sigma_{\varepsilon_t}^2 
+ E\left[ \left( \widehat{w_{i1}^G} - \widehat{w_{i1}^C} \right)^2 \right] \sigma_{\varepsilon_1}^2   \\ & +  2E\left[  \left( \widehat{w^G_{i t}} - \widehat{w^C_{i t}}  \right)^2   \right]\sigma_\alpha^2, \\
  =&   \frac{1}{qp(1-p)}(0.5\sigma_{\varepsilon_t}^2 + \sigma_{\varepsilon_1}^2 + \sigma_\alpha^2), \\
\end{aligned}
\end{equation*}
where the second equality follows from the independence of outcomes and sampling as well as $G$ and $\alpha_i$. As $n\to \infty$, we have  $E\left[\widehat{w^G_{i\tau}}^2 \right] \overset{p}{\to} \frac{1}{P(S_{ \tau}G=1)} = \frac{1}{P(S_{ \tau}S_{ \tau+1}G=1)+P(S_{ \tau-1}S_{ \tau}G=1)} = \frac{1}{2pq}$ for $1<\tau<\mathcal{T}$, and $E\left[\widehat{w^G_{i\tau}}^2 \right] \overset{p}{\to} \frac{1}{P(S_{ \tau}G=1)} = \frac{1}{pq}$ for $\tau=1$ or $\tau=\mathcal{T}$, because we have two overlapping samples in each period except for the first and last.  For $t < \mathcal{T}$, the asymptotic variance as $n \to \infty$ is thus

\begin{equation*}
\begin{aligned}
E\left[n \left(\widehat{ATT_{CS}(t)} - ATT(t)\right)^2 \right]  \overset{p}{\to} &
  \frac{1}{qp(1-p)}\left(1.5  \sigma_{\varepsilon_1}^2 + \sigma_\alpha^2 +  (t-1)\sigma_\eta^2 \right), \quad \text{for $\rho=1$,}
\end{aligned}
\end{equation*}  
\begin{equation*}
\begin{aligned}
E\left[n \left(\widehat{ATT_{CS}(t)} - ATT(t)\right)^2 \right]  \overset{p}{\to} &
  \frac{1}{qp(1-p)}\left( \sigma_\alpha^2 +  1.5\frac{\sigma_\eta^2}{1-\rho^2} \right), \quad \text{for $\rho\in(0,1)$,}
\end{aligned}
\end{equation*}
and therefore
\begin{equation*}
\begin{aligned}
E\left[n \left(\widehat{ATT_{CS}(t)} - ATT(t)\right)^2 \right]  \overset{p}{\to} &
  \frac{1}{qp(1-p)}\left(\sigma_\alpha^2 +  1.5 \sigma_\eta^2  \right), \quad \text{for $\rho=0$.}
\end{aligned}
\end{equation*}

In the complete panel setting, the weights are slightly different, and similar computations give 
\begin{equation*}
\begin{aligned}
E\left[ n\left(\widehat{ATT_{CS}(t)} - ATT(t)\right)^2 \right] \overset{p}{\to}    \frac{1}{p(1-p)}   \left[(\rho^{t-1}-1)^2\sigma_{\varepsilon_1}^2 + \sigma_\eta^2 \sum_{\tau=0}^{t-2}\rho^{2\tau } + 2\sigma_{\alpha}^2\right],
\end{aligned}
\end{equation*}

Following similar steps, it is easy to show that
the asymptotic variance of the  long-DiD for $n\to\infty$ is
\begin{equation*}
\begin{aligned}
E\left[ n\left(\widehat{ATT_{LD}(t)} - ATT(t)\right)^2 \right] \overset{p}{\to}    \frac{1}{qp(1-p)}   \left[(\rho^{t-1}-1)^2\sigma_{\varepsilon_1}^2 + \sigma_\eta^2 \sum_{\tau=0}^{t-2}\rho^{2\tau }\right],
\end{aligned}
\end{equation*}
that is $\frac{1}{qp(1-p)}   2\sigma_\eta^2$ if $\rho=0$ and $\frac{1}{qp(1-p)}   (t-1)\sigma_\eta^2$ if $\rho=1$,
assuming the rotating samples consist of the same individuals over the entire time horizon. Therefore, the chained DiD estimator achieves the minimum variance, i.e. that of the long DiD estimator when idiosyncratic errors follow a random walk. The chained DiD estimator is therefore an efficient estimator in both the balanced and unbalanced panel settings when errors follow a random walk.

Comparing CD and CS variance for $\rho \in (0,1)$, we have that CD has smaller variance if $2(t-1)\sigma_\eta^2 /(1+\rho) \leq 1.5 \sigma_\eta^2/(1-\rho^2) +\sigma_\alpha^2$. If $\rho=0$, this condition becomes $\sigma_\eta^2/2 \leq \sigma_\alpha^2$ for $t=2$, and $5\sigma_\eta^2/2 \leq \sigma_\alpha^2$ for $t=6$.

 \end{proof}
 
\subsection{General framework}

\subsubsection{Proofs for rotating panel setting}
\begin{proof}[Proof of Theorem \ref{theorem: general identification}]
This proof focuses on the identification of parameters in the general framework with a rotating panel structure. Let us define $ATT_X(g,\tau) = E[Y_{\tau}(1) - Y_{\tau}(0)|X,G_g = 1]$ to write its first-difference as 
\begin{equation*}
\begin{aligned}
\Delta ATT_X(g,\tau) & = ATT_X(g,\tau) - ATT_X(g,\tau-1) \\
& = E[Y_{\tau}(1) - Y_{\tau}(0)|X,G_g = 1] - E[Y_{\tau-1}(1) - Y_{\tau-1}(0)|X,G_g = 1] \\
& =  E[Y_{\tau} - Y_{\tau-1}|X,G_g = 1,S_{\tau,\tau-1}=1] - E[Y_{\tau} - Y_{\tau-1}|X,C = 1,S_{\tau,\tau-1}=1] \\
& =  A_X(g,\tau) - B_X(g,\tau) \\
\end{aligned}
\end{equation*}
where the third equality follows from the conditional parallel trends assumption and the sampling independence. Proofs of Corollary \ref{corrolary:attritionmodel1}] and \ref{corrolary:attritionmodel2}]  show how to adjust for sample selection on observables.

We can use the above expression to develop $ATT(g,t)$ into
\begin{equation}
\begin{aligned}
ATT(g,t) & =  E\left( E[Y_{t}(1) - Y_{t}(0)|X,G_g = 1] | G_g =1  \right)  \\
& =  E\left(ATT_X(g,t) | G_g =1  \right)  \\
& =  E\left( \sum_{\tau=g}^{t} \Delta ATT_X(g,\tau) | G_g =1  \right)  \\
& =  \sum_{\tau=g}^{t} E\left(  \Delta ATT_X(g,\tau) | G_g =1  \right)  \\
& =  \sum_{\tau=g}^{t} E\left(  A_X(g,\tau) - B_X(g,\tau)  | G_g =1  \right)  \\
\end{aligned}
\end{equation}
with 
\begin{equation}
\begin{aligned}
E\left(  A_X(g,\tau) | G_g =1 \right) & = E\left( Y_{\tau} - Y_{\tau-1} |G_g \right)   \\
& = E\left( E[Y_{\tau} - Y_{\tau-1}|X,G_g ] |G_g \right)   \\
& = E\left( E[Y_{\tau} - Y_{\tau-1}|X,G_g,S_{\tau}  S_{\tau-1} ] |G_g \right)   \\
& = E\left(  E[ \frac{S_{\tau}  S_{\tau-1}}{E[S_{\tau}  S_{\tau-1}|X,G_g]} (Y_{\tau} - Y_{\tau-1})|X,G_g ] |G_g \right)   \\
& = E\left(  E[ \frac{S_{\tau}  S_{\tau-1}}{E[S_{\tau}  S_{\tau-1}|G_g]} (Y_{\tau} - Y_{\tau-1})|X,G_g ] |G_g \right)   \\
& = E\left(   \frac{S_{\tau}  S_{\tau-1}}{E[S_{\tau}  S_{\tau-1}|X,G_g]} (Y_{\tau} - Y_{\tau-1}) |G_g \right)   \\
& = E_M \big( \left( Y_{\tau} - Y_{\tau-1} \right)\frac{G_g S_{\tau}  S_{\tau-1}}{E[S_{\tau}  S_{\tau-1}|G_g]E[G_g]} \big) \\
\end{aligned}
\end{equation}
by the law of iterated expectations and the definition of $F_M$. Following the proofs of Theorem 1 and B.1 of \cite{callaway2021difference}, the second term can be developed into\footnote{We alleviate notations by dropping $=1$ from conditioning sets.}

\begin{equation}
\begin{aligned}
& E (   B_X(g,\tau) | G_g =1)  = E ( E[Y_{\tau} - Y_{\tau-1}|X,C ,S_{\tau,\tau-1}]|G_g  )   \\
& =  E ( E[\frac{S_{\tau}  S_{\tau-1}}{E[S_{\tau}  S_{\tau-1}|X,C]}(Y_{\tau} - Y_{\tau-1})|X,C]|G_g  )   \\
& =  E ( E[\frac{S_{\tau}  S_{\tau-1}}{E[S_{\tau}  S_{\tau-1}|C]}(Y_{\tau} - Y_{\tau-1})|X,C]|G_g  )   \\
& = E\left( E[ \frac{S_{\tau}  S_{\tau-1}}{E[S_{\tau}  S_{\tau-1}|C]}\frac{C }{1-P(G_g=1|X,G_g+C)} (Y_{\tau} - Y_{\tau-1})|X,G_g+C=1 ]|G_g \right)   \\
& = \frac{E\left( G_g E[ \frac{S_{\tau}  S_{\tau-1}}{E[S_{\tau}  S_{\tau-1}|C]} \frac{C}{1-P(G_g|X,G_g+C)} (Y_{\tau} - Y_{\tau-1})|X,G_g+C ]|G_g+C  \right)}{{P(G_g=1|G_g+C)}},   \\
\end{aligned}
\end{equation}
where using the definition of $p_g$ yields
\begin{equation}
\begin{aligned}
... & = \frac{E\left( G_g E[ \frac{S_{\tau}  S_{\tau-1}}{E[S_{\tau}  S_{\tau-1}|C]} \frac{C}{1-p_g(X)} (Y_{\tau} - Y_{\tau-1})|X,G_g+C ]|G_g+C  \right)}{{P(G_g=1|G_g+C)}},   \\
& = \frac{E\left( E[ \frac{S_{\tau}  S_{\tau-1}}{E[S_{\tau}  S_{\tau-1}|C]} \frac{p_g(X)C}{1-p_g(X)} (Y_{\tau} - Y_{\tau-1})|X,G_g+C ]|G_g+C  \right)}{{E(G_g|G_g+C)}},   \\
& = \frac{E \left((G_g+C) E[ \frac{S_{\tau}  S_{\tau-1}}{E[S_{\tau}  S_{\tau-1}|C]} \frac{p_g(X)C}{1-p_g(X)} (Y_{\tau} - Y_{\tau-1})|X,G_g+C ]  \right)}{{E(G_g|G_g+C)E(G_g+C)}},   \\
& = \frac{E \left((G_g+C) E[ \frac{S_{\tau}  S_{\tau-1}}{E[S_{\tau}  S_{\tau-1}|C]} \frac{p_g(X)C}{1-p_g(X)} (Y_{\tau} - Y_{\tau-1})|X,G_g+C ]  \right)}{{E(G_g)}},   \\
& = \frac{E\left( E[(G_g+C) |X] E[ \frac{S_{\tau}  S_{\tau-1}}{E[S_{\tau}  S_{\tau-1}|C]}\frac{p_g(X) C}{1-p_g(X)} (Y_{\tau} - Y_{\tau-1})|X,G_g+C]   \right)}{{E(G_g)}}   \\
& = \frac{E\left( E[ \frac{S_{\tau}  S_{\tau-1}}{E[S_{\tau}  S_{\tau-1}|C]}\frac{p_g(X) C}{1-p_g(X)} (Y_{\tau} - Y_{\tau-1})|X]  \right)}{{E(G_g )}}   \\
& = \frac{E\left(  \frac{S_{\tau}  S_{\tau-1}}{E[S_{\tau}  S_{\tau-1}|C]}\frac{p_g(X) C}{1-p_g(X)} (Y_{\tau} - Y_{\tau-1})  \right)}{{E(G_g )}}   \\
& = \frac{E\left(  \frac{E(G_g S_{\tau}  S_{\tau-1} )}{E[S_{\tau}  S_{\tau-1}|C]E(G_g )}\frac{p_g(X) C S_{\tau}  S_{\tau-1}}{1-p_g(X)} (Y_{\tau} - Y_{\tau-1})  \right)}{{E(G_g S_{\tau}  S_{\tau-1} )}}   \\
\end{aligned}
\end{equation}
with
\begin{equation}
\begin{aligned}
E_M\left(  \frac{p_g(X) C S_{\tau,\tau-1}}{1-p_g(X)} \right) & =  E_M(   \frac{E_M(G_g|X,G_g+C=1)}{E_M(C|X,G_g+C=1)}C S_{\tau,\tau-1} ) \\
& =  E_M(   \frac{E_M(G_g|X)E_M(CS_{\tau,\tau-1}|X)}{E_M(C|X)}  )         \\
& =  E_M(   \frac{E_M(G_g|X)E(S_{\tau,\tau-1}|C,X)E_M(C|X)}{E_M(C|X)}  )         \\
& =  E_M(   E_M(G_g|X)E(S_{\tau,\tau-1}|C,X)  )         \\
    & =  E( S_{\tau,\tau-1}|C  )E_M(G_g).         \\
\end{aligned}
\end{equation}

Finally, the proof that identification is not guaranteed with the repeated cross-sectional estimator (cross-section DiD) presented in Appendix B of \cite{callaway2021difference} follows from taking a counterexample. Under Assumption \ref{ass: sampling on trends 2}, it is possible that $E[Y_t|X,C=1,S_t=1] = E[Y_t|X,C=1]$ but $E[Y_t|X,G_g=1,S_t=1] = E[Y_t|X,G_g=1] + \alpha  t$ with $\alpha >0$. Following the steps of the proof it is is easy to show that the cross-section DiD  identifies
\begin{equation*}
    ATT(g,t) + \alpha(t-g+1).
\end{equation*}

\end{proof}

\begin{proof}[Proof of Theorem \ref{theorem: asy;ptotic results}]
This proof is adapted from \cite{callaway2021difference} 's Theorem 2. We proceed in 3 steps.

\paragraph{Parametric propensity scores and notations.} First, we introduce additional notations and explain Assumption 5 in \cite{callaway2021difference}  about the estimation of propensity scores.\footnote{This is a standard assumption in the literature so it is not reproduced here.} Let $\mathcal{W}_{i} = (Y_{it},Y_{it+1},X_i,G_{i1},G_{i2},...,,G_{i\mathcal{T}},C_i)'$ denote the data for an individual $i$ observed in $t$ and $t+1$.  Assumption 5 in \cite{callaway2021difference}  assumes that the propensity scores,  parametrized as $p_g(X_i) = \Lambda (X_i'\pi_g^0)$ with  $\Lambda (\cdot)$ being a known function (logit or probit), can be parametrically estimated by maximum likelihood. We denote $\hat{p}_g(X_i) = \Lambda (X_i'\hat{\pi}_g)$ where $\hat{\pi}_g$ are estimated by ML, $\dot{p}_g = \partial p_g(u)/\partial u$, and  $\dot{p}_g(X) = \dot{p}_g(X_i'\pi_g^0)$. Under this assumption, the estimated parameter $\hat{\pi}_g$ is asymptotically linear, i.e.,
\begin{equation*}
\sqrt{n}(\hat{\pi}_g - \pi_g^0) = \frac{1}{\sqrt{n}} \sum_i \xi_{g}^{\pi}(\mathcal{W}_i)  + o_p(1),
\end{equation*}
where $\xi_{g}^{\pi}(\mathcal{W}_i)$ is defined in (3.1) is \cite{callaway2021difference}  and does not depend on the sampling process since $X$ is observed for all individuals.

Let us now define,
\begin{equation}\label{eq:psi's}
\psi_{gt}(\mathcal{W}_i) = \psi_{gt}^G(\mathcal{W}_i) + \psi_{gt}^G(\mathcal{W}_i),
\end{equation}
where
\begin{equation*}
\begin{aligned}
\psi_{gt}^G(\mathcal{W}_i) = & w^G_{it,t-1}(g)\left[(Y_{it}-Y_{it-1}) - E_M\left[w^G_{it,t-1}(g)(Y_{it}-Y_{it-1})   \right]   \right], \\
\psi_{gt}^C(\mathcal{W}_i) = & w^C_{it,t-1}(g,X)\left[(Y_{it}-Y_{it-1}) - E_M\left[w^C_{it,t-1}(g,X)(Y_{it}-Y_{it-1})   \right]   \right] + M_{gt}'\xi^{\pi}_g(\mathcal{W}_i),
\end{aligned}
\end{equation*}
and 
\begin{equation*}
\begin{aligned}
M_{gt} = \frac{E_M\left[ X (\frac{C S_t S_{t-1}}{1-p_g(X)})^2 \dot{p}_g(X)\left[(Y_{it}-Y_{it-1}) - E_M\left[w^C_{it,t-1}(g,X)(Y_{it}-Y_{it-1})   \right]   \right] \right]}{E_M[\frac{p_g(X)C}{1-p_g(X)}]}.
\end{aligned}
\end{equation*}
is a $k$ dimensional vector, $k$ being the number of covariates in $X$.
Finally, let $\widehat{\Delta ATT}_{g \leq t}$ and $\Delta ATT_{g \leq t}$ denote the vectors of all $\widehat{\Delta ATT}(g,t)$ and $\Delta ATT(g,t)$ for any $2 \leq g \leq t \leq \mathcal{T}$. Similarly,  the collection of $\psi_{gt}$ across all periods and groups such that $g\leq t$ is denoted by $\Psi_{g\leq t}$.

\paragraph{Asymptotic result for $\Delta ATT$.} Second, we show the asymptotic result for $\Delta ATT$. Recall that
\begin{equation*}
\begin{aligned}
 \widehat{ATT}(g,t)  =  \sum_{\tau=g}^t  \widehat{\Delta ATT}(g,\tau),
\end{aligned}
\end{equation*}
where 
\begin{equation*}
\begin{aligned}
\widehat{\Delta ATT}(g,\tau) & = \hat{E}_M \left[ \frac{G_{g} S_{\tau-1}S_{\tau}}{ \hat{E}_M [ G_{g} S_{\tau-1}S_{\tau}]}(Y_{\tau}-Y_{\tau-1})\right] - \hat{E}_M \left[ \frac{\frac{p_g(X)C S_{\tau-1}S_{\tau}}{1-p_g(X)}} {\hat{E}_M \left[\frac{p_g(X)C S_{\tau-1}S_{\tau}}{1-p_g(X)}\right] }(Y_{\tau}-Y_{\tau-1})\right] \\
& = \widehat{\Delta ATT}_g(g,\tau) - \widehat{\Delta ATT}_C(g,\tau),
\end{aligned}
\end{equation*}
where $\hat{E}$ denotes the empirical mean. We will show separately that, for all for all $2 \leq g \leq t \leq \mathcal{T}$,
\begin{equation}\label{eq:att g}
\sqrt{n}\left( \widehat{\Delta ATT}_{g}(g,t)-\Delta ATT_{g}(g,t) \right)= \frac{1}{\sqrt{n}} \sum_i \psi_{gt}^G(\mathcal{W}_i)  + o_p(1),
\end{equation}
and
\begin{equation}\label{eq:att c}
\sqrt{n}\left( \widehat{\Delta ATT}_{C}(g,t)-\Delta ATT_{C}(g,t) \right)= \frac{1}{\sqrt{n}} \sum_i \psi_{gt}^C(\mathcal{W}_i)  + o_p(1),
\end{equation}
which together implies 
\begin{equation}\label{eq:att g-c}
\sqrt{n}\left( \widehat{\Delta ATT}(g,t)-\Delta ATT(g,t) \right)= \frac{1}{\sqrt{n}} \sum_i \psi_{gt}(\mathcal{W}_i)  + o_p(1)
\end{equation}
 and the asymptotic normality of $\sqrt{n}\left( \widehat{\Delta ATT}_{g \leq t}-\Delta ATT_{g \leq t} \right)$ follows from the multivariate central limit theorem.
 
First, we show \eqref{eq:att g}. Let $\beta_g = E_M[G_g S_{\tau-1}S_{\tau}]$ and $\hat{\beta}_g = \hat{E}_M[G_g S_{\tau-1}S_{\tau}]$ and note that
\begin{equation*}
\begin{aligned}
\sqrt{n}\left(\hat{\beta}_g-\beta_g \right) = \frac{1}{\sqrt{n}}\sum_i \left(G_{ig} S_{i\tau-1}S_{i\tau} - E[G_g S_{\tau-1}S_{\tau}] \right) \overset{p}{\to} 0, \; as \; n  \to +\infty.
\end{aligned}
\end{equation*}
Then, for all $2 \leq g \leq t \leq \mathcal{T}$,
\begin{equation*}
\begin{aligned}
\sqrt{n}( \widehat{\Delta ATT}_{g}(g,t)-& \Delta ATT_{g}(g,t) )  =  \sqrt{n} \hat{E}_M \left[ \frac{G_{g} S_{t,t-1}}{\hat{\beta}_g}(Y_{t}-Y_{t-1})\right] - \sqrt{n} E_M \left[ \frac{G_{g} S_{t,t-1}}{\beta_g}(Y_t-Y_{t-1})\right] \\
= & \frac{\sqrt{n}}{\hat{\beta}_g} ( \hat{E}_M \left[ G_{g} S_{t,t-1}(Y_{t}-Y_{t-1})\right] - \frac{\hat{\beta}_g}{\beta_g}E_M \left[ G_{g} S_{t,t-1}(Y_t-Y_{t-1})\right] ) \\
= & \frac{\sqrt{n}}{\hat{\beta}_g} ( \hat{E}_M \left[ G_{g} S_{t,t-1}(Y_{t}-Y_{t-1})\right] - E_M \left[ G_{g} S_{t,t-1}(Y_t-Y_{t-1})\right] ), \\ 
\end{aligned}
\end{equation*}
and by the continuous mapping theorem,
\begin{equation*}
\begin{aligned}
\sqrt{n}( \widehat{\Delta ATT}_{g}(g,t)-& \Delta ATT_{g}(g,t) ) = \frac{\sqrt{n}}{\beta_g} ( \hat{E}_M \left[ G_{g} S_{t,t-1}(Y_{t}-Y_{t-1})\right] - E_M \left[ G_{g} S_{t,t-1}(Y_t-Y_{t-1})\right] ) \\ 
& - \sqrt{n}\left( \frac{1}{\beta_g}-\frac{1}{\hat{\beta}_g} \right)E_M \left[ G_{g} S_{t,t-1}(Y_t-Y_{t-1})\right] + o_p(1) \\
= & \frac{\sqrt{n}}{\beta_g} ( \hat{E}_M \left[ G_{g} S_{t,t-1}(Y_{t}-Y_{t-1})\right] - E_M \left[ G_{g} S_{t,t-1}(Y_t-Y_{t-1})\right] ) \\ 
& -  \frac{\sqrt{n}\left(\hat{\beta}_g-\beta_g\right)}{\beta_g^2} E_M \left[ G_{g} S_{t,t-1}(Y_t-Y_{t-1})\right] + o_p(1) \\
= & \frac{\sqrt{n}}{\beta_g} ( \hat{E}_M \left[ G_{g} S_{t,t-1}(Y_{t}-Y_{t-1})\right] - \frac{\hat{\beta}_g}{\beta_g}  E_M \left[ G_{g} S_{t,t-1}(Y_t-Y_{t-1})\right] )  + o_p(1) \\
= & \frac{1}{\sqrt{n}} \sum_i G_{ig}\frac{  S_{it,t-1}(Y_{it}-Y_{it-1}) -  \Delta ATT(g,t)}{\beta_g}  + o_p(1) \\
= & \frac{1}{\sqrt{n}} \sum_i w^G_{it,t-1}(g)\left[(Y_{it}-Y_{it-1}) - E_M\left[w^G_{t,t-1}(g)(Y_{t}-Y_{it-1})   \right]   \right]  + o_p(1) \\
= & \frac{1}{\sqrt{n}} \sum_i \psi_{gt}^G(\mathcal{W}_i)  + o_p(1), \\
\end{aligned}
\end{equation*}
proving \eqref{eq:att g}. 

Let us now turn to \eqref{eq:att c}. For an arbitrary function $g$, let 
\begin{equation*}
    w_{t}(g) = \frac{g(X)CS_{t,t-1}}{1-g(X)}
\end{equation*}
and note that
\begin{equation*}
\begin{aligned}
\sqrt{n}( \widehat{\Delta ATT}_{C}(g,t)-& \Delta ATT_{C}(g,t) )  =  \sqrt{n} \left( \hat{E}_M \left[  \frac{w_{t}(\hat{p}_g)}{\hat{E}_M [  w_{t}(\hat{p}_g)]} (Y_{t}-Y_{t-1})\right] -  E_M \left[  \frac{w_{t}(p_g)}{E_M [  w_{t}(p_g)]} (Y_{t}-Y_{t-1})\right] \right) \\
  = &  \frac{\sqrt{n}}{\hat{E}_M [  w_{t}(\hat{p}_g)]} \left( \hat{E}_M \left[  w_{t}(\hat{p}_g) (Y_{t}-Y_{t-1})\right] -  \frac{\hat{E}_M [  w_{t}(\hat{p}_g)]}{E_M [  w_{t}(p_g)]}  E_M \left[  w_{t}(p_g) (Y_{t}-Y_{t-1})\right] \right) \\
  = &  \frac{\sqrt{n}}{\hat{E}_M [  w_{t}(\hat{p}_g)]} \left( \hat{E}_M \left[  w_{t}(\hat{p}_g) (Y_{t}-Y_{t-1})\right] -  E_M \left[  w_{t}(p_g) (Y_{t}-Y_{t-1})\right] \right) \\
 & -    \frac{ E_M \left[  w_{t}(p_g) (Y_{t}-Y_{t-1})\right]}{\hat{E}_M [  w_{t}(\hat{p}_g)]E_M [  w_{t}(p_g)]} \sqrt{n} (\hat{E}_M \left[ w_{t}(\hat{p}_g) \right]  - \hat{E}_M \left[ w_{t}(p_g) \right]) \\
 = &  \frac{1 }{\hat{E}_M [  w_{t}(\hat{p}_g)]}\sqrt{n} A_n(\hat{p}_g) -  \frac{\Delta ATT_C(g,t)}{\hat{E}_M [  w_{t}(\hat{p}_g)]}\sqrt{n} B_n(\hat{p}_g) \\
  = &  \frac{1 }{E_M [  w_{t}(p_g)]}\sqrt{n} A_n(\hat{p}_g) -  \frac{\Delta ATT_C(g,t)}{E_M [  w_{t}(p_g)]}\sqrt{n} B_n(\hat{p}_g) +o_p(1), \\
\end{aligned}
\end{equation*}

where the last equality follows directly from Assumption 5, which implies Lemma A.2 and Lemma A.3 in \cite{callaway2021difference}.
Applying the mean value theorem yields
\begin{equation*}
\begin{aligned}
A_n(\hat{p}_g) = & \hat{E}_M \left[  w_{t}(p_g) (Y_{t}-Y_{t-1})\right] - E_M \left[  w_{t}(p_g) (Y_{t}-Y_{t-1})\right] \\
& + \hat{E}_M \left[X\frac{C S_{t,t-1}}{(1-p_g(X;\overline{\pi}))^2} \dot{p}_g(X;\overline{\pi}) (Y_{t}-Y_{t-1}) \right]'\left(\hat{\pi}_g-\pi^0_g \right),
\end{aligned}
\end{equation*}
where $\overline{\pi}$ is an intermediate point that satisfies $|\overline{\pi}_g\pi_g^0|\leq |\hat{\pi}_g\pi_g^0|$ a.s. Thus, by Assumption 5, the previously mentioned Lemmas, and the Classical Glivenko-Cantelli's theorem, 
\begin{equation*}
\begin{aligned}
A_n(\hat{p}_g) = & \hat{E}_M \left[  w_{t}(p_g) (Y_{t}-Y_{t-1})\right] - E_M \left[  w_{t}(p_g) (Y_{t}-Y_{t-1})\right] \\
& + \hat{E}_M \left[X\frac{C S_{t,t-1}}{(1-p_g(X))^2} \dot{p}_g(X) (Y_{t}-Y_{t-1}) \right]'\left(\hat{\pi}_g-\pi^0_g \right) + o_p(n^{-1/2}),
\end{aligned}
\end{equation*}
and using the same reasoning we obtain
\begin{equation*}
\begin{aligned}
B_n(\hat{p}_g) = & \hat{E}_M \left[  w_{t}(p_g) \right] - E_M \left[  w_{t}(p_g) \right] \\
& + \hat{E}_M \left[X\frac{C S_{t,t-1}}{(1-p_g(X))^2} \dot{p}_g(X) \right]'\left(\hat{\pi}_g-\pi^0_g \right) + o_p(n^{-1/2}).
\end{aligned}
\end{equation*}
Combining the above results and making use of the same Lemma yields \eqref{eq:att g-c} hence concludes the proof for $\Delta ATT$. The asymptotic covariance is given by $\Sigma_{\Delta} = E \left[\Psi_{g\leq\tau}(\mathcal{W}_i) \Psi_{g\leq\tau}(\mathcal{W}_i)' \right]$.

\paragraph{Asymptotic result for $ATT$.} 
Finally, by making use of \eqref{eq:att g-c} we have that 
\begin{equation*}
\begin{aligned}
 \sqrt{n}\left( \widehat{ ATT}(g,t)- ATT(g,t) \right)  =  &  \sum_{\tau=g}^t  \sqrt{n}\left( \widehat{\Delta ATT}(g,\tau)-\Delta ATT(g,\tau) \right) \\
= &   \frac{1}{\sqrt{n}} \sum_i \left[\sum_{\tau=g}^t \psi_{g\tau}(\mathcal{W}_i)\right]  + o_p(1) \\
& \overset{d}{\to}N(0,\Sigma),
\end{aligned}
\end{equation*}
where $\Sigma = E_M \left[\Phi_{g\leq\tau}(\mathcal{W}_i)\Phi_{g\leq\tau}(\mathcal{W}_i)' \right]$ with $ \Phi_{g\leq\tau}(\mathcal{W}_i) = \sum_{\tau=g}^t \Psi_{g\leq\tau}(\mathcal{W}_i)$.

Therefore, the influence function of the chained DiD estimator corresponds to the sum of influence functions of the short-term DiD estimator.

\end{proof}

\subsubsection{Bootstrap implementation for rotating panel setting}\label{sec:bootstrap}

    \paragraph{Bootstrapped confidence bands for $\widehat{ATT}(g,t)$.} The algorithm is as follows:
\begin{enumerate}
    \item Draw a vector of $\mathbf{V^{b}}=(V_1,...,V_i,...,V_n)'$, where $V_i$'s are iid zero mean random variables with unit variance, such as Bernoulli random variables with $Pr(V=1-\kappa)=\kappa/\sqrt{5}$ with $\kappa = (\sqrt{5}+1)/2$ as suggested by Mammen (1993).
    \item Compute the bootstrap draw $\widehat{ATT}_{g \leq t}^{\star b} =\widehat{ATT}_{g \leq t} + \hat{\Phi}_{g\leq\tau} \mathbf{V^{b}}  $ where $\hat{\Phi}_{g\leq\tau}$ is a consistent estimator of $\Phi_{g\leq\tau}$ (see below).
    \item Compute $\hat{R}^{\star b}(g,t) = \sqrt{n}(\widehat{ATT}^{\star b}(g,t)-\widehat{ATT}(g,t))$ for each element of the vector $\widehat{ATT}_{g \leq t}^{\star b}$.
    \item Repeat steps 1-3 $B$ times. Note: do not re-estimate propensity scores and parameters for each draw.
    \item Compute the bootstrapped covariance for each $(g,t)$ as $\hat{\Sigma}^{1/2}(g,t) = (q_{0.75}(g,t)-q_{0.25}(g,t))/(z_{0.75}-z_{0.25})$, where $q_p(g,t)$ is the $p^{th}$ sample quantile of $\hat{R}^{\star}$ (across B draws) and $z(g,t)$ is the $p^{th}$ quantile the standard normal distribution.
    \item For each $b$, compute $\text{t-stat}_{g\leq t}^b = \max_{(g,t)} |\hat{R}^{\star b}(g,t) |\hat{\Sigma}^{-1/2}(g,t)$.
    \item Construct $\hat{c}_{1-\alpha}$ as the empirical $(1-\alpha)$ quantile of the B boostrap draws of $\text{t-stat}_{g\leq t}^b$.
    \item Construct the bootstrapped simultaneous confidence band for $ATT(g,t)$ as $\hat{C}(g,t) = [\widehat{ATT}(g,t)\pm \hat{c}_{1-\alpha}\hat{\Sigma}^{-1/2}(g,t)/\sqrt{n}]$.
\end{enumerate}

This procedure requires to compute $\hat{\Phi}_{g\leq\tau}$, represented here as a $K\times n$ matrix, with $n$ being the number of observations and $K = \frac{\mathcal{T}(\mathcal{T}-1)}{2}$ being the number of $(g,t)$ element for any $2\leq g \leq t \leq \mathcal{T}$. This is done as follows:
\begin{enumerate}
    \item For every $(g,t)$, compute the n-dimensional vector $\psi_{gt}$ with $i^{th}$ element defined as $\psi_{gt}(i) = \psi_{gt}^G(i) + \psi_{gt}^C(i)$, where
    \begin{equation*}
\begin{aligned}
\psi_{gt}^G(i) = & w^G_{it,t-1}(g)\left[(Y_{it}-Y_{it-1}) - E_M\left[w^G_{t,t-1}(g)(Y_{t}-Y_{t-1})   \right]   \right], \\
\psi_{gt}^C(i) = & w^C_{it,t-1}(g,X)\left[(Y_{it}-Y_{it-1}) - E_M\left[w^C_{t,t-1}(g,X)(Y_{t}-Y_{t-1})   \right]   \right] + M_{gt}'\xi^{\pi}_g(i),
\end{aligned}
\end{equation*}
where $M_{gt}$, a $k$ dimensional vector ($k$ being the number of covariates in $X$), is defined as
\begin{equation*}
\begin{aligned}
M_{gt} = \frac{E\left[ X (\frac{C S_t S_{t-1}}{1-p_g(X)})^2 \dot{p}_g(X)\left[(Y_{it}-Y_{it-1}) - E\left[w^C_{t,t-1}(g,X)(Y_{t}-Y_{it-1})   \right]   \right] \right]}{E[\frac{p_g(X)C}{1-p_g(X)}]},
\end{aligned}
\end{equation*}
with $\hat{p}_g(X_i) = \Lambda (X_i'\hat{\pi}_g)$ being the parametric propensity score for covariates $X_i$ and $\dot{p}_g(X) = \partial \Lambda (X_i'\hat{\pi}_g)/\partial (X_i'\hat{\pi}_g)$. Furthermore,  $\xi^{\pi}_g(i)$ is a k-dimensional vector for each observation $i$, and is given by
\begin{equation*}
\begin{aligned}
\xi^{\pi}_g(i) = E_M\left[\frac{(G_g+C)\dot{p}_g(X)^2}{p_g(X_i)(1-p_g(X_i))}XX'\right]^{-1}X_i \frac{(G_g+C)(G_g-p_g(X_i))\dot{p}_g(X_i)}{p_g(X_i)(1-p_g(X_i))}.
\end{aligned}
\end{equation*}
\item Compute $\phi_{gt} = \sum_{\tau=g}^t \psi_{g\leq \tau}$ for all $2\leq g \leq t \leq \mathcal{T}$.
\item Concatenate all $\phi_{gt}$'s into a $K\times n$ matrix $\Phi_{g\leq t}$.
\end{enumerate}


\subsubsection{Summary parameters for rotating panel setting}\label{sec:summaryparam}
The group-time average treatment effect $ATT(g,t)$ consists of a building-block to study the dynamic effect of a treatment on different cohorts of treated individuals. In most applications, the main causal parameters of interest are not the $ATT(g,t)$ themselves but aggregate parameters of these building-blocks. In this section, we briefly mention the three main parameters of interest as proposed in \cite{callaway2021difference}, and show how the asymptotic results and multiplier bootstrap adapt to our setting.

\paragraph{Selective timing.}
The causal effect of a policy on the cohort treated in $g$ is given by

\begin{equation*}
    \theta_{S}(g) =  \frac{1}{\mathcal{T}-g+1}\sum_{t=g}^{\mathcal{T}}ATT(g,t),
\end{equation*}

and thus, an average causal effect across groups can be written as
\begin{equation*}
    \theta_{S}= \sum_{g=2}^{\mathcal{T}}\theta_{S}(g)Pr(G=g).
\end{equation*}

\paragraph{Dynamic treatment.}

In the presence of dynamic effects, the researcher may be interested in accounting for the length of exposure to the treatment. The causal effects of an exposure length $e \in \{0,1,2,... \}$ across groups is defined as
\begin{equation*}
    \theta_{D}(e) = \sum_{g=2}^{\mathcal{T}}\sum_{t=g+e}^{\mathcal{T}} ATT(g,t)Pr(G=g|t=g+e),
\end{equation*}
and therefore an average across exposure lengths is given by
\begin{equation*}
     \theta_{D} = \frac{1}{\mathcal{T}-1}\sum_{e=1}^{\mathcal{T}-1} \theta_{D}(e).
\end{equation*}

\paragraph{Calendar time.}

In some applications, the researcher may be interested in how treatment effects differ with calendar time. Let us consider 
\begin{equation*}
    \theta_{C}(t) = \sum_{g=2}^{t}ATT(g,t)Pr(G=g|g\leq t),
\end{equation*}
and therefore an average across exposure lengths is given by
\begin{equation*}
    \theta_{C} =   \frac{1}{\mathcal{T}-1}\sum_{t=2}^{\mathcal{T}}\theta_{C}(t),
\end{equation*}

The difference between $ \theta_{S}$, $ \theta_{D}$ and $ \theta_{C}$ is that the second and third attribute more weight to the groups with, respectively, longer exposure lengths, and treated in the earliest periods.

The asymptotic results and bootstrap procedure directly apply to the summary parameters. The following corollary summarizes these results.

\begin{corollary}\label{corollary: summary parameters } Under the Assumptions of Theorem \ref{theorem: asy;ptotic results}, for all parameters $\theta$ defined above, including those indexed by some variable, we have 
\begin{equation*}
    \sqrt{n}(\hat{\theta} - \theta) \overset{d}{\to} N(0,\Sigma_{\theta}),
\end{equation*}
as $n\to \infty$, where  $\Sigma_{\theta}$ is defined in the proof and the bootstrap procedure is defined.
\end{corollary}

\begin{proof}[Proof of Corollary \ref{corollary: summary parameters }]
All summary parameters defined in the text can be generically written as
\begin{equation*}
    \theta = \sum_{g=2}^{\mathcal{T}} \sum_{t=2}^{\mathcal{T}} w_{gt} ATT(g,t),
\end{equation*}
where $w_{gt}$ are some random weights. Estimators can be defined as
\begin{equation*}
    \hat{\theta} = \sum_{g=2}^{\mathcal{T}} \sum_{t=2}^{\mathcal{T}} \hat{w}_{gt} \widehat{ATT}(g,t),
\end{equation*}

where estimated weights are such that 
\begin{equation*}
    \sqrt{n}(\hat{w}_{gt} - w_{gt}) = \frac{1}{\sqrt{n}}\sum_{i=1}^n \xi_{gt}^w(\mathcal{W}_i) + o_p(1),
 \end{equation*}
with first and second moments given by $E[\xi_{gt}^w(\mathcal{W})] = 0$ and $E[\xi_{gt}^w(\mathcal{W})\xi_{gt}^w(\mathcal{W})']$ finite and positive definite. This condition is satisfied by the sample analogs of weights appearing in the summary parameters $\theta$'s presented in the main text.
The application of Theorem \ref{theorem: asy;ptotic results} yields
\begin{equation*}
\begin{aligned}
    \sqrt{n}(\hat{\theta}-\theta)  = & \frac{1}{\sqrt{n}}\sum_{i=1}^n l^w(\mathcal{W}_i) + o_p(1) \\
    & \overset{d}{\to} N(0,E[l^w(\mathcal{W})^2])
    \end{aligned}
 \end{equation*}
 as $n\to \infty$, and where
\begin{equation*}
\begin{aligned}
    l^w(\mathcal{W}_i)  = & \sum_{g=2}^{\mathcal{T}} \sum_{t=2}^{\mathcal{T}}  \left\{ w_{gt}\sum_{\tau=g}^t\psi_{g\tau}(\mathcal{W}_i)   + \xi_{gt}^w(\mathcal{W}_i) ATT(g,t) \right\},
    \end{aligned}
 \end{equation*}
 for $\psi_{gt}$ defined in \eqref{eq:psi's} and $\xi_{gt}^w$ correspond to the estimation errors of weights. The same bootstrap procedure can hence be used for $\hat{\theta}$ using a consistent estimate of the influence function $l^w$.
 \end{proof}

\subsubsection{Bootstrap for summary parameters}\label{bootstrapsummary} All estimators of summary parameters defined in the text can be generically written as
\begin{equation*}
    \hat{\theta} = \sum_{g=2}^{\mathcal{T}} \sum_{t=2}^{\mathcal{T}} \hat{w}_{gt} \widehat{ATT}(g,t),
\end{equation*}
where the weights $\hat{w}_{gt}$'s are possibly random. In simple settings they are not. For instance, let us consider
\begin{equation*}
    \theta_s(h) = \sum_{g=2}^T\sum_{t=2}^T w_{gt} ATT(g,t),
 \end{equation*}
 where $w_{gt} = \frac{1}{T-g+1}$ for $t\geq g$ and $g=h$, and $0$ otherwise.  The algorithm is as follows:
\begin{enumerate}
    \item Draw a vector of $\mathbf{V^{b}}=(V_1,...,V_i,...,V_n)'$, where $V_i$'s are iid zero mean random variables with unit variance, such as Bernoulli random variables with $Pr(V=1-\kappa)=\kappa/\sqrt{5}$ with $\kappa = (\sqrt{5}+1)/2$ as suggested by Mammen (1993).
    \item Compute the bootstrap draw $\hat{\theta}^{\star b} =\hat{\theta} + \hat{L}' \mathbf{V^{b}}  $ where $\hat{L}$ is a consistent estimator of the n-dimensional vector $L$ with $i^{th}$ element given by
    \begin{equation*}
    L(i) = \sum_{g=2}^T\sum_{t=2}^T w_{gt} \phi_{gt}(i),
 \end{equation*}
 where $\phi_{gt}(i)$ is defined in the previous algorithm.
    \item Compute $\hat{R}^{\star b} = \sqrt{n}(\hat{\theta}^{\star b}-\hat{\theta})$.
    \item Repeat steps 1-3 $B$ times.
    \item Compute the bootstrapped covariance as $\hat{\Sigma}^{1/2} = (q_{0.75}-q_{0.25})/(z_{0.75}-z_{0.25})$, where $q_p$ is the $p^{th}$ sample quantile of $\hat{R}^{\star}$ (across B draws) and $z$ is the $p^{th}$ quantile of the standard normal distribution.
    \item For each $b$, compute $\text{t-stat}_{g\leq t}^b = \max_{(g,t)} |\hat{R}^{\star b}(g,t) |\hat{\Sigma}^{-1/2}(g,t)$.
    \item Construct $\hat{c}_{1-\alpha}$ as the empirical $(1-\alpha)$ quantile of the B boostrap draws $\text{t-stat}^b$.
    \item Construct the bootstrapped confidence interval for $\theta$ as $\hat{C} = [\hat{\theta}\pm \hat{c}_{1-\alpha}\hat{\Sigma}^{-1/2}/\sqrt{n}]$.
\end{enumerate}
 
If the weights $w_{gt}$'s are random, the influence function is changed to
\begin{equation*}
    L(i) = \sum_{g=2}^T\sum_{t=2}^T w_{gt}  \phi_{gt}(i) + \gamma_{gt}^w(i) ATT(g,t),
 \end{equation*}
 where $\gamma_{gt}^w(i)$ is an error function. For example, let us consider
 \begin{equation*}
    \theta_s = \sum_{g=2}^T\sum_{t=2}^T w_{gt} ATT(g,t),
 \end{equation*}
 where $w_{gt} = P(G=g)\frac{1}{T-g+1}$ for $t\geq g$ , and $0$ otherwise. Define $\hat{w}_{gt} =  \frac{1}{T-g+1} \frac{1}{n} \sum_{i=1}^n G_i$ for $t\geq g$ and $0$ otherwise. A consistent estimator of the error function is given by
  \begin{equation*}
    \hat{\gamma}_{gt}^w(i) = \frac{1}{T-g+1}\left(G_i - \frac{1}{n}\sum_{i=1}^n G_i\right).
 \end{equation*}

\subsubsection{Not yet treated as the control group}\label{sec:notyettreated}
In the previous model, we assumed the existence of a true control group, i.e.  a group of individuals that are never treated. In many applications, this situation is not realistic. Instead, the researcher can use the individuals that are ``not yet treated'', that is treated in $g>t$ to define a control group. The extension to this setting being developed in length in \cite{callaway2021difference}, we only explain how it applies to the chained DiD.

Most importantly, the parallel trend assumption is modified to
\begin{assumption}[Conditional Parallel Trends]\label{ass: unconditional parallel trend 3}
For all $t=2,...,\mathcal{T}$, $g=2,...,\mathcal{T}$, such that $g \leq t$,
\begin{equation}
E \left[ Y_t(0) - Y_{t-1}(0) | X,G_g=1\right] = E \left[ Y_t(0) - Y_{t-1}(0) | X,D_t=0\right] \text{ } a.s..
\end{equation}
\end{assumption}

Following minor modifications to Theorem C.1. in \cite{callaway2021difference}, using the ``not yet treated'' as the control group only changes the weight $w^C_{\tau\tau-1}(g,X)$ used in our Theorem \ref{theorem: general identification} to  
\begin{equation*}
w^C_{\tau\tau-1}(g,X) = \frac{p_{g,t}(X)(1-D_t) S_{\tau,\tau-1}}{1-p_{g,t}(X)}/E_M[\frac{p_{g,t}(X)(1-D_t) S_{\tau,\tau-1}}{1-p_{g,t}(X)}].
\end{equation*}
We observe two changes. First, the binary variable $C$ becomes $1-D_t$. Second, generalized propensity score is now also a function of $t$: $p_{g,t}(X)=P(G_g=1|X,(G_g=1 \cup D_t=1))$. The propensity scores must hence be estimated for pairs $(g,t)$ because the control group evolves through time. The asymptotic properties of the two-step estimator remain similar, with minor changes to the asymptotic covariance.

\subsubsection{General missing data patterns}\label{sec:proofgeneralmissingdata}

Under Assumption \ref{ass:generalmissing sampling 2}, the data generating process consists of random draws from the  mixture distribution $F_M(y,y',g_1, ...,g_{\mathcal{T}}, c,s_1,...,s_{\mathcal{T}},x)$ defined as
\begin{equation*}
\begin{aligned}
& \sum_{t=1}^{\mathcal{T}} \lambda_{t-k,t} F_{Y_{t-k},Y_{t},G_1,...,G_{\mathcal{T}},C,X|S_{t-k,t}}(y_{t-k},y_{t},g_1,...,g_{\mathcal{T}},C,X|s_{t-k,t}=1),
\end{aligned}
\end{equation*}
where $\lambda_{t-k,t}=P(S_{t-k,t} = 1)$ is the probability of being sampled in both $t$ and $t-k$, $y$ and $y'$ denote the outcome $y_{t-k}$ and $y_{t}$, respectively, for an individual sampled at $t-k$ and $t$. Again, expectations under the mixture distribution does not correspond to population expectations.  This difference arises because of different sampling probabilities across time periods and because Assumption \ref{ass:generalmissing sampling on trends 2} does not preclude from some forms of dependence between the sampling process and the unobservable heterogeneity in $Y_{it}$.

\begin{proof}[Proof of Theorem \ref{theorem: general asymptotic results}] 

Let us define the vector of parameters $\Theta = \mathbf{ATT}$ that includes all $\theta_{\tau} = ATT(\tau)$, for all $\tau>1$ since $\theta_{1} = ATT(1) = 0$ by construction. The inverse problem in \eqref{eq:inverse problem1} corresponds to the set of moment equalities
\begin{equation}
    E_M\left[h_i( \mathcal{W}_i |\Theta)\right] = 0_{L_{\Delta}},
\end{equation}
where $h_i( \mathcal{W}_i |\Theta)$ is a $L_{\Delta}$-dimensional vector of which each element is defined by
\begin{equation}\label{eq:brick general}
    \left[
         w^G_{i\tau\tau-k}(g)\left(Y_{i\tau} - Y_{i\tau-k}\right) -  w^C_{i\tau\tau-k}(g,X)  \left(Y_{i\tau} - Y_{i\tau-k}\right) - \theta_{\tau} + \theta_{\tau-k} \right],
\end{equation}
possibly for all $\tau \geq 2$, and $1 \leq k<\tau$, with the weights 
\begin{equation*}
w^G_{\tau\tau-k}(g) = \frac{G_g S_{\tau,\tau-k}}{E_M[G_g S_{\tau,\tau-k}]}
\end{equation*}
and 
\begin{equation*}
w^C_{\tau\tau-k}(g,X) = \frac{p_g(X)C S_{\tau,\tau-k}}{1-p_g(X)}/E_M[\frac{p_g(X)C S_{\tau,\tau-k}}{1-p_g(X)}].
\end{equation*}

The previous asymptotic results in Theorem \ref{theorem: general identification} and \ref{theorem: asy;ptotic results} apply to \eqref{eq:brick general} up to minor modifications under Assumptions \ref{ass: unconditional parallel trend 2} - \ref{ass: existence of treatment groups 2}, a standard assumption on the parametric estimates of the propensity scores (Assumption 5 in \cite{callaway2021difference}  or 4.4 in \cite{abadie_2005}), and Assumptions \ref{ass:generalmissing sampling 2} and \ref{ass:generalmissing sampling on trends 2} so that we can safely assume that consistency and  asymptotic normality holds for $\widehat{\bm{\Delta ATT}}$ as $n\to\infty$  with covariance $\Omega$, which is defined later. Here, we focus on the aspects of the proofs which differ, namely the optimal combination of each ``chain link'' using GMM. The optimal GMM estimator consists in minimizing
\begin{equation}
     \hat{E}_M\left[h_i( \mathcal{W}_i |\Theta)\right]'\Omega^{-1}\hat{E}_M\left[h_i( \mathcal{W}_i |\Theta)\right],
\end{equation}
with respect to $\Theta$, using the optimal weighting matrix $\Omega^{-1}$ which corresponds the inverse of the covariance of $h_i$ \citep{hansen1982large}, hence that of $\bm{\Delta ATT}$. Let us rewrite this problem as
\begin{equation}
     \max\limits_{\bm{ ATT}} \quad -(\bm{\widehat{\Delta ATT}}-W\bm{ ATT})'\Omega^{-1}(\bm{\widehat{\Delta ATT}}-W\bm{ ATT}),
\end{equation}
then the first-order condition with respect to $\bm{ ATT}$ is given by
\begin{equation}
     -2(\bm{\widehat{\Delta ATT}})-W\bm{ ATT})'\Omega^{-1}W = 0,
\end{equation}
which, in turn, leads to the proposed estimator: 
\begin{equation}
     \widehat{\bm{ ATT}} = (W'\Omega^{-1}W)^{-1}W'\Omega^{-1}\bm{\widehat{\Delta ATT}}.
\end{equation}

The necessary and sufficient rank condition for GMM identification in this linear setting is that the rank of  $\Omega^{-1}W$ is equal to the number of columns \citep{newey1994large}. This condition is satisfied if both the covariance matrix $\Omega$ and the weight matrix $W$ are non-singular. Remark further that if $W$ is not full row rank then some $ATT(g,t)$ are not identified by the collection of $\Delta_k ATT(g,t)$'s identified in the dataset at hand. 

Proving consistency requires introducing standard assumptions for GMM estimators. We assume that (i) $\Omega$ and the weight matrix $W$ are non-singular; (ii) the true value $\Theta_0$ lies within a compact set; and (iii) $E_M[\sup_{\Theta}||h_i(\mathcal{W}_i|\Theta)||]<\infty$. In addition to our previous assumptions, applying Theorem 2.6 in \citet{newey1994large} yields the desired consistency result. Assuming further that (iv) $\Theta_0$ lies in the interior of the compact set; (v) $E[||h_i(\mathcal{W}_i|\Theta_0)||^2,]<\infty$; (vi) $W'\Omega^{-1}W$ non-singular, then asymptotic normality follows from Theorem 3.4 in \citet{newey1994large}.\footnote{All the other sufficient conditions used by these theorems are trivially satisfied in this linear setting.}

Note that the two-step GMM estimator requires estimating $\Omega$. We proceed as follows.
 For every $(g,t,k)$, compute the n-dimensional vector $\psi_{gtk}$ with $i^{th}$ element defined as
 \begin{equation}\label{eq:psi's 2}
 \psi_{gtk}(i) = \psi_{gtk}^G(i) + \psi_{gtk}^C(i),
 \end{equation}
 where
    \begin{equation*}
\begin{aligned}
\psi_{gtk}^G(i) = & w^G_{it,t-k}(g)\left[(Y_{it}-Y_{it-k}) - E_M\left[w^G_{t,t-k}(g)(Y_{t}-Y_{t-k})   \right]   \right], \\
\psi_{gtk}^C(i) = & w^C_{it,t-k}(g,X)\left[(Y_{it}-Y_{it-k}) - E_M\left[w^C_{t,t-k}(g,X)(Y_{t}-Y_{t-k})   \right]   \right] + M_{gtk}'\xi^{\pi}_g(i),
\end{aligned}
\end{equation*}
where $M_{gtk}$, a $k$ dimensional vector ($k$ being the number of covariates in $X$), is defined as
\begin{equation*}
\begin{aligned}
M_{gtk} = \frac{E\left[ X (\frac{C S_t S_{t-k}}{1-p_g(X)})^2 \dot{p}_g(X)\left[(Y_{it}-Y_{it-k}) - E\left[w^C_{t,t-k}(g,X)(Y_{t}-Y_{it-k})   \right]   \right] \right]}{E[\frac{p_g(X)C}{1-p_g(X)}]},
\end{aligned}
\end{equation*}
with $\hat{p}_g(X_i) = \Lambda (X_i'\hat{\pi}_g)$ being the parametric propensity score for covariates $X_i$ and $\dot{p}_g(X) = \partial \Lambda (X_i'\hat{\pi}_g)/\partial (X_i'\hat{\pi}_g)$. Furthermore,  $\xi^{\pi}_g(i)$ is a k-dimensional vector for each observation $i$, and is given by
\begin{equation*}
\begin{aligned}
\xi^{\pi}_g(i) = E_M\left[\frac{(G_g+C)\dot{p}_g(X)^2}{p_g(X_i)(1-p_g(X_i))}XX'\right]^{-1}X_i \frac{(G_g+C)(G_g-p_g(X_i))\dot{p}_g(X_i)}{p_g(X_i)(1-p_g(X_i))}.
\end{aligned}
\end{equation*}
Concatenate all $\psi_{gtk}$'s into a $L_{\Delta}\times n$ matrix $\Psi$, and compute $\hat{\Omega} = \hat{E}[\Psi(i)\Psi(i)']$.

Therefore, the asymptotic covariance of $\widehat{\bm{ATT
}}$ is 
\begin{equation}
    \bm{\Sigma} = (W'\Omega^{-1}W)^{-1},
\end{equation}
and its corresponding influence function to be used in the bootstrap procedure detailed in Online Appendix \ref{sec:bootstrap} is the empirical counterpart of
\begin{equation}
    \Phi = (W'\Omega^{-1}W)^{-1}W'\Omega^{-1}\Psi.
    \end{equation}
    Finally, it is easy to show that the choice of the optimal weighting $\Omega^{-1}$ is the same if the objective is instead to minimize the variance of a linear transformation $R'\bm{ATT}$, where $R$ is a vector of weights, like for all  the summary parameters considered in Online Appendix \ref{sec:summaryparam}. 
    In that case, the bootstrap for summary parameters in Online Appendix \ref{bootstrapsummary} apply with the (general) influence defined as follows. Let the weights in $R$ be random, the influence function is changed to
\begin{equation*}
    L(i) = R'\Phi(i) + \gamma^w(i)' \bm{ATT},
 \end{equation*}
 where $\gamma^w$ is the error function that depends on the randomness of the weights, as defined in Online Appendix \ref{bootstrapsummary}.

\end{proof}

\pagebreak

\begin{proof}[Proof of Corollary \ref{corrolary:attritionmodel1}]
This proof focuses on the identification of parameters in the general framework with a rotating panel structure. We show that how attrition models can be combined with our approach. Let us modify the definition from Theorem \ref{theorem: general identification}, $ATT_X(g,\tau) = E[Y_{\tau}(1) - Y_{\tau}(0)|X,G_g = 1]$ to write its first-difference as 
\begin{equation*}
\begin{aligned}
\Delta ATT_X(g,\tau) & = ATT_X(g,\tau) - ATT_X(g,\tau-1) \\
& =  E[Y_{\tau} - Y_{\tau-1}|X,G_g = 1] - E[Y_{\tau} - Y_{\tau-1}|C = 1] \\
& =  A_X(g,\tau) - B_X(g,\tau) \\
\end{aligned}
\end{equation*}
where the second equality follows from the conditional parallel trends assumption. Again, we can use the above expression to develop $ATT(g,t)$ into
\begin{equation}
\begin{aligned}
ATT(g,t) & =    \sum_{\tau=g}^{t} E\left(  A_X(g,\tau) - B_X(g,\tau)  | G_g =1 \right)  \\
\end{aligned}
\end{equation}
with 
\begin{equation}
\begin{aligned}
E\left(  A_X(g,\tau) | G_g =1 \right) & = E\left( Y_{\tau} - Y_{\tau-1} |G_g \right)   \\
& = E\left( E[Y_{\tau} - Y_{\tau-1}|X,G_g ] |G_g \right)   \\
& = E\left( E[Y_{\tau} - Y_{\tau-1}|X,G_g,S_{\tau}  S_{\tau-1} ] |G_g \right)   \\
& = E\left(  E[ \frac{S_{\tau}  S_{\tau-1}}{E[S_{\tau}  S_{\tau-1}|X,G_g]} (Y_{\tau} - Y_{\tau-1})|X,G_g ] |G_g \right)   \\
& = E\left(   \frac{S_{\tau}  S_{\tau-1}}{E[S_{\tau}  S_{\tau-1}|X,G_g]} (Y_{\tau} - Y_{\tau-1}) |G_g \right)   \\
& = E_M \big( \left( Y_{\tau} - Y_{\tau-1} \right)\frac{G_g S_{\tau}  S_{\tau-1}}{E[S_{\tau}  S_{\tau-1}|X,G_g]E[G_g]} \big) \\
& = E_M \big( \frac{E[G_g S_{\tau}  S_{\tau-1}]}{E[S_{\tau}  S_{\tau-1}|X,G_g]E[G_g]}\left( Y_{\tau} - Y_{\tau-1} \right) \frac{G_g S_{\tau}  S_{\tau-1}}{E[G_g S_{\tau}  S_{\tau-1}]} \big) \\
& = E_M \big( \frac{E[S_{\tau}  S_{\tau-1}|G_g]}{E[S_{\tau}  S_{\tau-1}|X,G_g]}\left( Y_{\tau} - Y_{\tau-1} \right) \frac{G_g S_{\tau}  S_{\tau-1}}{E[S_{\tau}  S_{\tau-1}|G_g]E[G_g]} \big)
\end{aligned}
\end{equation}
by the law of iterated expectations and the definition of $F_M$. The third equality follows from the conditional mean independence between $Y_{\tau} - Y_{\tau-1}$ and $S_{\tau}  S_{\tau-1}$ conditional on $X$ and treatment assignment. Remark that if the conditional probability $E[S_{\tau}  S_{\tau-1}|X,G_g]$ does not depend on $X$, then we are back to Theorem \ref{theorem: general identification}. Similarly, we obtain 

\begin{equation}
\begin{aligned}
& E (   B_X(g,\tau) | G_g =1)  = E ( E[Y_{\tau} - Y_{\tau-1}|X,C ,S_{\tau,\tau-1}]|G_g  )   \\
& =  E ( E[\frac{S_{\tau}  S_{\tau-1}}{E[S_{\tau}  S_{\tau-1}|X,C]}(Y_{\tau} - Y_{\tau-1})|X,C]|G_g  )   \\
& = E\left( E[ \frac{S_{\tau}  S_{\tau-1}}{E[S_{\tau}  S_{\tau-1}|X,C]}\frac{C }{1-P(G_g=1|X,G_g+C)} (Y_{\tau} - Y_{\tau-1})|X,G_g+C=1 ]|G_g \right)   \\
& = \frac{E\left( G_g E[ \frac{S_{\tau}  S_{\tau-1}}{E[S_{\tau}  S_{\tau-1}|X,C]} \frac{C}{1-P(G_g|X,G_g+C)} (Y_{\tau} - Y_{\tau-1})|X,G_g+C ]|G_g+C  \right)}{{P(G_g=1|G_g+C)}},   \\
\end{aligned}
\end{equation}
where using the definition of $p_g$ yields
\begin{equation}
\begin{aligned}
... & = \frac{E\left( G_g E[ \frac{S_{\tau}  S_{\tau-1}}{E[S_{\tau}  S_{\tau-1}|X,C]} \frac{C}{1-p_g(X)} (Y_{\tau} - Y_{\tau-1})|X,G_g+C ]|G_g+C  \right)}{{P(G_g=1|G_g+C)}},   \\
& = \frac{E\left( E[ \frac{S_{\tau}  S_{\tau-1}}{E[S_{\tau}  S_{\tau-1}|X,C]} \frac{p_g(X)C}{1-p_g(X)} (Y_{\tau} - Y_{\tau-1})|X,G_g+C ]|G_g+C  \right)}{{E(G_g|G_g+C)}},   \\
& = \frac{E \left((G_g+C) E[ \frac{S_{\tau}  S_{\tau-1}}{E[S_{\tau}  S_{\tau-1}|X,C]} \frac{p_g(X)C}{1-p_g(X)} (Y_{\tau} - Y_{\tau-1})|X,G_g+C ]  \right)}{{E(G_g|G_g+C)E(G_g+C)}},   \\
& = \frac{E \left((G_g+C) E[ \frac{S_{\tau}  S_{\tau-1}}{E[S_{\tau}  S_{\tau-1}|X,C]} \frac{p_g(X)C}{1-p_g(X)} (Y_{\tau} - Y_{\tau-1})|X,G_g+C ]  \right)}{{E(G_g)}},   \\
& = \frac{E\left( E[(G_g+C) |X] E[ \frac{S_{\tau}  S_{\tau-1}}{E[S_{\tau}  S_{\tau-1}|X,C]}\frac{p_g(X) C}{1-p_g(X)} (Y_{\tau} - Y_{\tau-1})|X,G_g+C]   \right)}{{E(G_g)}}   \\
& = \frac{E\left( E[ \frac{S_{\tau}  S_{\tau-1}}{E[S_{\tau}  S_{\tau-1}|X,C]}\frac{p_g(X) C}{1-p_g(X)} (Y_{\tau} - Y_{\tau-1})|X]  \right)}{{E(G_g )}}   \\
& = \frac{E\left(  \frac{S_{\tau}  S_{\tau-1}}{E[S_{\tau}  S_{\tau-1}|X,C]}\frac{p_g(X) C}{1-p_g(X)} (Y_{\tau} - Y_{\tau-1})  \right)}{{E(G_g )}}   \\
& = \frac{E\left(  \frac{E(G_g S_{\tau}  S_{\tau-1} )}{E[S_{\tau}  S_{\tau-1}|X,C]E(G_g )}\frac{p_g(X) C S_{\tau}  S_{\tau-1}}{1-p_g(X)} (Y_{\tau} - Y_{\tau-1})  \right)}{{E(G_g S_{\tau}  S_{\tau-1} )}}   \\
\end{aligned}
\end{equation}

Therefore, one can estimate an attrition model ex-ante for each treatment group and the control group to obtain $E[S_{\tau}  S_{\tau-1}|X,G_g]$ and $E[S_{\tau}  S_{\tau-1}|X,C]$, then reweight all first differences in outcome by $\frac{E(S_{\tau}  S_{\tau-1} |C)}{E[S_{\tau}  S_{\tau-1}|X,C]}$ and $\frac{E(S_{\tau}  S_{\tau-1} |G_g)}{E[S_{\tau}  S_{\tau-1}|X,G_g]}$ before applying our chained DiD estimator. Notice that, in general, the attrition model does not require to use the same explanatory variable $X$ than the propensity score $p_g(X)$. Remark, however, that we must be able to identify $E(G_g)$ in the population.

\end{proof}

\begin{proof}[Asymptotics for Corollary \ref{corrolary:attritionmodel1}]
We now derive the asymptotic distributions considering the added uncertainty introduced by the estimation of a first-step selection model (propensity scores) under the assumptions of Corollary \ref{corrolary:attritionmodel1}.

First, we introduce additional notations about the estimation of propensity scores in the modified weights

\begin{equation}    
\tilde w^G_{\tau\tau-1}(g,X) = \frac{E(S_{\tau}  S_{\tau-1}|G_g )}{E[S_{\tau}  S_{\tau-1}|X,G_g]} w^G_{\tau\tau-1}(g,X),
\end{equation}
\begin{equation}    
\tilde w^C_{\tau\tau-1}(g,X) = \frac{E(S_{\tau}  S_{\tau-1}|C  )}{E[S_{\tau}  S_{\tau-1}|X,C]} w^C_{\tau\tau-1}(g,X).
\end{equation}

Let us denote the added term
\begin{equation}    
l_t(q_a) = \frac{E(S_{\tau}  S_{\tau-1}|a  )}{q_a(X)},
\end{equation}    
with $q_a(X_i) = E[S_{\tau}  S_{\tau-1}|X_i,a] = \Lambda (X_i'\rho_a^0)$ a parametric propensity scores with  $\Lambda (\cdot)$ being a known function (logit or probit), that can be parametrically estimated by maximum likelihood. Remark that we could use other $X$'s here. We denote $\hat{q}_a(X_i) = \Lambda (X_i'\hat{\rho}_a)$ where $\hat{\rho}_a$ are estimated by ML, $\dot{q}_a = \partial q_a(u)/\partial u$, and  $\dot{q}_a(X) = \dot{q}_a(X_i'\rho_a^0)$. Under this assumption, the estimated parameter $\hat{\rho}_a$ is asymptotically linear.

Let us now define,
\begin{equation}\label{eq:psi's}
\psi_{gt}(\mathcal{W}_i) = \psi_{gt}^G(\mathcal{W}_i) + \psi_{gt}^G(\mathcal{W}_i),
\end{equation}
where
\begin{equation*}
\begin{aligned}
\psi_{gt}^G(\mathcal{W}_i) = & \tilde{w}^G_{it,t-1}(g,X)\left[(Y_{it}-Y_{it-1}) - E_M\left[\tilde{w}^G_{it,t-1}(g,X)(Y_{it}-Y_{it-1})   \right]   \right] + N_{gt}^{G'}\xi^{\rho}_g(\mathcal{W}_i), \\
\psi_{gt}^C(\mathcal{W}_i) = & \tilde{w}^C_{it,t-1}(g,X)\left[(Y_{it}-Y_{it-1}) - E_M\left[\tilde{w}^C_{it,t-1}(g,X)(Y_{it}-Y_{it-1})   \right]   \right]  + M_{gt}'\xi^{\pi}_g(\mathcal{W}_i) + N_{gt}^{C'}\xi^{\rho}_C(\mathcal{W}_i),
\end{aligned}
\end{equation*}
with 
\begin{equation*}
\begin{aligned}
M_{gt} = \frac{E_M\left[l_t(q_C) X (\frac{C S_t S_{t-1}}{1-p_g(X)})^2 \dot{p}_g(X)\left[(Y_{it}-Y_{it-1}) - E_M\left[\tilde{w}^C_{it,t-1}(g,X)(Y_{it}-Y_{it-1})   \right]   \right] \right]}{E_M[\frac{p_g(X)C}{1-p_g(X)}]},
\end{aligned}
\end{equation*}
\begin{equation*}
\begin{aligned}
N_{gt}^C = -\frac{E_M\left[ \frac{p_g(X)C}{1-p_g(X)}X\frac{E( S_{\tau}  S_{\tau-1}|C )}{q_C(X;\overline{\pi})^2} \dot{q}_C(X;\overline{\rho})\left[(Y_{it}-Y_{it-1}) - E_M\left[\tilde{w}^C_{it,t-1}(g,X)(Y_{it}-Y_{it-1})   \right]   \right] \right]}{E_M[\frac{p_g(X)C}{1-p_g(X)}]},
\end{aligned}
\end{equation*}
\begin{equation*}
\begin{aligned}
N_{gt}^G = -\frac{E_M\left[ G_gS_{t,t-1}X\frac{E(S_{\tau}  S_{\tau-1} |G_g )}{q_G(X;\overline{\pi})^2} \dot{q}_G(X;\overline{\rho})\left[(Y_{it}-Y_{it-1}) - E_M\left[\tilde{w}^C_{it,t-1}(g,X)(Y_{it}-Y_{it-1})   \right]   \right] \right]}{E_M[G_gS_{t,t-1}]},
\end{aligned}
\end{equation*}
are a $k$ dimensional vectors, $k$ being the number of covariates in $X$.
Finally, let $\widehat{\Delta ATT}_{g \leq t}$ and $\Delta ATT_{g \leq t}$ denote the vectors of all $\widehat{\Delta ATT}(g,t)$ and $\Delta ATT(g,t)$ for any $2 \leq g \leq t \leq \mathcal{T}$. Similarly,  the collection of $\psi_{gt}$ across all periods and groups such that $g\leq t$ is denoted by $\Psi_{g\leq t}$.

Second, we show the asymptotic result for $\Delta ATT$. Recall that
\begin{equation*}
\begin{aligned}
 \widehat{ATT}(g,t)  =  \sum_{\tau=g}^t  \widehat{\Delta ATT}(g,\tau),
\end{aligned}
\end{equation*}
where now
\begin{equation*}
\begin{aligned}
\widehat{\Delta ATT}(g,\tau) & = \hat{E}_M \left[ l_t(q_g)\frac{G_{g} S_{\tau-1}S_{\tau}}{ \hat{E}_M [ G_{g} S_{\tau-1}S_{\tau}]}(Y_{\tau}-Y_{\tau-1})\right] - \hat{E}_M \left[l_t(q_C) \frac{\frac{p_g(X)C S_{\tau-1}S_{\tau}}{1-p_g(X)}} {\hat{E}_M \left[\frac{p_g(X)C S_{\tau-1}S_{\tau}}{1-p_g(X)}\right] }(Y_{\tau}-Y_{\tau-1})\right] \\
& = \widehat{\Delta ATT}_g(g,\tau) - \widehat{\Delta ATT}_C(g,\tau),
\end{aligned}
 \end{equation*}

Let us treat each term separately. For $\widehat{\Delta ATT}_C(g,\tau)$, take an arbitrary function $g$, let 
\begin{equation*}
    w_{t}(g) = \frac{g(X)CS_{t,t-1}}{1-g(X)}
\end{equation*}
and note that
\begin{equation*}
\begin{aligned}
\sqrt{n}( \widehat{\Delta ATT}_{C}(g,t)-& \Delta ATT_{C}(g,t) )  =  \sqrt{n} \left( \hat{E}_M \left[   \frac{l_t(\hat{q}_C)w_{t}(\hat{p}_g)}{\hat{E}_M [  w_{t}(\hat{p}_g)]} (Y_{t}-Y_{t-1})\right] -  E_M \left[  \frac{l_t(q_C)w_{t}(p_g)}{E_M [  w_{t}(p_g)]} (Y_{t}-Y_{t-1})\right] \right) \\
  = &  \frac{\sqrt{n}}{\hat{E}_M [  w_{t}(\hat{p}_g)]} \left( \hat{E}_M \left[  l_t(\hat{q}_C)w_{t}(\hat{p}_g) (Y_{t}-Y_{t-1})\right] -  \frac{\hat{E}_M [ w_{t}(\hat{p}_g)]}{E_M [  w_{t}(p_g)]}  E_M \left[  l_t(q_C)w_{t}(p_g) (Y_{t}-Y_{t-1})\right] \right) \\
  = &  \frac{\sqrt{n}}{\hat{E}_M [  w_{t}(\hat{p}_g)]} \left( \hat{E}_M \left[  l_t(\hat{q}_C)w_{t}(\hat{p}_g) (Y_{t}-Y_{t-1})\right] -  E_M \left[  l_t(q_C)w_{t}(p_g) (Y_{t}-Y_{t-1})\right] \right) \\
 & -    \frac{ E_M \left[  w_{t}(p_g) (Y_{t}-Y_{t-1})\right]}{\hat{E}_M [  w_{t}(\hat{p}_g)]E_M [  w_{t}(p_g)]} \sqrt{n} (\hat{E}_M \left[ l_t(q_C)w_{t}(\hat{p}_g) \right]  - \hat{E}_M \left[ l_t(q_C)w_{t}(p_g) \right]) \\
 = &  \frac{1 }{\hat{E}_M [  w_{t}(\hat{p}_g)]}\sqrt{n} A_n(\hat{p}_g,\hat{q}_C) -  \frac{\Delta ATT_C(g,t)}{\hat{E}_M [  w_{t}(\hat{p}_g)]}\sqrt{n} B_n(\hat{p}_g,\hat{q}_C) \\
  = &  \frac{1 }{E_M [  w_{t}(p_g)]}\sqrt{n} A_n(\hat{p}_g,\hat{q}_C) -  \frac{\Delta ATT_C(g,t)}{E_M [  w_{t}(p_g)]}\sqrt{n} B_n(\hat{p}_g,\hat{q}_C) +o_p(1), \\
\end{aligned}
\end{equation*}

Applying the mean value theorem and the Classical Glivenko-Cantelli's theorem yield
\begin{equation*}
\begin{aligned}
A_n(\hat{p}_g,\hat{q}_C) = & \hat{E}_M \left[ l_{t}(q_C) w_{t}(p_g) (Y_{t}-Y_{t-1})\right] - E_M \left[  l_{t}(q_C)w_{t}(p_g) (Y_{t}-Y_{t-1})\right] \\
& + \hat{E}_M \left[l_{t}(q_C)X\frac{C S_{t,t-1}}{(1-p_g(X;\overline{\pi}))^2} \dot{p}_g(X;\overline{\pi}) (Y_{t}-Y_{t-1}) \right]'\left(\hat{\pi}_g-\pi^0_g \right) \\
& - \hat{E}_M \left[w_{t}(p_g)X\frac{E(S_{\tau}  S_{\tau-1}|C )}{q_C(X;\overline{\pi})^2} \dot{q}_C(X;\overline{\rho}) (Y_{t}-Y_{t-1}) \right]'\left(\hat{\rho}_C-\rho^0_C \right) + o_p(n^{-1/2}),
\end{aligned}
\end{equation*}
\begin{equation*}
\begin{aligned}
B_n(\hat{p}_g) = & \hat{E}_M \left[  l_{t}(q_C)w_{t}(p_g) \right] - E_M \left[  l_{t}(q_C)w_{t}(p_g) \right] \\
& + \hat{E}_M \left[l_{t}(q_C)X\frac{C S_{t,t-1}}{(1-p_g(X;\overline{\pi}))^2} \dot{p}_g(X;\overline{\pi})  \right]'\left(\hat{\pi}_g-\pi^0_g \right) \\
& - \hat{E}_M \left[w_{t}(p_g)X\frac{E(S_{\tau}  S_{\tau-1} |C)}{q_C(X;\overline{\pi})^2} \dot{q}_C(X;\overline{\rho})  \right]'\left(\hat{\rho}_C-\rho^0_C \right) + o_p(n^{-1/2}),
\end{aligned}
\end{equation*}
where $\overline{\pi}$ and $\overline{\rho}$ are intermediate points.
Applying the same reasoning for $\Delta ATT_g(g,t)$ gives
\begin{equation*}
\begin{aligned}
\sqrt{n}( \widehat{\Delta ATT}_{g}(g,t)- \Delta ATT_{g}(g,t) )  
 =   \frac{1 }{E_M [  w_{t}]}\sqrt{n} A_n(\hat{q}_C) -  \frac{\Delta ATT_g(g,t)}{E_M [  w_{t}]}\sqrt{n} B_n(\hat{p}_g,\hat{q}_C) +o_p(1), \\
\end{aligned}
\end{equation*}
with 
\begin{equation*}
    w_{t} = G_gS_{t,t-1},
\end{equation*}
and
\begin{equation*}
\begin{aligned}
A_n(\hat{q}_C) = & \hat{E}_M \left[ l_{t}(q_g) w_{t} (Y_{t}-Y_{t-1})\right] - E_M \left[  l_{t}(q_g)w_{t} (Y_{t}-Y_{t-1})\right] \\
& - \hat{E}_M \left[w_{t}X\frac{E(G_g S_{\tau}  S_{\tau-1} )}{q_g(X;\overline{\pi})^2E(G_g )} \dot{q}_g(X;\overline{\rho}) (Y_{t}-Y_{t-1}) \right]'\left(\hat{\rho}_g-\rho^0_g \right) + o_p(n^{-1/2}),
\end{aligned}
\end{equation*}
\begin{equation*}
\begin{aligned}
B_n(\hat{p}_g) = & \hat{E}_M \left[  l_{t}(q_g)w_{t} \right] - E_M \left[  l_{t}(q_g)w_{t} \right] \\
& - \hat{E}_M \left[w_{t}X\frac{E(G_g S_{\tau}  S_{\tau-1} )}{q_g(X;\overline{\pi})^2E(G_g )} \dot{q}_g(X;\overline{\rho})  \right]'\left(\hat{\rho}_g-\rho^0_g \right) + o_p(n^{-1/2}).
\end{aligned}
\end{equation*}

Combining the above results and making use of the same Lemma yields \eqref{eq:att g-c} hence concludes the proof for $\Delta ATT$. The asymptotic covariance is given by $\Sigma_{\Delta} = E \left[\Psi_{g\leq\tau}(\mathcal{W}_i) \Psi_{g\leq\tau}(\mathcal{W}_i)' \right]$. The asymptotic result for $ATT$ follows from the same reasoning as in our main theorem.
\end{proof}

\begin{proof}[Proof of Corollary \ref{corrolary:attritionmodel2}]
This proof shows how to adapt our approach to a sampling assumption related to sequential missing at random \citep{hoonhout2019nonignorable}. We focus on a simple case to illustrate the method. Let us assume that an individual can be sampled at most two periods in a row. The first time a unit is sampled is a function of $X$ and $G_g$ or $C$, but the probability that the unit is sampled again in the next period also depends on the realization of its outcome. Formally, we assume $Y_{t} \perp S_t | Y_{t-1},X,G_g,S_t=1$ but $Y_{t} \perp S_t | X,G_g,S_{t-1}=0$.   Let us only focus on $A_X(g,\tau)$ to show how the weights change. We have
\begin{equation}
\begin{aligned}
E\left(  A_X(g,\tau) | G_g =1 \right) & = E\left( Y_{\tau} - Y_{\tau-1} |G_g \right)   \\
& = E\left( E[Y_{\tau} - Y_{\tau-1}|X,G_g,Y_{\tau-1} ] |G_g \right)   \\
& = E\left( E[Y_{\tau} - Y_{\tau-1}|X,G_g,Y_{\tau-1},S_{\tau}=1,  S_{\tau-1}=1 ] |G_g \right)   \\
& = E\left(  E[ \frac{S_{\tau} }{E[S_{\tau}  |X,Y_{\tau-1},S_{\tau-1}=1,G_g]} (Y_{\tau} - Y_{\tau-1})|X,G_g,Y_{\tau-1},S_{\tau-1} ] |G_g \right)   \\
& = E\left(  E[ \frac{S_{\tau}S_{\tau-1} }{E[S_{\tau}  |X,Y_{\tau-1},S_{\tau-1}=1,G_g]E[S_{\tau-1}  |X,G_g]} (Y_{\tau} - Y_{\tau-1})|X,G_g ] |G_g \right)   \\
& = E_M \big( \left( Y_{\tau} - Y_{\tau-1} \right)\frac{G_g S_{\tau}  S_{\tau-1}}{E[S_{\tau}  |X,Y_{\tau-1},S_{\tau-1}=1,G_g]E[S_{\tau-1}  |X,G_g]E[G_g]} \big) \\
& = E_M \big( \frac{E[G_g S_{\tau}  S_{\tau-1}]}{E[S_{\tau}  |X,Y_{\tau-1},S_{\tau-1}=1,G_g]E[S_{\tau-1}  |X,G_g]E[G_g]}\left( Y_{\tau} - Y_{\tau-1} \right) \frac{G_g S_{\tau}  S_{\tau-1}}{E[G_g S_{\tau}  S_{\tau-1}]} \big),
\end{aligned}
\end{equation}
where the third equality follows from $Y_{t} \perp S_t | Y_{t-1},X,G_g$. Remark that if the conditional probability $E[S_{\tau}|X,Y_{\tau-1},S_{\tau-1},G_g]$ does not depend on $Y_{\tau-1},S_{\tau-1}$, then we are back to the previous case. 

Therefore, our approach accommodates general attrition models with only minor modifications to the chained DiD estimator. We do not derive the asymptotic distribution for this case, but it follows from the same steps as in Corollary \ref{corrolary:attritionmodel1}. 
\end{proof}

\pagebreak

\section{Appendix for the Application}\label{appe:application}

In this appendix, we present the detailed results in various tables. The next two tables give detailed results corresponding to Figures \ref{fig:figure1}, \ref{fig:figure2}, \ref{fig:figure3}, and \ref{fig:figure4}.

\begin{table}[H]
\centering
\caption{Effects on total workforce (exhaustively observed outcome) }
\begin{tabular}{l*{5}{c}}
\hline
           & \multicolumn{ 5}{c}{{\bf log(total workforce)}} \\
           \cmidrule{2-6}
           &    Long DiD & \multicolumn{ 2}{c}{Chained DiD} & \multicolumn{ 2}{c}{Cross Section DiD} \\
                      \cmidrule(lr){2-2} \cmidrule(lr){3-4} \cmidrule(lr){5-6}
           & Exhaustive & Exhaustive & Unbalanced & Exhaustive & Unbalanced \\
           & (1) & (2) & (3) & (4) & (5) \\
\hline
$\beta_{-3}$ &     -0.013 &     -0.013 &      0.032 &     -0.013 &      0.075 \\
           & [-0.084,0.058] & [-0.079,0.053] & [-0.059,0.123] & [-0.144,0.117] & [-0.609,0.76] \\
$\beta_{-2}$ &     -0.023 &     -0.023 &     -0.008 &     -0.023 &      0.012 \\
           & [-0.091,0.044] & [-0.087,0.041] & [-0.097,0.081] & [-0.13,0.083] & [-0.549,0.574] \\
$\beta_{-1}$ &     -0.005 &     -0.005 &     -0.005 &     -0.005 &      0.024 \\
           & [-0.054,0.044] & [-0.052,0.041] & [-0.066,0.056] & [-0.09,0.08] & [-0.442,0.491] \\
    $ref.$ &          0 &          0 &          0 &          0 &          0 \\
$\beta_{1}$ &      0.027 &      0.027 &      0.046 &      0.028 &      0.062 \\
           & [-0.026,0.08] & [-0.023,0.078] & [-0.021,0.112] & [-0.103,0.159] & [-0.372,0.495] \\
$\beta_{2}$ &    0.057** &    0.056** &    0.081** &      0.057 &      0.092 \\
           & [-0.006,0.121] & [-0.005,0.118] & [0.002,0.159] & [-0.048,0.163] & [-0.426,0.61] \\
$\beta_{3}$ &     0.07** &    0.068** &    0.089** &      0.069 &      0.051 \\
           &   [0,0.14] &  [0,0.135] & [0.006,0.172] & [-0.075,0.214] & [-0.506,0.607] \\
$\beta_{4}$ &    0.084** &    0.077** &   0.119*** &      0.083 &      0.087 \\
           & [0.002,0.167] &  [0,0.154] & [0.024,0.214] & [-0.091,0.257] & [-0.53,0.704] \\
$\beta_{5}$ &   0.123*** &   0.111*** &   0.149*** &     0.119* &      0.076 \\
           & [0.032,0.213] & [0.025,0.197] & [0.037,0.261] & [-0.032,0.269] & [-0.533,0.686] \\
\hline
      Pre- &     -0.014 &     -0.014 &      0.006 &     -0.014 &      0.037 \\
     trend & [-0.048,0.021] & [-0.049,0.021] & [-0.04,0.053] & [-0.079,0.051] & [-0.294,0.368] \\
\hline
\end{tabular}  
\label{tab:table5}
\begin{tablenotes} \item Notes: \footnotesize 
This table shows the dynamic treatment effects relative to the beginning of the treatment obtained on a panel balanced on exhaustive variables. ``Exhaustive" refers to the use of an exhaustively observed outcome without pretending that this variable is imperfectly observed. ``Unbalanced" refers to the use of an exhaustively observed outcome pretending that this variable is observed from R\&D survey, that is, from an unbalanced repeated panel. 
95\% confidence intervals are obtained from the multiplier bootstrap. 
\sym{*} \(p<0.10\), \sym{**} \(p<0.05\), \sym{***} \(p<0.01\)
\end{tablenotes}
\end{table}

\begin{table}[H]
\centering
\caption{\label{tab:table6}Effects on highly qualified workers (exhaustively observed outcome)}
{
\begin{tabular}{l*{5}{c}}
\hline
           & \multicolumn{ 5}{c}{{\bf log(highly qualified workforce) }} \\
           \cmidrule{2-6}
           &    Long DiD & \multicolumn{ 2}{c}{Chained DiD} & \multicolumn{ 2}{c}{Cross Section DiD} \\
           \cmidrule(lr){2-2} \cmidrule(lr){3-4} \cmidrule(lr){5-6}
           & Exhaustive & Exhaustive & Unbalanced & Exhaustive & Unbalanced \\
           & (1) & (2) & (3) & (4) & (5) \\
\hline
$\beta_{-3}$ &      0.033 &      0.033 &      0.088 &      0.033 &      0.194 \\
           & [-0.059,0.126] & [-0.053,0.12] & [-0.059,0.235] & [-0.081,0.148] & [-0.307,0.696] \\
$\beta_{-2}$ &      0.018 &      0.018 &      0.013 &      0.018 &      0.006 \\
           & [-0.066,0.103] & [-0.06,0.096] & [-0.131,0.157] & [-0.08,0.116] & [-0.411,0.423] \\
$\beta_{-1}$ &      0.018 &      0.018 &      0.015 &      0.018 &      0.019 \\
           & [-0.058,0.094] & [-0.054,0.09] & [-0.099,0.129] & [-0.065,0.101] & [-0.328,0.365] \\
    $ref.$ &          0 &          0 &          0 &          0 &          0 \\
$\beta_{1}$ &      0.022 &      0.022 &       0.02 &      0.024 &      0.034 \\
           & [-0.056,0.1] & [-0.048,0.091] & [-0.09,0.129] & [-0.083,0.13] & [-0.287,0.355] \\
$\beta_{2}$ &      0.044 &      0.043 &      0.067 &      0.045 &      0.065 \\
           & [-0.038,0.126] & [-0.032,0.118] & [-0.051,0.184] & [-0.05,0.141] & [-0.319,0.449] \\
$\beta_{3}$ &     0.075* &    0.072** &      0.061 &      0.072 &       0.02 \\
           & [-0.015,0.164] & [-0.008,0.152] & [-0.066,0.189] & [-0.05,0.194] & [-0.388,0.428] \\
$\beta_{4}$ &    0.109** &    0.104** &      0.108 &      0.105 &      0.008 \\
           & [0.007,0.211] & [0.011,0.196] & [-0.04,0.256] & [-0.04,0.25] & [-0.439,0.456] \\
$\beta_{5}$ &    0.125** &   0.117*** &     0.128* &     0.121* &      0.069 \\
           & [0.014,0.236] & [0.016,0.217] & [-0.032,0.288] & [-0.014,0.257] & [-0.39,0.528] \\
\hline
      Pre- &      0.023 &      0.023 &      0.039 &      0.023 &      0.073 \\
     trend & [-0.022,0.068] & [-0.018,0.064] & [-0.035,0.113] & [-0.033,0.08] & [-0.173,0.319] \\
\hline
\end{tabular}

}
\begin{tablenotes} \item Notes: \footnotesize 
This table shows the dynamic treatment effects relative to the beginning of the treatment obtained on a panel balanced on exhaustive variables. ``Exhaustive'' refers to the use of an exhaustively observed outcome without pretending that this variable is imperfectly observed. ``Unbalanced'' refers to the use of an exhaustively observed outcome pretending that this variable is observed from R\&D survey, that is, from an unbalanced repeated panel. 
95\% confidence intervals are obtained from the multiplier bootstrap.
\sym{*} \(p<0.10\), \sym{**} \(p<0.05\), \sym{***} \(p<0.01\)
\end{tablenotes}
\end{table}

\begin{table}[H]
\centering
\caption{\label{tab:table7}Effects on employment variables observed from R\&D survey }
{
\begin{tabular}{l*{4}{c}}
\hline
\multicolumn{ 5}{c}{{\bf Unbalanced variables from R\&D survey in log}} \\
            \cmidrule{2-5} 
           & \multicolumn{ 2}{c}{Chained DiD} & \multicolumn{ 2}{c}{Cross Section DiD} \\
           \cmidrule(lr){2-3} \cmidrule(lr){4-5}
           & total workforce & researchers & total workforce & researchers \\
           & (1) & (2) & (3) & (4) \\
\hline
$\beta_{-3}$ &      0.053 &     -0.041 &      0.104 &      0.106 \\
           & [-0.066,0.172] & [-0.194,0.112] & [-0.598,0.806] & [-0.249,0.462] \\
$\beta_{-2}$ &      0.012 &      0.017 &      0.064 &      0.056 \\
           & [-0.111,0.135] & [-0.103,0.137] & [-0.554,0.682] & [-0.246,0.357] \\
$\beta_{-1}$ &      0.009 &      0.007 &       0.02 &      0.033 \\
           & [-0.074,0.092] & [-0.104,0.119] & [-0.481,0.522] & [-0.219,0.285] \\
    $ref.$ &          0 &          0 &          0 &          0 \\
$\beta_{1}$ &      0.053 &       0.06 &      0.075 &      0.082 \\
           & [-0.032,0.139] & [-0.051,0.17] & [-0.409,0.56] & [-0.171,0.336] \\
$\beta_{2}$ &      0.038 &    0.124** &      0.052 &      0.132 \\
           & [-0.073,0.149] & [0.001,0.248] & [-0.499,0.603] & [-0.164,0.427] \\
$\beta_{3}$ &      0.047 &   0.181*** &      0.057 &      0.139 \\
           & [-0.068,0.162] & [0.041,0.32] & [-0.525,0.638] & [-0.178,0.455] \\
$\beta_{4}$ &    0.129** &   0.245*** &      0.081 &      0.148 \\
           & [0.007,0.25] & [0.093,0.397] & [-0.57,0.733] & [-0.189,0.485] \\
$\beta_{5}$ &   0.168*** &   0.288*** &      0.116 &      0.221 \\
           & [0.034,0.302] & [0.128,0.448] & [-0.546,0.778] & [-0.144,0.585] \\
\hline
      Pre- &      0.025 &     -0.006 &      0.063 &      0.065 \\
     trend & [-0.032,0.082] & [-0.075,0.064] & [-0.319,0.444] & [-0.116,0.246] \\
\hline

\end{tabular}

}
\begin{tablenotes} \item Notes: \footnotesize 
This table shows the dynamic treatment effects relative to the beginning of the treatment. The dynamic effects are estimated with outcome variables observed from R\&D survey, that is, from an unbalanced repeated panel. 
95\% confidence intervals are obtained from the multiplier bootstrap.
\sym{*} \(p<0.10\), \sym{**} \(p<0.05\), \sym{***} \(p<0.01\)
\end{tablenotes}
\end{table}

These estimates can be compared to the next results where administrative data are not reweighted using the exhaustively observed variables. Results are presented in Tables \ref{tab:table8}, \ref{tab:table9}, and \ref{tab:table10}. This introduces some differences because there are more individuals in some periods compared to the results presented earlier.

\begin{table}[H]
\centering
\caption{Effects on total workforce (exhaustively observed outcome with the complete panel data) }
\begin{tabular}{l*{5}{c}}
\hline
           & \multicolumn{ 5}{c}{{\bf log(total workforce)}} \\
           \cmidrule{2-6}
           &    Long DiD & \multicolumn{ 2}{c}{Chained DiD} & \multicolumn{ 2}{c}{Cross Section DiD} \\
                      \cmidrule(lr){2-2} \cmidrule(lr){3-4} \cmidrule(lr){5-6}
           & Exhaustive & Exhaustive & Unbalanced & Exhaustive & Unbalanced \\
           & (1) & (2) & (3) & (4) & (5) \\
\hline
$\beta_{-3}$ &     -0.001 &      0.001 &      0.019 &          0 &      0.002 \\
           & [-0.056,0.054] & [-0.07,0.071] & [-0.067,0.105] & [-0.174,0.174] & [-0.46,0.464] \\
$\beta_{-2}$ &     -0.027 &     -0.027 &     -0.018 &     -0.022 &      0.001 \\
           & [-0.079,0.026] & [-0.092,0.038] & [-0.099,0.063] & [-0.173,0.128] & [-0.405,0.408] \\
$\beta_{-1}$ &     -0.017 &     -0.017 &     -0.013 &     -0.016 &     -0.022 \\
           & [-0.061,0.026] & [-0.068,0.033] & [-0.076,0.051] & [-0.14,0.107] & [-0.374,0.33] \\
    $ref.$ &          0 &          0 &          0 &          0 &          0 \\
$\beta_{1}$ &    0.045** &    0.045** &     0.051* &      0.049 &      0.052 \\
           & [0.002,0.087] & [-0.006,0.095] & [-0.008,0.11] & [-0.073,0.17] & [-0.282,0.386] \\
$\beta_{2}$ &   0.066*** &   0.077*** &   0.088*** &      0.057 &      0.046 \\
           & [0.013,0.118] & [0.013,0.142] & [0.015,0.161] & [-0.086,0.199] & [-0.348,0.44] \\
$\beta_{3}$ &   0.069*** &   0.085*** &    0.089** &      0.055 &     -0.001 \\
           & [0.009,0.128] & [0.01,0.159] & [0.004,0.174] & [-0.085,0.194] & [-0.408,0.406] \\
$\beta_{4}$ &   0.074*** &    0.085** &    0.11*** &       0.05 &     -0.027 \\
           & [0.009,0.139] &   [0,0.17] & [0.016,0.204] & [-0.117,0.216] & [-0.473,0.42] \\
$\beta_{5}$ &   0.094*** &    0.097** &   0.134*** &      0.052 &      -0.07 \\
           & [0.019,0.169] & [0.003,0.19] & [0.028,0.24] & [-0.121,0.226] & [-0.545,0.406] \\
\hline
      Pre- &     -0.015 &     -0.015 &     -0.004 &     -0.013 &     -0.006 \\
     trend & [-0.043,0.013] & [-0.05,0.021] & [-0.047,0.039] & [-0.087,0.061] & [-0.238,0.225] \\
\hline
\end{tabular}  
 
\label{tab:table8}
\begin{tablenotes} \item Notes: \footnotesize 
This table shows the dynamic treatment effects relative to the beginning of the treatment obtained on a panel balanced on exhaustive variables. ``Exhaustive" refers to the use of an exhaustively observed outcome without pretending that this variable is imperfectly observed. ``Unbalanced" refers to the use of an exhaustively observed outcome pretending that this variable is observed from R\&D survey, that is, from an unbalanced repeated panel. 
95\% confidence intervals are obtained from the multiplier bootstrap.
\sym{*} \(p<0.10\), \sym{**} \(p<0.05\), \sym{***} \(p<0.01\)
\end{tablenotes}
\end{table}

\begin{table}[H]
\centering
\caption{\label{tab:table9}Effects on highly qualified workers (exhaustively observed outcome with the complete panel data)}
{
\begin{tabular}{l*{5}{c}}
\hline
           & \multicolumn{ 5}{c}{{\bf log(highly qualified workforce) }} \\
           \cmidrule{2-6}
           &    Long DiD & \multicolumn{ 2}{c}{Chained DiD} & \multicolumn{ 2}{c}{Cross Section DiD} \\
           \cmidrule(lr){2-2} \cmidrule(lr){3-4} \cmidrule(lr){5-6}
           & Exhaustive & Exhaustive & Unbalanced & Exhaustive & Unbalanced \\
           & (1) & (2) & (3) & (4) & (5) \\
\hline
$\beta_{-3}$ &      0.019 &      0.025 &       0.07 &       0.04 &      0.102 \\
           & [-0.046,0.085] & [-0.053,0.102] & [-0.059,0.199] & [-0.091,0.171] & [-0.235,0.438] \\
$\beta_{-2}$ &      0.024 &      0.013 &      0.021 &      0.023 &     -0.026 \\
           & [-0.041,0.089] & [-0.062,0.087] & [-0.105,0.148] & [-0.09,0.137] & [-0.327,0.274] \\
$\beta_{-1}$ &      0.016 &      0.016 &      0.012 &      0.023 &     -0.017 \\
           & [-0.043,0.075] & [-0.053,0.085] & [-0.093,0.116] & [-0.076,0.122] & [-0.281,0.247] \\
    $ref.$ &          0 &          0 &          0 &          0 &          0 \\
$\beta_{1}$ &      0.028 &      0.028 &       0.02 &      0.031 &      0.019 \\
           & [-0.031,0.088] & [-0.04,0.095] & [-0.085,0.125] & [-0.069,0.132] & [-0.228,0.267] \\
$\beta_{2}$ &      0.045 &      0.047 &      0.077 &      0.044 &      0.044 \\
           & [-0.021,0.111] & [-0.028,0.123] & [-0.049,0.202] & [-0.068,0.155] & [-0.239,0.328] \\
$\beta_{3}$ &    0.078** &    0.077** &      0.064 &      0.072 &     -0.007 \\
           & [0.008,0.147] & [-0.006,0.16] & [-0.068,0.197] & [-0.041,0.185] & [-0.3,0.286] \\
$\beta_{4}$ &   0.116*** &   0.108*** &      0.101 &     0.103* &     -0.052 \\
           & [0.042,0.19] & [0.017,0.199] & [-0.038,0.24] & [-0.03,0.235] & [-0.385,0.281] \\
$\beta_{5}$ &   0.134*** &   0.119*** &     0.13** &     0.111* &     -0.031 \\
           & [0.048,0.22] & [0.021,0.218] & [-0.023,0.283] & [-0.029,0.251] & [-0.382,0.32] \\
\hline
      Pre- &       0.02 &      0.018 &      0.034 &      0.029 &      0.019 \\
     trend & [-0.015,0.054] & [-0.022,0.058] & [-0.031,0.1] & [-0.032,0.089] & [-0.151,0.19] \\
\hline
\end{tabular}

}
\begin{tablenotes} \item Notes: \footnotesize 
This table shows the dynamic treatment effects relative to the beginning of the treatment obtained on a panel balanced on exhaustive variables. ``Exhaustive" refers to the use of an exhaustively observed outcome without pretending that this variable is imperfectly observed. ``Unbalanced" refers to the use of an exhaustively observed outcome pretending that this variable is observed from R\&D survey, that is, from an unbalanced repeated panel. 
95\% confidence intervals are obtained from the multiplier bootstrap.
\sym{*} \(p<0.10\), \sym{**} \(p<0.05\), \sym{***} \(p<0.01\)
\end{tablenotes}
\end{table}

\begin{table}[H]
\centering
\caption{\label{tab:table10}Effects on employment variables observed from R\&D survey  with the complete panel data}
{
\begin{tabular}{l*{4}{c}}
\hline
\multicolumn{ 5}{c}{{\bf Unbalanced variables from R\&D survey in log}} \\
            \cmidrule{2-5} 
           & \multicolumn{ 2}{c}{Chained DiD} & \multicolumn{ 2}{c}{Cross Section DiD} \\
           \cmidrule(lr){2-3} \cmidrule(lr){4-5}
           & total workforce & researchers & total workforce & researchers \\
           & (1) & (2) & (3) & (4) \\
\hline
$\beta_{-3}$ &      0.036 &      -0.03 &      0.017 &      0.081 \\
           & [-0.113,0.185] & [-0.149,0.089] & [-0.471,0.504] & [-0.148,0.31] \\
$\beta_{-2}$ &      0.028 &     -0.018 &      0.079 &      0.044 \\
           & [-0.114,0.171] & [-0.125,0.088] & [-0.368,0.526] & [-0.164,0.252] \\
$\beta_{-1}$ &      0.022 &      0.007 &     -0.016 &      0.037 \\
           & [-0.1,0.144] & [-0.082,0.097] & [-0.4,0.369] & [-0.149,0.223] \\
    $ref.$ &          0 &          0 &          0 &          0 \\
$\beta_{1}$ &       0.05 &      0.073 &      0.045 &      0.095 \\
           & [-0.066,0.165] & [-0.023,0.168] & [-0.32,0.409] & [-0.095,0.285] \\
$\beta_{2}$ &      0.047 &   0.152*** &      0.022 &      0.141 \\
           & [-0.084,0.177] & [0.048,0.257] & [-0.378,0.423] & [-0.059,0.341] \\
$\beta_{3}$ &      0.064 &   0.214*** &      0.014 &      0.134 \\
           & [-0.073,0.201] & [0.099,0.329] & [-0.402,0.43] & [-0.074,0.341] \\
$\beta_{4}$ &      0.12* &   0.244*** &     -0.035 &      0.086 \\
           & [-0.029,0.269] & [0.121,0.368] & [-0.501,0.43] & [-0.137,0.309] \\
$\beta_{5}$ &    0.142** &   0.297*** &     -0.046 &      0.138 \\
           & [-0.013,0.297] & [0.163,0.432] & [-0.508,0.415] & [-0.099,0.374] \\
\hline
      Pre- &      0.029 &     -0.014 &      0.027 &      0.054 \\
     trend & [-0.042,0.099] & [-0.073,0.046] & [-0.203,0.257] & [-0.065,0.173] \\
\hline

\end{tabular}

}
\begin{tablenotes} \item Notes: \footnotesize 
This table shows the dynamic treatment effects relative to the beginning of the treatment. The dynamic effects are estimated with outcome variables observed from R\&D survey, that is, from an unbalanced repeated panel. 
95\% confidence intervals are obtained from the multiplier bootstrap.
\sym{*} \(p<0.10\), \sym{**} \(p<0.05\), \sym{***} \(p<0.01\)
\end{tablenotes}
\end{table}

\end{document}